\documentclass[letterpaper,11pt]{article}

\usepackage{amsthm,amsmath,amssymb,amsfonts}
\usepackage[american]{babel}
\usepackage{graphicx}
\usepackage[giveninits=true,maxbibnames=99,style=alphabetic,maxalphanames=4,minalphanames=3]{biblatex}
\usepackage{hyperref}
\usepackage{xurl}
\usepackage[capitalise]{cleveref}
\hypersetup{
  breaklinks=true,
  pdfusetitle=true,
  hypertexnames=false,
  colorlinks   = true, 
  urlcolor     = blue800, 
  linkcolor    = green700, 
  citecolor    = blue900 
}
\usepackage{mathtools}
\usepackage{enumitem}
\usepackage{csquotes}
\usepackage{algpseudocodex}
\usepackage{algorithm}
\usepackage{xparse}
\usepackage{xspace}
\usepackage{color}
\usepackage{caption}
\usepackage{subcaption}
\usepackage{fullpage}
\usepackage{placeins}
\usepackage{makecell}
\usepackage{thmtools}
\usepackage{thm-restate}
\usepackage{xargs}
\usepackage[letterpaper, margin=1in]{geometry}
\usepackage[colorinlistoftodos,prependcaption,textsize=small,textwidth=\marginparwidth]{todonotes}
\newcommandx{\td}[2][1=]{\todo[linecolor=OliveGreen,backgroundcolor=OliveGreen!25,bordercolor=OliveGreen,#1]{#2}}
\newcommandx{\question}[2][1=]{\todo[linecolor=red,backgroundcolor=red!25,bordercolor=red,#1]{#2}}

\definecolorset{HTML}{}{}{%
slate50,f8fafc;%
slate100,f1f5f9;%
slate200,e2e8f0;%
slate300,cbd5e1;%
slate400,94a3b8;%
slate500,64748b;%
slate600,475569;%
slate700,334155;%
slate800,1e293b;%
slate900,0f172a;%
slate950,020617;%
gray50,f9fafb;%
gray100,f3f4f6;%
gray200,e5e7eb;%
gray300,d1d5db;%
gray400,9ca3af;%
gray500,6b7280;%
gray600,4b5563;%
gray700,374151;%
gray800,1f2937;%
gray900,111827;%
gray950,030712;%
zinc50,fafafa;%
zinc100,f4f4f5;%
zinc200,e4e4e7;%
zinc300,d4d4d8;%
zinc400,a1a1aa;%
zinc500,71717a;%
zinc600,52525b;%
zinc700,3f3f46;%
zinc800,27272a;%
zinc900,18181b;%
zinc950,09090b;%
neutral50,fafafa;%
neutral100,f5f5f5;%
neutral200,e5e5e5;%
neutral300,d4d4d4;%
neutral400,a3a3a3;%
neutral500,737373;%
neutral600,525252;%
neutral700,404040;%
neutral800,262626;%
neutral900,171717;%
neutral950,0a0a0a;%
stone50,fafaf9;%
stone100,f5f5f4;%
stone200,e7e5e4;%
stone300,d6d3d1;%
stone400,a8a29e;%
stone500,78716c;%
stone600,57534e;%
stone700,44403c;%
stone800,292524;%
stone900,1c1917;%
stone950,0c0a09;%
red50,fef2f2;%
red100,fee2e2;%
red200,fecaca;%
red300,fca5a5;%
red400,f87171;%
red500,ef4444;%
red600,dc2626;%
red700,b91c1c;%
red800,991b1b;%
red900,7f1d1d;%
red950,450a0a;%
orange50,fff7ed;%
orange100,ffedd5;%
orange200,fed7aa;%
orange300,fdba74;%
orange400,fb923c;%
orange500,f97316;%
orange600,ea580c;%
orange700,c2410c;%
orange800,9a3412;%
orange900,7c2d12;%
orange950,431407;%
amber50,fffbeb;%
amber100,fef3c7;%
amber200,fde68a;%
amber300,fcd34d;%
amber400,fbbf24;%
amber500,f59e0b;%
amber600,d97706;%
amber700,b45309;%
amber800,92400e;%
amber900,78350f;%
amber950,451a03;%
yellow50,fefce8;%
yellow100,fef9c3;%
yellow200,fef08a;%
yellow300,fde047;%
yellow400,facc15;%
yellow500,eab308;%
yellow600,ca8a04;%
yellow700,a16207;%
yellow800,854d0e;%
yellow900,713f12;%
yellow950,422006;%
lime50,f7fee7;%
lime100,ecfccb;%
lime200,d9f99d;%
lime300,bef264;%
lime400,a3e635;%
lime500,84cc16;%
lime600,65a30d;%
lime700,4d7c0f;%
lime800,3f6212;%
lime900,365314;%
lime950,1a2e05;%
green50,f0fdf4;%
green100,dcfce7;%
green200,bbf7d0;%
green300,86efac;%
green400,4ade80;%
green500,22c55e;%
green600,16a34a;%
green700,15803d;%
green800,166534;%
green900,14532d;%
green950,052e16;%
emerald50,ecfdf5;%
emerald100,d1fae5;%
emerald200,a7f3d0;%
emerald300,6ee7b7;%
emerald400,34d399;%
emerald500,10b981;%
emerald600,059669;%
emerald700,047857;%
emerald800,065f46;%
emerald900,064e3b;%
emerald950,022c22;%
teal50,f0fdfa;%
teal100,ccfbf1;%
teal200,99f6e4;%
teal300,5eead4;%
teal400,2dd4bf;%
teal500,14b8a6;%
teal600,0d9488;%
teal700,0f766e;%
teal800,115e59;%
teal900,134e4a;%
teal950,042f2e;%
cyan50,ecfeff;%
cyan100,cffafe;%
cyan200,a5f3fc;%
cyan300,67e8f9;%
cyan400,22d3ee;%
cyan500,06b6d4;%
cyan600,0891b2;%
cyan700,0e7490;%
cyan800,155e75;%
cyan900,164e63;%
cyan950,083344;%
sky50,f0f9ff;%
sky100,e0f2fe;%
sky200,bae6fd;%
sky300,7dd3fc;%
sky400,38bdf8;%
sky500,0ea5e9;%
sky600,0284c7;%
sky700,0369a1;%
sky800,075985;%
sky900,0c4a6e;%
sky950,082f49;%
blue50,eff6ff;%
blue100,dbeafe;%
blue200,bfdbfe;%
blue300,93c5fd;%
blue400,60a5fa;%
blue500,3b82f6;%
blue600,2563eb;%
blue700,1d4ed8;%
blue800,1e40af;%
blue900,1e3a8a;%
blue950,172554;%
indigo50,eef2ff;%
indigo100,e0e7ff;%
indigo200,c7d2fe;%
indigo300,a5b4fc;%
indigo400,818cf8;%
indigo500,6366f1;%
indigo600,4f46e5;%
indigo700,4338ca;%
indigo800,3730a3;%
indigo900,312e81;%
indigo950,1e1b4b;%
violet50,f5f3ff;%
violet100,ede9fe;%
violet200,ddd6fe;%
violet300,c4b5fd;%
violet400,a78bfa;%
violet500,8b5cf6;%
violet600,7c3aed;%
violet700,6d28d9;%
violet800,5b21b6;%
violet900,4c1d95;%
violet950,2e1065;%
purple50,faf5ff;%
purple100,f3e8ff;%
purple200,e9d5ff;%
purple300,d8b4fe;%
purple400,c084fc;%
purple500,a855f7;%
purple600,9333ea;%
purple700,7e22ce;%
purple800,6b21a8;%
purple900,581c87;%
purple950,3b0764;%
fuchsia50,fdf4ff;%
fuchsia100,fae8ff;%
fuchsia200,f5d0fe;%
fuchsia300,f0abfc;%
fuchsia400,e879f9;%
fuchsia500,d946ef;%
fuchsia600,c026d3;%
fuchsia700,a21caf;%
fuchsia800,86198f;%
fuchsia900,701a75;%
fuchsia950,4a044e;%
pink50,fdf2f8;%
pink100,fce7f3;%
pink200,fbcfe8;%
pink300,f9a8d4;%
pink400,f472b6;%
pink500,ec4899;%
pink600,db2777;%
pink700,be185d;%
pink800,9d174d;%
pink900,831843;%
pink950,500724;%
rose50,fff1f2;%
rose100,ffe4e6;%
rose200,fecdd3;%
rose300,fda4af;%
rose400,fb7185;%
rose500,f43f5e;%
rose600,e11d48;%
rose700,be123c;%
rose800,9f1239;%
rose900,881337;%
rose950,4c0519}

\setuptodonotes{color=blue!15}

\mathchardef\mhyphen="2D 

\newcommand\newmathabbrev[2]{\newcommand{#1}{\ensuremath{#2}\xspace}}

\newcommand\cfont\mathsf
\newmathabbrev\p{\cfont{P}}
\newmathabbrev{\NN}{\mathbb N}
\newmathabbrev{\FF}{\mathbb F}
\newmathabbrev{\KK}{\mathbb K}
\newmathabbrev{\LL}{\mathbb L}
\newmathabbrev{\Proj}{\mathbb P}
\newmathabbrev{\KKbar}{\overline{\mathbb K}}
\newmathabbrev{\Ubar}{\overline{U}}
\newmathabbrev{\Wbar}{\overline{W}}
\newmathabbrev\NP{\cfont{NP}}
\newmathabbrev\NC{\cfont{NC}}
\newmathabbrev\DTIME{\cfont{DTIME}}
\newmathabbrev\tSAT{3\cfont{\mhyphen{}SAT}}
\newmathabbrev\MA{\cfont{MA}}
\newmathabbrev\AM{\cfont{AM}}
\newmathabbrev\NPDAG{\cfont{NP\mhyphen{}DAG}}
\newmathabbrev\QMADAG{\cfont{QMA\mhyphen{}DAG}}
\newmathabbrev\yes{\mathrm{yes}}
\newmathabbrev\no{\mathrm{no}}
\newmathabbrev\US{\cfont{US}}
\newmathabbrev\FP{\cfont{FP}}
\newmathabbrev\PP{\cfont{PP}}
\newmathabbrev\CeP{\cfont{C_=P}}
\newmathabbrev\coCeP{\cfont{coC_=P}}
\newmathabbrev\PH{\cfont{PH}}
\newmathabbrev\SAT{\cfont{SAT}}
\newmathabbrev\QSAT{\cfont{QSAT}}
\newmathabbrev\hQSAT{\mhyphen\QSAT}
\newmathabbrev\SPP{\cfont{SPP}}
\newmathabbrev\GapP{\cfont{GapP}}
\newmathabbrev\BQP{\cfont{BQP}}
\newmathabbrev\QP{\cfont{QP}}
\newmathabbrev\StoqMA{\cfont{StoqMA}}
\newmathabbrev\coNP{\cfont{coNP}}
\newmathabbrev\AzPP{\cfont{A_0PP}}
\newmathabbrev\QMA{\cfont{QMA}}
\newmathabbrev\SuperQMA{\cfont{SuperQMA}}
\newmathabbrev\PureSuperQMA{\cfont{PureSuperQMA}}
\newmathabbrev\PSQMA{\cfont{PSQMA}}
\newmathabbrev\QMAplus{\QMA{+}}
\newmathabbrev\coQMA{\cfont{coQMA}}
\newmathabbrev\BPP{\cfont{BPP}}
\newmathabbrev\QCMA{\cfont{QCMA}}
\newmathabbrev\pNPlog{\p^{\NP[\log]}}
\newmathabbrev\pNP{\p^{\NP}}
\newmathabbrev\pNPtwo{\p^{\NP[2]}}
\newmathabbrev\pNPone{\p^{\NP[1]}}
\newmathabbrev\pParSAT{\p^{||\SAT}}
\newmathabbrev\pQMApar{\p^{||\QMA}}
\newmathabbrev\pCpar{\p^{||\C}}
\newmathabbrev\pStoqMApar{\p^{||\StoqMA}}
\newmathabbrev\pQMAlog{\p^{\QMA[\log]}}
\newmathabbrev\pClog{\p^{\textup{C}[\log]}}
\newmathabbrev\pC{\p^{\textup{C}}}
\newmathabbrev\QMASPACE{\cfont{QMASPACE}}
\newmathabbrev\pQMAtlog{\p^{\QMA(2)[\log]}}
\newmathabbrev\pStoqMAlog{\p^{\StoqMA[\log]}}
\newmathabbrev\pQMApt{\p^{\Vert\QMA(2)}}
\newmathabbrev\pQMA{\p^{\QMA}}
\newmathabbrev\SharpP{\cfont{\#P}}
\newmathabbrev\pSharP{\p^{\SharpP[1]}}
\newmathabbrev\PromisePP{\cfont{PromisePP}}
\newmathabbrev\lett{\le_\mathrm{tt}}
\newmathabbrev\YES{\mathsf{YES}}
\newmathabbrev\NO{\mathsf{NO}}
\newmathabbrev\PSPACE{\cfont{PSPACE}}
\newmathabbrev\IP{\cfont{IP}}
\newmathabbrev\POLY{\cfont{POLY}}
\newmathabbrev\DAG{\cfont{DAG}}
\newmathabbrev\StoqMADAG{\StoqMA\mhyphen\cfont{DAG}}
\newmathabbrev\CDAG{C\mhyphen\cfont{DAG}}
\newmathabbrev\CDAGf{C\mhyphen\cfont{DAG}_f}
\newmathabbrev\CDAGs{C\mhyphen\cfont{DAG}_s}
\newmathabbrev\CDAGd{C\mhyphen\cfont{DAG}_{d}}
\newmathabbrev\CDAGo{C\mhyphen\cfont{DAG}_1}
\newmathabbrev\LOGS{\cfont{LOGS}}
\newmathabbrev\TAUT{\cfont{TAUTOLOGY}}
\newmathabbrev\SBQP{\cfont{SBQP}}
\newmathabbrev\Fc{F_\coNP}
\newmathabbrev\Fa{F_\AzPP}
\newmathabbrev\GSCON{\cfont{GSCON}}
\newmathabbrev\GSCONexp{\GSCON_\cfont{exp}}
\newmathabbrev\QMAexp{\QMA_\cfont{exp}}
\newmathabbrev\UQMA{\cfont{UQMA}}
\newmathabbrev\RR{\mathbb{R}}
\newmathabbrev\Trees{\cfont{TREES}}
\newmathabbrev\apxsim{\cfont{APX\mhyphen{}SIM}}
\newmathabbrev\AWPP{\cfont{AWPP}}
\newmathabbrev\X{\mathcal{X}}
\newmathabbrev\Y{\mathcal{Y}}
\renewcommand\H{\ensuremath{\mathcal{H}}}
\newmathabbrev\Z{\mathcal{Z}}
\newmathabbrev\ZZ{\mathbb{Z}}
\newmathabbrev\Hprop{H_\mathrm{prop}}
\newmathabbrev\Hin{H_\mathrm{in}}
\newmathabbrev\Hout{H_\mathrm{out}}
\newmathabbrev\Hstab{H_\mathrm{stab}}
\newmathabbrev\Lext{\L_\mathrm{ext}}
\newmathabbrev\BTWNP{\cfont{BTW}(\NP)}
\newmathabbrev\BSN{\cfont{BSN}}
\newmathabbrev\SN{\cfont{SN}}
\newmathabbrev\BD{\cfont{BD}}
\newmathabbrev\HYPERTREE{\cfont{NP\mhyphen{}HYPERTREE}}
\newmathabbrev\Hext{H_\mathrm{ext}}
\newmathabbrev\Hpropt{\tilde{H}_\mathrm{prop}}
\newmathabbrev\Hint{\tilde{H}_\mathrm{in}}
\newmathabbrev\Houtt{\tilde H_\mathrm{out}}
\newmathabbrev\EXP{\cfont{EXP}}
\newmathabbrev\A{\mathcal{A}}
\newmathabbrev\rmU{\mathrm{U}}
\newmathabbrev\rmO{\mathrm{O}}

\newcommand{\calC}{\mathcal{C}}

\newcommand{\wtrho}{\widetilde{\rho}}
\newcommand{\wta}{\widetilde{a}}

\newcommand{\wts}{\widetilde{s}}

\newcommand{\whz}{\widehat{z}}
\newcommand{\why}{\widehat{y}}

\newcommand{\whpsi}{\widehat{\psi}}

\newcommand{\wheta}{\widehat{\eta}}
\newcommand{\whrho}{\widehat{\rho}}
\newcommand{\whsigma}{\widehat{\sigma}}

\newcommand{\wht}{\widehat{t}}
\newcommand{\whtp}{\widehat{t'}}

\newcommand\B{\mathcal B}
\newcommand\Piacc{\Pi_\mathrm{acc}}

\newmathabbrev\DAGSSAT{\DAGS(\SAT)}
\newmathabbrev\DAGS{\mathrm{DAGS}}
\newmathabbrev\DAGSNP{\DAGS(\NP)}

\newmathabbrev\AND{\cfont{AND}}

\newmathabbrev\STCONN{{S,T}\cfont{\mhyphen{}CONN}}
\newmathabbrev\CNF{\cfont{CNF}}
\newmathabbrev\NEXP{\cfont{NEXP}}
\newmathabbrev\NPSPACE{\cfont{NPSPACE}}
\newmathabbrev\QCMASPACE{\cfont{QCMASPACE}}
\newmathabbrev\BQPSPACE{\cfont{BQPSPACE}}
\newmathabbrev{\PCP}{\cfont{PCP}}
\newmathabbrev\BQUPSPACE{\cfont{BQ_UPSPACE}}
\newmathabbrev\QMAt{\QMA(2)}
\newmathabbrev\QMAtexp{\QMAt_{\exp}}
\newmathabbrev\MIP{\cfont{MIP}}
\newmathabbrev\QMIP{\cfont{MIP}}
\newmathabbrev\QSZK{\cfont{QSZK}}
\newmathabbrev\QIP{\cfont{QIP}}
\newmathabbrev\MIPt{\MIP(2)}
\newmathabbrev\BellQMA{\cfont{BellQMA}}
\newmathabbrev\BellPureSymQMA{\cfont{BellPureSymQMA}}
\newmathabbrev\BellSymQMA{\cfont{BellSymQMA}}
\newmathabbrev\BellQMAt{\BellQMA(2)}
\newmathabbrev\BellQMAtexp{\BellQMAt_{\exp}}
\newmathabbrev\Upyth{U_{\mathrm{Pyth.}}}
\newmathabbrev\CLDM{\cfont{CLDM}}
\newmathabbrev\kCLDM{k\mhyphen\CLDM}
\newmathabbrev\lCLDM{l\mhyphen\CLDM}
\newmathabbrev\tCLDM{2\mhyphen\CLDM}
\newmathabbrev\PureCLDM{\cfont{PureCLDM}}
\newmathabbrev\UniquePureCLDM{\cfont{UniquePureCLDM}}
\newmathabbrev\SpectralPureCLDM{\cfont{SpectralPureCLDM}}
\newmathabbrev\precisePureCLDM{\cfont{Precise\mhyphen PureCLDM}}
\newmathabbrev\preciseQMAt{\cfont{PreciseQMA(2)}}
\newmathabbrev\preciseQMA{\cfont{PreciseQMA}}
\newmathabbrev\kPureCLDM{k\mhyphen\PureCLDM}
\newmathabbrev\UniquekPureCLDM{\cfont{Unique}\mhyphen k\mhyphen\PureCLDM}
\newmathabbrev\pureNrep{\cfont{Pure}\mhyphen N\mhyphen\cfont{Representability}}
\newmathabbrev\Nrep{N\mhyphen\cfont{Representability}}
\newmathabbrev\tPureCLDM{2\mhyphen\PureCLDM}
\newmathabbrev\RDM{\cfont{RDM}}
\newmathabbrev\SH{\cfont{SH}}
\newmathabbrev\SLH{\cfont{SLH}}
\newmathabbrev\kRDM{k\mhyphen\RDM}
\newmathabbrev\tRDM{2\mhyphen\RDM}
\newmathabbrev\PureRDM{\cfont{PureRDM}}
\newmathabbrev\BosonPureRDM{\cfont{BosonPure}\mhyphen N \mhyphen \cfont{Representability}}
\newmathabbrev\BosonRDM{\cfont{Boson}\mhyphen N \mhyphen \cfont{Representability}}
\newmathabbrev\kPureRDM{k\mhyphen\PureRDM}
\newmathabbrev\tPureRDM{2\mhyphen\PureRDM}
\newcommand\linear{\mathrm{L}}
\newcommand\SepD{\mathrm{SepD}}

\newcommand{\iu}{\mathrm{i}}

\protected\def\verythinspace{%
  \ifmmode
    \mskip0.5\thinmuskip
  \else
    \ifhmode
      \kern0.08334em
    \fi
  \fi
}

\newcommand{\CC}{\mathbb C}
\newcommand{\QQ}{\mathbb Q}

\newcommand{\be}{\begin{equation}}
\newcommand{\ee}{\end{equation}}

\renewcommand{\epsilon}{\varepsilon}

\DeclareMathOperator{\Herm}{Herm}

\DeclareMathOperator{\Tr}{Tr}

\DeclareMathOperator\rk{rk}
\DeclareMathOperator\sign{sign}
\DeclareMathOperator\Aff{Aff}

\DeclareMathOperator\CSP{CSP}

\DeclareMathOperator\Enc{\mathsf{Enc}}

\newcommand{\psis}{\psi^{(s)}}

\newcommand{\psishist}{\psi^{(s)}_{\textup{hist}}}
\newcommand{\rhohist}{\rho_{\textup{hist}}}

\newcommand{\psiotp}{\psi^{\mathsf{otp}}}
\newmathabbrev{\ChkEnc}{\mathsf{ChkEnc}}
\newmathabbrev{\Chk}{\mathsf{Chk}}
\newmathabbrev{\ResGen}{\mathsf{ResGen}}
\newcommand{\Vxs}{V_{x}^{(s)}}

\newcommand{\Vxotp}{V_{x}^{\mathsf{otp}}}
\newcommand{\Votp}{V^{\mathsf{otp}}}
\newcommand{\Votpd}{V^{\mathsf{otp}\dagger}}
\newcommand{\calVotp}{\calV^{\mathsf{otp}}}
\newcommand{\fextract}{f_{\textup{ex}}}
\newmathabbrev{\SimC}{\mathsf{Sim}_{\calC}}
\newmathabbrev{\SimV}{\mathsf{Sim}_{V^{(s)}}}
\newmathabbrev{\SimVb}{\mathsf{Sim}_{\overline{V}^{(s)}}}
\newmathabbrev{\CNOT}{\mathsf{CNOT}}
\newmathabbrev{\Tgate}{\mathsf{T}}
\newmathabbrev{\Pgate}{\mathsf{P}}
\newmathabbrev{\Zgate}{\mathsf{Z}}
\newmathabbrev{\Xgate}{\mathsf{X}}
\newmathabbrev{\Ygate}{\mathsf{Y}}
\newmathabbrev{\Hgate}{\mathsf{H}}
\newmathabbrev{\FTR}{\text{Th}(\mathbb{R})}
\newmathabbrev{\sol}{sol}

\newcommand{\Uenc}{U_{\mathsf{enc}}}
\newcommand{\Udec}{U_{\mathsf{dec}}}

\newcommand{\calP}{\mathcal{P}}
\newcommand{\calG}{\mathcal{G}}
\newcommand{\calW}{\mathcal{W}}
\newcommand{\calV}{\mathcal{V}}
\newcommand{\calS}{\mathcal{S}}

\newcommand{\calR}{\mathcal{R}}
\newcommand{\calA}{\mathcal{A}}

\newcommand{\tilT}{\tilde{T}}

\newcommand{\tilZ}{\tilde{Z}}

\newcommand{\tilQ}{\tilde{Q}}
\newcommand{\tilY}{\tilde{Y}}
\newcommand{\tilepsilon}{\tilde{\epsilon}}
\newcommand{\tilzeta}{\tilde{\zeta}}
\newcommand{\tilX}{\tilde{X}}

\newcommand{\sfT}{\mathsf{T}}

\newcommand\pmat[1]{\begin{pmatrix}#1\end{pmatrix}}
\newcommand\bmat[1]{\begin{bmatrix}#1\end{bmatrix}}
\newcommand\psmallmat[1]{\begin{psmallmatrix}#1\end{psmallmatrix}}

\newcommand{\poly}{\mathrm{poly}}
\DeclareMathOperator{\polylog}{polylog}

\DeclareMathOperator{\Span}{Span}

\DeclarePairedDelimiter\bra{\langle}{\rvert}
\DeclarePairedDelimiter\ket{\lvert}{\rangle}

\DeclarePairedDelimiter\abs{\lvert}{\rvert}
\DeclarePairedDelimiter\norm{\lVert}{\rVert}

\DeclarePairedDelimiter\fnorm{\lVert}{\rVert_{\mathrm F}}

\DeclarePairedDelimiter\trnorm{\lVert}{\rVert_{\mathrm{tr}}}
\DeclarePairedDelimiter\maxnorm{\lVert}{\rVert_{\mathrm{max}}}
\DeclarePairedDelimiterX\braket[2]{\langle}{\rangle}{#1 \delimsize\vert #2}
\DeclarePairedDelimiterX\ketbra[2]{\lvert}{\rvert}{#1 \delimsize\rangle\delimsize\langle #2}

\newcommand{\braketc}[1]{\braket{#1}{#1}}

\newcommand{\ketbrab}[1]{\ketbra{#1}{#1}}

\setlist[itemize]{noitemsep, topsep=0pt}
\setlist[enumerate]{noitemsep, topsep=0pt}

\declaretheorem[numberwithin=section]{theorem}
\declaretheorem[sibling=theorem]{corollary}
\declaretheorem[sibling=theorem]{lemma}
\declaretheorem[sibling=theorem]{claim}
\declaretheorem[sibling=theorem]{proposition}
\declaretheorem[sibling=theorem]{conjecture}

\declaretheorem[sibling=theorem]{definition}
\declaretheorem[sibling=theorem,style=definition]{remark}
\declaretheorem[sibling=theorem,style=definition]{example}

\crefname{claim}{Claim}{Claims}
\Crefname{claim}{Claim}{Claims}

\newcommand{\spa}[1]{\mathcal{#1}}

\newcommand{\pacc}{p_{\textup{acc}}}

\newcommand{\Shist}{\spa{S}_{\textup{hist}}}

\newcommand{\Ayes}{A_{\textup{yes}}} 
\newcommand{\Ano}{A_{\textup{no}}} 
\newcommand{\psihist}{\psi_{\textup{hist}}}

\newcommand{\Time}{\textup{Time}}

\makeatletter
\newcommand{\subalign}[1]{%
  \vcenter{%
    \Let@ \restore@math@cr \default@tag
    \baselineskip\fontdimen10 \scriptfont\tw@
    \advance\baselineskip\fontdimen12 \scriptfont\tw@
    \lineskip\thr@@\fontdimen8 \scriptfont\thr@@
    \lineskiplimit\lineskip
    \ialign{\hfil$\m@th\scriptstyle##$&$\m@th\scriptstyle{}##$\hfil\crcr
      #1\crcr
    }%
  }%
}
\NewDocumentCommand{\LeftComment}{s m}{%
  \Statex \IfBooleanF{#1}{\hspace*{\ALG@thistlm}}\(\triangleright\) #2}

\def\moverlay{\mathpalette\mov@rlay}
\def\mov@rlay#1#2{\leavevmode\vtop{%
   \baselineskip\z@skip \lineskiplimit-\maxdimen
   \ialign{\hfil$\m@th#1##$\hfil\cr#2\crcr}}}
\newcommand{\charfusion}[3][\mathord]{
    #1{\ifx#1\mathop\vphantom{#2}\fi
        \mathpalette\mov@rlay{#2\cr#3}
      }
    \ifx#1\mathop\expandafter\displaylimits\fi}

\makeatother

\newcolumntype{M}[1]{>{\centering\arraybackslash$}m{#1}<{$}}

\makeatletter
\newcommand\@rcolwidth{0.67em}

\def\@rarray[#1]{\arraycolsep=0pt\array{*\c@MaxMatrixCols {M{#1}}}}
\makeatother

\newcommand{\splitatcommas}[1]{
  \begingroup
  \begingroup\lccode`~=`, \lowercase{\endgroup
    \edef~{\mathchar\the\mathcode`, \penalty0 \noexpand\hspace{0pt plus 1em}}%
  }\mathcode`,="8000 #1%
  \endgroup
}

\AtEveryBibitem{%
  \clearlist{language}%
}

\renewbibmacro*{doi+eprint+url}{%
  \iftoggle{bbx:doi}
    {\printfield{doi}}
    {}%
  \newunit\newblock
  \ifboolexpr{togl {bbx:eprint} and test {\iffieldundef{doi}}}
    {\usebibmacro{eprint}}
    {}%
  \newunit\newblock
  \ifboolexpr{togl {bbx:url} and test {\iffieldundef{doi}}  and test {\iffieldundef{eprint}}}
    {\usebibmacro{url+urldate}}
    {}}

\usepackage{tikz}
\usetikzlibrary{quantikz2}

\title{The Pure-State Consistency of Local Density Matrices Problem:\\In PSPACE and Complete for a Class between QMA and QMA(2)}
\author{Jonas Kamminga\footnote{Department of Computer Science and Institute for Photonic Quantum Systems (PhoQS), Paderborn University, Germany. Email: \{jonas.kamminga, dorian.rudolph\}@upb.de.} \and Dorian Rudolph\footnotemark[1]}
\date{April 3, 2024}
\begin{document}

\maketitle
\begin{abstract}
In this work we investigate the computational complexity of the pure consistency of local density matrices ($\mathsf{PureCLDM}$) and pure $N$-representability ($\mathsf{Pure}$-$N$-$\mathsf{Representability}$; analog of $\mathsf{PureCLDM}$ for bosonic or fermionic systems) problems. In these problems the input is a set of reduced density matrices and the task is to determine whether there exists a global \emph{pure} state consistent with these reduced density matrices. While mixed $\mathsf{CLDM}$, i.e. where the global state can be mixed, was proven to be $\mathsf{QMA}$-complete by Broadbent and Grilo [JoC 2022], almost nothing was known about the complexity of the pure version. Before our work the best upper and lower bounds were $\mathsf{QMA}(2)$ and $\mathsf{QMA}$. Our contribution to the understanding of these problems is twofold.

Firstly, we define a pure state analogue of the complexity class $\mathsf{QMA}^+$ of Aharanov and Regev [FOCS 2003],  which we call $\mathsf{PureSuperQMA}$. We prove that both $\mathsf{pure}$-$N$-$\mathsf{Representability}$ and $\mathsf{PureCLDM}$ are complete for this new class. Along the way we supplement Broadbent and Grilo by proving hardness for 2-qubit reduced density matrices and showing that mixed $N$-$\mathsf{Representability}$ is $\mathsf{QMA}$-complete.

Secondly, we improve the upper bound on $\mathsf{PureCLDM}$. Using methods from algebraic geometry, we prove that $\mathsf{PureSuperQMA} \subseteq \mathsf{PSPACE}$. Our methods, and the $\mathsf{PSPACE}$ upper bound, are also valid for $\mathsf{PureCLDM}$ with exponential or even perfect precision, hence $\mathsf{precisePureCLDM}$ is not $\mathsf{preciseQMA}(2) = \mathsf{NEXP}$-complete, unless $\mathsf{PSPACE} = \mathsf{NEXP}$. We view this as evidence for a negative answer to the longstanding open question whether $\mathsf{PureCLDM}$ is $\mathsf{QMA}(2)$-complete.
 
The techniques we develop for our $\mathsf{PSPACE}$ upper bound are quite general. We are able to use them for various applications: from proving $\mathsf{PSPACE}$ upper bounds on other quantum problems to giving an efficient parallel ($\mathsf{NC}$) algorithm for (non-convex) quadratically constrained quadratic programs with few constraints.
\end{abstract}

\newpage
\section{Introduction}
``Are these local density matrices consistent with some global state?'' This quantum variant of the $\NP$-complete \emph{marginal problem} is known as the \emph{consistency of local density matrices problem} ($\CLDM$) or \emph{quantum marginal problem} and as the \emph{$N$-representability} problem ($\Nrep$)\footnote{The $N$-representability problem is the analog of the $\CLDM$ problem for bosonic/fermionic systems. Due to symmetry/antisymmetry these systems have a different notion of \emph{reduced density matrix} than qubits, which requires technical considerations when translating results from qubit systems (see \cref{sec:pureNrep}).} when dealing with indistinguishable particles. It is of fundamental interest to quantum physics. In fact, it was already recognized as an important question in the sixties \cite{Col63}. At that time, the hope was that the ground state energy of quantum systems could be found using reduced density matrices. This hope was supported by the fact that Hamiltonians showing up in nature are all \emph{local}. One requirement would be that it is possible to check that alleged reduced density matrices are indeed consistent with a valid global quantum state, hence the interest in the $\CLDM$ problem.

Over the years it became apparent that this hope would not materialize, especially when Kitaev proved that computing ground state energies of local Hamiltonians is $\QMA$-hard \cite{KSV02}. 
Also the $\CLDM$ problem itself was proven to be hard. First by Liu, who proved that the (mixed state) $\CLDM$ problem is contained in $\QMA$ and $\QMA$-complete under Turing reductions \cite{Liu06} and later, together with Christandl and Verstraete, proved that the $N$-representability problem is also $\QMA$-complete under Turing reductions \cite{LCV07}. 
This was improved by Broadbent and Grilo who showed that the mixed $\CLDM$ is also $\QMA$-hard under Karp reductions, establishing it as $\QMA$-complete \cite{BG22}.

In this work we focus our attention on a variant of $\CLDM$ that is of interest to physics and complexity theory. This variant, the \emph{pure} consistency of local density matrices problem, or $\PureCLDM$ asks not whether there exists any state consistent with the reduced density matrices, but demands that this consistent state is pure. Physically, this restriction to pure states is well motivated as isolated systems will always be in a pure state. If one has some reduced density matrices of some molecule it might be more interesting to know whether these are consistent with some global state of the molecule in isolation, i.e. a pure state, rather than some larger state where the molecule is entangled with its environment. 

From a complexity point of view $\PureCLDM$ is interesting because of its connections with the class $\QMAt$. While the complexity of mixed $\CLDM$ is well understood to be $\QMA$, the best known upper bound to $\PureCLDM$ is the class $\QMAt$ and it has been a longstanding open question whether $\PureCLDM$ is complete for $\QMAt$ \cite{LCV07}. If it would be complete that would be interesting for several reasons. Firstly, it would be a highly physically motivated complete problem for $\QMAt$, a class that has attracted a lot of attention but generally lacks physical complete problems. Secondly, there has been a large interest in the relation between purity testing and $\QMAt$. In many ways, the power of $\QMAt$ seems to come from the ability to do SWAP tests \cite{HM13}. On the other hand, a correspondence between protocols over separable states and protocols with a purity constraint has been used in several works (e.g. \cite{HM13,YSWNG21}). Recently, it has been suggested that purity testing could already be enough to give $\QMAt$ its power \cite{BFLMW24}. Establishing $\PureCLDM$ as a complete problem for $\QMAt$ would suggest that indeed SWAP tests are not necessary to capture the power of $\QMAt$ and indeed a purity constraint is already enough. 

\subsection{Our results}
In this work we investigate the complexity of $\PureCLDM$ and $\pureNrep$ and give evidence towards a negative answer to the longstanding open question whether they are $\QMAt$-complete.

\begin{definition}[{\kPureCLDM}, informal]
  We are given pairs $(\rho_m, C_m), \dots, (\rho_m, C_m)$, where the $\rho_i$ are reduced density matrices and $C_i \subseteq [n]$ with $\abs{C_i} \le k$ for all $i$. Additionally, we are given parameters $\alpha$ and $\beta$ with $\beta - \alpha \ge 1/\poly(n)$. Decide whether:
  \begin{itemize}
    \item YES: there exists a consistent pure state, that is, a state $\ket{\psi}$ such that $\norm{\Tr_{\overline{C_i}}(\ketbrab{\psi}) - \rho_i} \le \alpha$ for all $i \in [m]$.
    \item NO: all pure states are ``far from'' consistent, that is, for all $\ket{\psi}$, there is an $i \in [m]$ with $\norm{\Tr_{\overline{C_i}}(\ketbrab{\psi}) - \rho_i} \ge \beta$.
  \end{itemize}
\end{definition}

\subsubsection{$\PureCLDM$ is complete for $\PureSuperQMA$}
Our first main result on the complexity of $\PureCLDM$ is that it is complete for the class of (promise) problems decided by a ``super-verifier'' (c.f. \cite{AR03}) that is given a pure state proof. Following \cite{AR03} we call this class $\PureSuperQMA$.\footnote{\cite{AR03} use the name $\QMAplus$, but recently $\QMA^+$ has been used to refer to $\QMA$ with proofs with nonnegative amplitudes \cite{JW23,BFM24}. As $\QMAplus$ is sometimes referred to as $\QMA$ with a super-verifier, we use $\SuperQMA$.}

\begin{definition}[{\PureSuperQMA}, informal]
  A promise problem $A$ is in $\PureSuperQMA(m,\epsilon,\delta)$ if there exist $m$ constraints $\calV = \{(V_{x,i},r_{x,i},s_{x,i})\}_{i\in [m]}$ such that:
  \begin{itemize}
    \item $\forall x\in \Ayes\; \exists \ket{\psi} \colon \Pr_i(\abs{p(V_{x,i},\psi) - r_{x,i}}\le s_{x,i})=1$, that is, the $V_{x,i}$ accept $\ket{\psi}$ with acceptance probability at most $s_{x,i}$ away from $r_{x,i}$.
    \item $\forall x\in \Ano\; \forall \ket{\psi} \colon \Pr_i(\abs{p(V_{x,i},\psi)- r_{x,i}}\le s_{x,i} + \epsilon) \le 1-\delta$, that is, at least a $\delta$ fraction of the $V_{x,i}$ accept $\ket{\psi}$ with probability more than $s_{x,i} + \epsilon$ away from $r_{x,i}$.
  \end{itemize}
  Here $p(V, \psi)$ denotes the acceptance probability of the circuit $V$ when ran on input $\ket{\psi}$.

  In the following we refer to $m$ as the number of \emph{checks} or \emph{constraints} and to $\epsilon$ and $\delta$ as precision parameters.
  We denote the union of $\PureSuperQMA(m,\epsilon,\delta)$ where $m$ is polynomial and $\epsilon$ and $\delta$ are inverse polynomial as $\PureSuperQMA(\poly, 1/\poly, 1/\poly) \eqcolon \PureSuperQMA$. 
\end{definition}
We show that the definition of $\PureSuperQMA$ can be significantly simplified: we can set all $r_{x,i} = 1/2$ and all $s_{x,i} = 0$ without changing the complexity. 

Our proof of $\PureSuperQMA$-completeness holds already for $\kPureCLDM_1$, which is $\kPureCLDM$ with ``exact consistency'' as well as for Bosonic and Fermionic $\pureNrep$.
\begin{restatable}{theorem}{tPureCLDMcomplete}
  \label{thm:2PureCLDM-complete}
  $\kPureCLDM_1$ is \PureSuperQMA-complete for all $k\ge2$.
\end{restatable}
\begin{restatable}{theorem}{fermioncomplete}
  Fermionic $\pureNrep_1$ is $\PureSuperQMA$-complete.
\end{restatable}
\begin{restatable}{theorem}{bosoncomplete}
  Bosonic $\pureNrep_1$ is $\PureSuperQMA$-complete.
\end{restatable}

Note that $\PureSuperQMA$ as the complete complexity class is quite natural, since Liu's proof for $\CLDM\in\QMA$ also goes via $\SuperQMA=\QMA$.
Since a pure state analogue to this latter statement is not known, we have to remain the super-verifier regime.
Along the way to \cref{thm:2PureCLDM-complete}, we show that the $\kCLDM$ problem is already $\QMA$-hard for $k \ge 2$. This improves upon the results by Broadbent and Grilo \cite{BG22} who show hardness for $k \ge 5$. We also resolve one of their open questions by showing that the (mixed) fermionic and bosonic $N$-representability problems are $\QMA$-hard, already for $2$-particle reduced density matrices.

\subsubsection{$\PureSuperQMA$ as a ``light'' version of $\QMAt$}
Next, we prove results suggesting $\PureSuperQMA$ can be viewed as a ``light'' version of $\QMAt$. It lies between $\QMA$ and $\QMAt$ but contains $\NP$ already with $\log$-size proofs and becomes $\NEXP$ when the number of checks and the precision are exponential. This mirrors $\QMAt$ that also contains $\NP$ with $\log$-size proofs \cite{BT12} and becomes equal to $\NEXP$ when made exponentially precise \cite{Per12}. On the other hand, $\QMA$ with $\log$-size proofs is equal to $\BQP$ \cite{MW05}, which is believed not to contain $\NP$, and $\preciseQMA = \PSPACE$ \cite{FL16} is believed to be different from $\NEXP$.
\Cref{fig:zoo} depicts the relationships between the classes discussed here.

\begin{restatable}{proposition}{PureSuperQMAinQMAtwo}
  \label{prop:PureSuperQMA in QMA(2)}
  $\QMA \subseteq \PureSuperQMA \subseteq \PureSuperQMA(\exp,1/\poly,1/\poly) \subseteq \QMAt$.
\end{restatable}

\begin{restatable}{proposition}{logsizeproofs}
  \label{logsizeproofs}
  $\NP$ can be decided in $\PureSuperQMA$ already with $\log$-size proofs.
\end{restatable}

\begin{restatable}{proposition}{expPSQMAisNEXP}
  $\PureSuperQMA(\exp,1/\exp,1/\exp) = \NEXP$
\end{restatable}

\subsubsection{$\PureSuperQMA$ is upper bounded by $\PSPACE$}
Proving a non-trivial upper bound on the complexity of $\QMAt$ is a major open question in the field. As such, one might wonder whether our ``light'' version of $\QMAt$ does admit a non-trivial upper bound. Our second main result is that it does indeed, namely $\PSPACE$. This sharpens the upper bound on the complexity of $\PureCLDM$ from $\QMAt \subseteq \NEXP$ to $\PSPACE$. 
\begin{restatable}{theorem}{PureSuperQMAinPSPACE}
  \label{thm:PureSuperQMAinPSPACE}
  $\PureSuperQMA \subseteq \PureSuperQMA(\poly, 0, 1/\poly) \subseteq \PSPACE$.
\end{restatable} 
Our proof relies on methods from algebraic geometry and also works for $\PureSuperQMA$ with exponential or even perfect precision \emph{as long as the number of constraints is polynomial}.

What does this mean for $\QMAt$? Of course showing that $\PureCLDM$ is $\QMAt$-hard would imply $\QMAt \subseteq \PSPACE$ but there is a catch. To prove \cref{thm:PureSuperQMAinPSPACE} we use methods from algebraic geometry that also work for $\PureSuperQMA$ with exponential precision and even with perfect precision. This allows us to obtain the following corollary.
\begin{corollary}
  If $\precisePureCLDM$ is $\preciseQMAt$-hard, then $\PSPACE = \NEXP$.
\end{corollary}
This corollary can be interpreted either as evidence that $\PureCLDM$ is not $\QMAt$-hard, or as a barrier to a hardness proof: any such proof must fail in the precise case, assuming that $\PSPACE \ne \NEXP$.

Furthermore, by combining \cref{prop:PureSuperQMA in QMA(2),thm:2PureCLDM-complete} we conclude that $\PureSuperQMA$ can only be $\QMAt$-complete if a polynomial number of constraints give the same power as exponentially many constraints.
\begin{corollary}
  If $\PureSuperQMA(\poly, 1/\poly, 1/\poly) \subsetneq \PureSuperQMA(\exp, 1/\poly, 1/\poly)$, then $\PureCLDM$ is not $\QMAt$-hard.
\end{corollary}

Lastly, the fact that $\QMAt$ contains $\NP$ already with $\log$-size proofs is sometimes taken as evidence that $\QMAt$ could be equal to $\NEXP$. We argue that our results imply that this is not evidence at all. Indeed, $\PureSuperQMA$ also contains $\NP$ with $\log$-size proofs but is contained in $\PSPACE$.

\begin{figure}[t]
\begin{center}
  \begin{tikzpicture}[thick]
    \node (nexp) at (3,7.5) {$\NEXP=\cfont{PrecQMA}(2)$};
    \node (qsi3) at (1.5,6) {$\cfont{Q\Sigma}_3$};
    \node (qmat) at (0,4.5) {$\QMAt$};
    \node (pspace) at (6,4.5) {$\PSPACE=\cfont{PrecQMA}$};
    \node (psqmae) at (0,3) {$\PureSuperQMA\scriptsize(\exp,1/\poly,1/\poly)$};
    \node (bellqma) at (6,3) {$\BellPureSymQMA$};
    \node (psqma) at (3,1.5) {$\PureSuperQMA$};
    \node (qma) at (3,0) {\QMA};
    \draw (qma) -- (psqma);
    \draw (psqma.north west) -- (psqmae);
    \draw (psqmae) -- (qmat);
    \draw (qmat) -- (qsi3);
    \draw (qsi3) -- (nexp);
    \draw (pspace) -- (nexp);
    \draw (psqma.north east) -- (bellqma);
    \draw (bellqma) -- (pspace);
    \draw (bellqma.north west) -- (qmat);
    \draw[dashed,violet800] (qmat) -- (pspace) node [above,midway]{?};
  \end{tikzpicture}
\end{center}
\caption{Diagram of the complexity classes discussed in this paper. Containment of $\PureSuperQMA$ in $\PSPACE$ is proven in \cref{thm:PureSuperQMAinPSPACE}. The class $\BellPureSymQMA$ and its containment in $\PSPACE$ is discussed in \cref{sec:BellPureSymQMA}. Containment of $\PureSuperQMA(\exp,1/\poly,1/\poly)$ in $\QMAt$ is discussed in \cref{prop:PureSuperQMA in QMA(2)}. $\PSPACE=\preciseQMA$ is shown in \cite{FL16}. The third level of the quantum polynomial hierarchy, $\cfont{Q\Sigma}_3$, is between $\QMAt$ and $\NEXP$ \cite{GSSSY22}. $\NEXP=\preciseQMAt$ is shown in \cite{BT12,Per12}. The relationship between $\PSPACE$ and $\QMAt$ remains a major open problem.}
\label{fig:zoo}
\end{figure}

\subsubsection{Solving systems of quadratic equations efficiently in parallel}
To prove \cref{thm:PureSuperQMAinPSPACE} we rely on techniques by Grigoriev and Pasechnik \cite{GP05} for solving certain polynomials. We call such polynomials \emph{GP systems}.
\begin{definition}[GP system (informal)]
  A GP system is a polynomial of the form $p(Q(X))$, where $p\colon \RR^k \to \RR$ is a degree $d$ polynomial in ``few'' variables, and $Q\colon \RR^N \to \RR^k$ is quadratic in ``many'' variables. We assume the coefficients of $p$ and $Q$ are integers of bit size at most $L$. A solution to the system is a point $X^* \in \RR^N$ such that $p(Q(X^*)) = 0$. 
\end{definition}
Grigoriev and Pasechnik give an algorithm that can compute a set of univariate representation of solutions of $p(Q(X))$ that intersects all connected components of the solution set \cite[Theorem~1.2]{GP05}. Their algorithm requires $(dN)^{O(k)}$ arithmetic operations.

We show how to write a $\PureSuperQMA$ verifier as a GP system with $k = \poly(n)$ and $N,L = 2^{\poly(n)}$ such that the system has a solution iff the verifier accepts some pure proof. The results by Grigoriev and Pasechnik then immediately imply $\PureSuperQMA \subseteq \EXP$. To prove containment in $\PSPACE$ we modify the algorithm to work efficiently in parallel. The techniques we develop turn out to be much more general. They also show that solving quadratic systems with few constraints is in $\NC$, or Nick's Class, which captures efficient parallel computation.
\begin{definition}[Nick's Class \cite{Arora2009}]
  The $i$-th level of $\NC$, $\NC^i$, consists of all languages decidable in time $O(\log^i(n))$ using $\poly(n)$ parallel processors. $\NC$ is the union of $\NC^i$ over all $i\in\NN$.\footnote{Equivalently, $\NC^i$ can be defined as the class of languages decidable by an $O(\log^in)$-depth circuit of size $\poly(n)$.}
\end{definition}
The class $\NC(\poly)$ of all decision problems decidable in $\poly(n)$ time on $\exp(n)$ processors is equal to $\PSPACE$ \cite{Bor77}.

Our modification of Grigoriev and Pasechnik's algorithm yields:
\begin{restatable}{theorem}{GPparallel}
  \label{thm:parallelalgforGP}
  Let $p,Q$ be a GP system. There is a parallel algorithm for deciding whether $p(Q(X))$ has a root. The algorithm uses
  \begin{equation}
    L^2(k^2Nd)^{O(k)}
  \end{equation}
  parallel processors and needs
  \begin{equation}
    \poly(\log N, \log d, k, \log L)
  \end{equation}
  time. In particular, if $\log N, \log d, \log L, k \le \poly(n)$ the computation can be done in $\NC(\poly) = \PSPACE$. If $N, d, L \le \poly(n)$ and $k = O(1)$ then the algorithm is in $\NC$.
\end{restatable}

We further build upon \cite{GP05} and consider the following optimization problem which includes the task of quantifying how inconsistent local density matrices are.
\begin{definition}[OptGP]
  \label{def:OptGP}
  An OptGP instance consists of a degree $\le d$ polynomial $r\colon \RR^k \to \RR$ and a GP system $p(Q(X))$ with bounded solution set $Z = \{X \in \RR^N \colon p(Q(X)) = 0\}$. The task is to compute or approximate $\theta = \min_{X \in Z} r(Q(X))$. 
\end{definition}
We give a parallel algorithm for approximating the optimum of a OptGP system as well as approximating a witness to this optimum.
\begin{restatable}{theorem}{OptGP}
  \label{thm:approx-opt}
  Let $r, p, Q$ be an OptGP instance. There is a parallel algorithm that computes a $\delta$-approximation to $\theta = \min_{X \in Z} r(Q(X))$ and to some $X^*$ with $r(Q(X^*)) = \theta$ and $p(Q(X^*)) = 0$. The algorithm uses 
  \begin{equation}
    \poly\left(L, |\log \delta|, (kNd)^{O(k^2)}\right)
  \end{equation}
  parallel processors and needs
  \begin{equation}
    \poly\left(\log N, \log d, k, \log L, \log|\log \delta|\right)
  \end{equation}
  time.
\end{restatable}
Grigoriev and Pasechnik \cite[Theorem 1.5]{GP05} also state a theorem with a sequential algorithm for the optimization problem that also applies to unbounded $Z$, but their proof is not yet available to the best of our knowledge.

Finally, we showcase the applicability of the developed algorithms. First, we improve upon a result by Shi and Wu. In \cite{SW15} they give a $\PSPACE$ algorithm for optimizing the energy of ``decomposable'' Hamiltonians over separable states. Using our framework we are able to reprove this fact, and even get a better runtime dependence on the error.

Second, we show that deciding if there exists a \emph{unique}\footnote{We say a state $\ket{\phi}$ is the unique state consistent with some local density matrices if any state that is orthogonal to $\ket{\phi}$ is far from consistent.} pure state that is consistent with given local density matrices is also in $\PSPACE$. In other words, we can decide in $\PSPACE$ whether the local density matrices \emph{fully} describe the physics of the system.

Third, we show how to decide a variant of $\PureCLDM$, where the input only specifies the spectrum of the local density matrices. This version is sometimes referred to as the quantum marginal problem, although others use that name for our $\PureCLDM$.

Fourth, we consider a pure-state variant of $\BellSymQMA(\poly)$. This class (see \cite{ABDFS08,BH17}) captures $\QMA$ protocols where Merlin promises to send $\poly$ many copies of a proof in tensor product, but Arthur has measure each copy independently and then decide to accept or reject based on the classical measurement outcomes. It turns out that $\BellSymQMA(\poly)$ is actually equal to $\QMA$ \cite{BH17}. The variant we consider, $\BellPureSymQMA(\poly)$ adds the additional promise that the proof is pure and does not obviously equal $\QMA$. It is easily seen that $\PureSuperQMA \subseteq \BellPureSymQMA(\poly)$ as Arthur can approximate the acceptance probability of the checks by doing many independent measurements. We slightly strengthen \cref{thm:PureSuperQMAinPSPACE} by proving that also $\BellPureSymQMA(\poly) \subseteq \PSPACE$.

Lastly, we consider an entirely classical problem: quadratic programming with few constraints. A quadratically constrained quadratic program (QCQP) is an optimization problem of the form
\begin{subequations}
  \begin{align}
    \min_{x \in \RR^n} &x^\sfT A_0 x + a_0^\sfT x \\
    \text{subject to } &x^\sfT A_i x + a_i^\sfT x \le b_i \quad \forall i \in [m].
  \end{align}
\end{subequations}
As a final application of the GP framework, we show that QCQPs can be solved in $\NC$ \emph{if the number of constraints $m$ is constant}. Note that we do not require any assumptions on the $A_i$ such as positivity, symmetry or bounds on the condition number.

\subsection{Proof techniques}
We now sketch the proofs of our main theorems, organized by topic.
\paragraph{$\PureSuperQMA$-hardness.}
The proof of the $\PureSuperQMA$-completeness of $\PureCLDM$ closely follows Broadbent and Grilo's proof of the $\QMA$-hardness of the mixed $\CLDM$ problem but with several key changes. Before we elaborate on those, let us sketch the original proof.

Starting with an arbitrary $\QMA$-verifier $V$, one can use Kitaev's circuit-to-Hamiltonian construction to obtain a Hamiltonian whose low energy states include the history state of the computation. One would like to construct local density matrices that are consistent with a global state (the history state) if and only if the original computation was accepting. However, one obstacle is the dependence of the history state on the witness (or proof) state. To circumvent this problem, Broadbent and Grilo use \emph{$s$-simulatable} codes from \cite{GSY19}. These are codes whose codewords are \emph{$s$-simulatable}, that is, their reduced density matrices on at most $s$ qubits can be efficiently computed by a classical algorithm, just like their evolution under local unitaries. They now consider a different verification circuit $V'$ that implements the original circuit $V$ on data encoded with such an $s$-simulatable code, starting from a similarly encoded proof state. From the properties of the code, it follows that the reduced density matrices of the history state corresponding to $V'$ can be efficiently constructed. As the final step of the proof, it is shown that these reduced density matrices indeed are consistent if and only if the original computation was accepting.

To adapt this approach to our needs we make several important changes.
\begin{enumerate}
  \item To make sure that the proof is indeed encoded correctly, Broadbent and Grilo add a step to their protocol enforcing this, which essentially boils down to decoding and immediately encoding again. To make sure that the reduced density matrices can also be computed \emph{during} this process, they ask for the proof as encrypted by a one-time pad, together with the keys. This one-time pad is then undone only after checking the encoding.
  
  It is this one-time padding that makes the consistent state a mixed state, which we want to avoid. To do so we reduce the number of possible one-time pad keys and do a separate check for each key. We do this by using the same key for the one-time pad encryption of every qubit. This means that individual proof qubits are still in a maximally mixed state, but there are only $4$ different keys.
  We abstract this change away into a modified super-verifier that has an accepting proof with maximally mixed $1$-local density matrices (see \cref{sec:2-local}).

  \item We use the $2$-local circuit-to-Hamiltonian construction of \cite{KKR06} instead of Kitaev's original $5$-local construction \cite{KSV02}. This is so we can easily extend to the $N$-representability problem, which is ordinarily defined with $2$-particle density matrices. Using this different construction causes some technical issues. To resolve these we introduce an ``Extraction Lemma'', which allows extracting $1$-local density matrices at certain time steps from the $2$-local density matrices of the history state (see \cref{sec:2-local}).
  
  \item We need to check the proof against multiple constraints. For each constraint, we apply its circuit, decode the output qubit, encode the output qubit, and finally undo the circuit (see \cref{sec:sim-super-verifier}).
  The output probability can be extracted from the time step between decoding and encoding.
\end{enumerate}

\paragraph{$\PSPACE$ upper bound.}
The main obstacle to prove the $\PSPACE$ upper bound is that the purity constraint is not a convex constraint. This prevents convex optimization approaches from being used, which are the standard for proving containment in $\PSPACE$. We take a wholly different approach: we convert a $\PureSuperQMA$ instance into a system of polynomials and use methods from algebraic geometry to solve these.

We begin by writing $\Pr(V_i \text{ accepts } \ket{\psi}) - \frac{1}{2}$ as a real multivariate polynomial. The real and complex parts of every coefficient of the proof state will be represented by separate variables.
This yields for every constraint $V_i$, a polynomial $Q_i\colon \RR^{2N} \to \RR$, where $N = 2^n$ is the dimension of the proof state. These $Q_i$ are polynomials in exponentially many variables, which might seem bad as it is $\NEXP$-hard to determine if a general polynomial in exponentially many variables of degree $\ge 4$ has a zero.\footnote{The statement for degree $d\ge 4$ follows because linear programming over $0\mhyphen 1$ with exponentially many variables is $\NEXP$-hard. Restriction to $0\mhyphen 1$ can be enforced by the quartic polynomial equality $\sum_i (x_i^2 - 1)^2 = 0$. It is also $\NEXP$-hard to determine if a system of degree $3$ polynomials has a zero. To see this we use $\QMAt_{\cfont{exp}} = \NEXP$. The acceptance probability $\bra{\psi} \Pi_{acc} \ket{\psi}$ is a quadratic polynomial and the restriction to separable proofs can be enforced using the degree $3$ polynomial $\bra{\psi}(\ket{\phi_1}\otimes \ket{\phi_2}) - 1 = 0$.} However, and this turns out to be crucial, the $Q_i$ have a degree of a most $2$. We combine the $Q_i$ by taking another specially constructed polynomial $p$, this one with polynomially many variables and degree $d = \poly(n)$, and considering $p(Q(X)) = 0$. We ensure that this latter equation will have a solution iff the $\PureSuperQMA$ verifier accepts.

To solve this system, we use results by Grigoriev and Pasechnik \cite{GP05}. To our knowledge, this is the first time these techniques are used in a quantum context. Because the techniques are quite general and powerful we hope they will find more use there. 

Grigoriev and Pasechnik exhibit an algorithm for solving such systems $p(Q(X))=0$ with quadratic $Q$ in exponential time. We will refer to such polynomials as GP systems. We modify their algorithm to get an efficient parallel algorithm, that is, an $\NC(\poly) = \PSPACE$ algorithm for deciding if there is a zero. 
Broadly, their original algorithm consists of two steps. First, they show how such a system $p(Q(X)) = 0$ can be reduced to a set of (exponentially many) different polynomial systems, each consisting of polynomially many equations in only polynomially many variables. These smaller systems are called ``pieces''. They prove that solutions to the original system, at least one in every connected component, can be recovered from the solutions of the pieces. The pieces could be solved using standard methods in exponential time or $\PSPACE$, but there is a catch: for the reduction of the number of variables, Grigoriev and Pasechnik rely on three key assumptions. These are almost always\footnote{They are generically true. Informally, this means that there is some polynomial that is 0 iff they do not hold.} satisfied, but can fail for certain degenerate cases. To circumvent this issue, they consider small perturbations of the original system and show that for sufficiently small values of these perturbations all assumptions are satisfied. Next, they show that the solutions to the original system are exactly equal to the limits of solutions of the perturbed system as the perturbations go to $0$. The second part of their work is concerned with the computation of these limits.

To get an efficient parallel algorithm, we mostly leave the first step as it is, but compute the limits differently. Whereas Grigoriev and Pasechnik consider the solutions of the perturbed systems as Puiseux series in the (infinitesimal) perturbations, we consider the perturbations as variables and the zeros as a set-valued function of these variables. We show that in this perspective the zeros of the original system are still equal to the limits of the solutions of the perturbed system. Our new perspective allows us to write the limit of the set of solutions as the set of points satisfying some formula in the first-order theory of the reals. A $\PSPACE$ algorithm for deciding the first-order theory of the reals (we use \cite{Ren92-2}) can now be used to determine, for each piece, whether the corresponding solution set is empty. Doing these checks for all of the exponentially many pieces in parallel results in a $\PSPACE$ algorithm for deciding if $p(Q(X))$ has any solutions.

The approximation algorithm follows by using an algorithm to find approximate solutions to first-order theory of the reals formulas \cite{Ren92-4}. We cannot directly use this to extract the entire solution though, as the number of entries is too big. Instead, we isolate a solution using a univariate encoding and extract all entries in parallel.

Having developed the above machinery, the application we exhibit follow straightforwardly.

\subsection{Related work}
The computational complexity of (mixed) $\CLDM$ and $\Nrep$ has previously been studied by Liu, Broadbent and Grilo, as mentioned before. Liu \cite{Liu06} proves that (mixed) $\CLDM$ is contained in $\QMA$ and hard under Turing reductions. Similar results for $\Nrep$ are proven in \cite{LCV07}. This was improved by Broadbent and Grilo who proved (among other results regarding zero-knowledge proof systems) that (mixed) $\CLDM$ is also $\QMA$-hard under Karp reductions, thereby fully resolving its complexity \cite{BG22}. Both Liu, and Broadbent and Grilo do not intensively study $\PureCLDM$, although \cite{LCV07} does show containment of fermionic $\pureNrep$ in $\QMAt$, leaving hardness as an open question. A similar containment for bosonic $\pureNrep$ was shown in \cite{WMN10}.

That does not mean that $\PureCLDM$ and $\pureNrep$ have not been studied before.
There is a large body of work focussing on finding necessary and/or sufficient conditions for reduced density matrices to be consistent with a global state. Among these works is \cite{Kly04}, which focuses on the case where the reduced density matrices are non-overlapping. The paper establishes conditions that are necessary and sufficient for the existence of a consistent pure state in this case. Mazziotti \cite{Maz16} derives necessary conditions for a two-fermion density matrix to have a consist global $N$-fermion pure state.

\cite{YSWNG21} rewrite $\PureCLDM$ as an optimization problem over separable state. They then apply the method of symmetric extensions to this notoriously hard problem to describe $\PureCLDM$ as a hierarchy of SDP's. That is, they describe SDP's depending on a parameter $N$ such that any ``No'' instance will be discovered by the SDP for sufficiently large $N$. They do not, however, prove any upper bounds on the size of $N$ required.

In \cite{BFLMW24}, the authors consider $\QMA$ with an \emph{internally separable} proof. They prove that when this proof is mixed, the class is contained in $\EXP$, whereas it is equal to $\NEXP$ if the proof is pure. This provides the, to our knowledge first, instance where pure proofs are provably stronger than mixed proofs, modulo standard complexity theoretic assumptions.

An algorithm for solving polynomial systems more general than those considered in \cref{thm:parallelalgforGP} is given in \cite{Gri13}. It shows that a system of $k$ polynomials of degree $d$ in $n$ variables can be solved in time $\poly\bigl(n^{d^{3k}}\bigr)$. One downside to this algorithm is that it finds a solution over the complex numbers instead of the reals. This makes it hard to constrain the norm of variables, as the complex conjugate is not a polynomial.

\subsection{Discussion and open questions}
Our work sheds some more light on the complexity of $\pureNrep$ and $\PureCLDM$. However, the story is far from complete as the relation between $\PureSuperQMA$ and $\QMA$ or $\PSPACE$ remains poorly understood. We conjecture \begin{conjecture}
  $\QMA \subsetneq \PureSuperQMA \subsetneq \QMAt$.
\end{conjecture}

We give some evidence that $\PureSuperQMA$ differs from $\QMAt$. Indeed, we prove that their precise versions are equal only if $\PSPACE = \NEXP$. However, this does not necessarily carry over from the precise setting to the ``standard'' setting. It would therefore be nice to see more evidence that $\PureSuperQMA \subsetneq \QMAt$, such as an oracle separation. Of course, separating $\PureSuperQMA$ from $\QMAt$ relative to an oracle is at least as hard as separating $\QMA$ from $\QMAt$ in this way, something that has been eluding researchers to this date. Perhaps, however, the new perspective offered by $\PureSuperQMA$ can lead to new insights.

Recently, it has been suggested that purity testing is at the heart of $\QMAt$'s power \cite{BFLMW24}. While we provide evidence that $\PureCLDM$ is not $\QMAt$-hard, that does not mean the end for this suggestion. One way to formalize the idea that $\QMAt$'s power derives from purity would be to prove that $\QMAt = \PureSuperQMA(\exp, 1/\poly, 1/\poly)$. Note that our results do not provide evidence against this equality, as the $\PSPACE$ upper bound crucially relies on there being only polynomially many constraints.

The relationship between $\QMA$ and $\PureSuperQMA$ with a constant number of checks could also be interesting. It can be shown that with a single check $\PureSuperQMA(1, 1/\poly, 0) = \QMA$ as there is no single check passed by a mixed state but not by a pure state\footnote{If there exists a mixed state passing the check then there exist pure states $\ket{\psi}, \ket{\phi}$ such that $\ket{\phi}$ is accepted with probability $\le 1/2$ and $\ket{\psi}$ is accepted with probability $\ge 1/2$. By interpolating between these states and the Intermediate Value Theorem it follows that there exists some pure state accepted with probability exactly $1/2$.} We have been unable to find a set of 2 constraints with a consistent mixed state but no consistent pure state but leave $\PureSuperQMA(2, 1/\poly, 0) \stackrel{?}{=} \QMA$ as an open question.

Lastly, it would be nice to see if the GP system framework used for our $\PSPACE$ upper bound can find other uses. An approach one could take here is to try to use it for a $\PSPACE$ or $\EXP$ upper bound on $\QMAt$. There are two main obstacles here. Firstly, any such approach needs to make essential use of the promise gap in order to work for $\QMAt$ but not for $\QMAt_{\cfont{exp}}$ (assuming $\PSPACE \ne \NEXP$). Secondly, naively converting a $\QMAt$ instance into polynomials yields degree 3, for which the techniques from \cite{GP05} no longer work.

\subsection{Organization}
In \cref{sec:preliminaries} we begin by giving some introductory definitions. \cref{sec:PureSuperQMA} then discusses $\PureSuperQMA$ and its relation to other complexity classes. In \cref{sec:complete} we prove our first main result: the $\PureSuperQMA$ completeness of $\PureCLDM$. \cref{sec:pureNrep} covers $\pureNrep$. \cref{sec:PSPACE,sec:approximation,sec:applications} are all concerned with the GP system framework and our $\PSPACE$ upper bound. In \cref{sec:PSPACE} we discuss our parallel algorithm for deciding if a GP system is solvable. We spend particular effort covering the methods from \cite{GP05} here, hoping that this paves the way to their further use. \cref{sec:approximation} expands upon this by giving an algorithm for OptGP. Finally, in \cref{sec:applications} we discuss some application of the GP framework.

\section{Preliminaries}\label{sec:preliminaries}
\subsection{Bosons and fermions}

Fermions are indistinguishable quantum particles whose wave function is antisymmetric under exchange of particles. It is convenient to represent them in second quantization, that is, in the occupation number basis.
An $N$-fermion state with $d$ modes can be represented in the second quantization:

\begin{equation}
  \ket{\psi} = \sum_{\subalign{n_1,\dots,n_d&\in \{0,1\}\\n_1+\dotsm+ n_d&=N}} c_{n_1,\dots,n_d} (a_1^\dagger)^{n_1}\dotsm (a_d^\dagger)^{n_d}\ket{\Omega} = \sum_{\subalign{n_1,\dots,n_d&\in \{0,1\}\\n_1+\dotsm+ n_d&=N}}c_{n_1,\dots,n_d}\ket{n_1,\dots,n_d},
\end{equation}
where $a_i,a_i^\dagger$ are the annihilation and creation operators for a fermion in mode $i$ and $\ket{\Omega}$ is the vacuum state.
The $a_i,a_i^\dagger$ satisfy the anticommutation relations $\{a_i,a_j\}=\{a_i^\dagger,a_j^\dagger\}=0$ and $\{a_i,a_j^\dagger\} =\delta_{ij}$.
Their action on a Fock state\footnote{That is, a state of occupation numbers.} is given by
\begin{align}
  a_i^\dagger\ket{n_1,\dots,n_d} &= (-1)^{\sum_{j<i}n_j}\sqrt{1-n_i}\ket{n_1,\dots,n_{i-1},n_i+1,n_{i+1},\dots,n_d}\\
  a_i\ket{n_1,\dots,n_d} &= (-1)^{\sum_{j<i}n_j}\sqrt{n_i}\ket{n_1,\dots,n_{i-1},n_i-1,n_{i+1},\dots,n_d},
\end{align}
where for fermions the occupation number will be $n_i\in\{0,1\}$ by the Pauli exclusion principle.
The $2$-fermion reduced density matrix ($2$-RDM) $\rho^{[2]} = \Tr_{3,\dots,N}(\ketbrab\psi)$ is of size $\frac{d(d-1)}2\times\frac{d(d-1)}2$ and its elements are given by
\begin{equation}\label{eq:rhoijkl}
  \rho^{[2]}_{ijkl} = \frac1{N(N-1)}\Tr\left((a_k^\dagger a_l^\dagger a_ja_i)\ketbrab\psi\right)
\end{equation}

Bosonic states are defined in the same way as fermionic states, but with the creation and annihilation operators
\begin{align}
  a_i^\dagger\ket{n_1,\dots,n_d} &= \sqrt{n_i+1}\ket{n_1,\dots,n_{i-1},n_i+1,n_{i+1},\dots,n_d}\\
  a_i\ket{n_1,\dots,n_d} &= \sqrt{n_i}\ket{n_1,\dots,n_{i-1},n_i-1,n_{i+1},\dots,n_d},
\end{align}
Bosons do not adhere to the Pauli exclusion principle, so occupation numbers are in principle unbounded.

\subsection{\texorpdfstring{$\pureNrep$}{Pure N-representability} and \texorpdfstring{$\kPureCLDM$}{kPureCLDM}}

\begin{definition}[$r$-body $\pureNrep$]\label{def:rdm}
  We are given an $r$-fermion reduced density matrix $\rho^{[r]}$ of $d$ modes with $\poly(d)$ bits of precision, the fermion number $N\le d$, as well as thresholds $\alpha,\beta$ with $\beta-\alpha\ge1/\poly(d)$. Decide:
  \begin{itemize}
    \item YES: There exists an $N$-fermion state $\ket\psi$ such that $\trnorm{\Tr_{r+1,\dots,N}(\ketbrab\psi)-\rho^{[r]}}\le\alpha$.
    \item NO: For all $N$-fermion states $\ket\psi$, $\trnorm{\Tr_{r+1,\dots,N}(\ketbrab\psi)-\rho^{[2]}}\ge\beta$.
  \end{itemize}
  As is customary we simply write $\pureNrep$ when $r = 2$.
  For $\alpha=0$, we denote the problem by $\pureNrep_1$\footnote{This notation is in reference to $\QMA_1$, the variant of $\QMA$ with perfect completeness.}.
  Define $r$-body $\Nrep$ and $\RDM_1$ analogously, but allowing a mixed state in place of $\ket{\psi}$.
\end{definition}

\begin{definition}
  Define $\BosonPureRDM$ and $\BosonRDM$ by replacing ``$N$-fermion states'' with ``$N$-boson states'' in \cref{def:rdm}.
\end{definition}

In this paper we mainly work with the more general $\kPureCLDM$ problem on qubits.

\begin{definition}[$\kPureCLDM$]\label{def:kPureCLDM}
  We are given a set of reduced density matrices $\rho_1,\dots,\rho_m$ with $\poly(n)$ bits of precision, where each $\rho_i$ acts on qubits $C_i\subseteq[n]$ with $\abs{C_i}\le k$, as well as thresholds $\alpha,\beta$ with $\beta-\alpha\ge1/\poly(n)$. Decide:
  \begin{itemize}
    \item YES: There exists a state $\ket\psi\in\CC^{2^n}$ such that $\norm{\Tr_{\overline{C_i}}(\ketbrab\psi)-\rho_i}\le\alpha$ for all $i\in[m]$.
    \item NO: For all states $\ket\psi\in\CC^{2^n}$, there exists an $i\in[m]$ such that $\norm{\Tr_{\overline{C_i}}(\ketbrab\psi)-\rho_i}\ge\beta$.
  \end{itemize}
  For $\alpha=0$, we denote the problem by $\kPureCLDM_1$.
  Define $\kCLDM$ and $\kCLDM_1$ analogously, but allowing a mixed state in place of $\ket{\psi}$.
\end{definition}

Note that $\PureCLDM$ and (mixed) $\CLDM$ are indeed different. The following example shows that a consistent mixed state may exist even if no consistent pure state exists.

\begin{example}\label{ex:inconsistent}
  Let $\rho = \frac1n\sum_{i=1}^n\ketbrab{\psi_i}$, where $\ket{\psi_i} = \ket{0^{i-1}10^{n-i}}\in \CC^{2^n}$.
  Then all $2$-local reduced density matrices of $\rho$ are $\rho_{ij} = \frac{n-2}n\ketbrab{00} + \frac1n\ketbrab{01} + \frac1n\ketbrab{10}$.
  Assume there exists a pure state $\sigma = \ketbrab{\phi}$ such that $\sigma_{ij} = \rho_{ij}$ for all $i,j\in[n]$.
  Then $\ket{\phi}\in\Span\{\ket{0^{n}},\ket{\psi_1},\dots,\ket{\psi_n}\}$ since $\rho$ has no overlap with any string of Hamming weight $\ge2$.
  Hence, $\ket{\phi}$ must be of the form $\ket{\phi} = \sum_{i=1}^n a_i\ket{\psi_i}$ with $\abs{a_i} = \sqrt{1/n}$.
  Then $\sigma_{12} = \frac2n\ketbrab{\eta} + \frac{n-2}{n}\ketbrab{00}$, where $\ket{\eta} = \sqrt{n/2}(a_1\ket{10} + a_2\ket{01})$.
  However, then $\sigma_{12}\ne\rho_{12}$ since their rank differs, which contradicts the choice of $\ket\phi$.
\end{example}

\section{PureSuperQMA}\label{sec:PureSuperQMA}
In this section we define our main complexity theoretic tool: the complexity class $\PureSuperQMA$. Following the definition we prove that $\PureSuperQMA$ contains $\PureCLDM$ and some other properties of the class, suggesting it can be thought of as a ``light'' version of $\QMAt$. We end the section by establishing a canonical form for $\PureSuperQMA$, cleaning up the definition. This canonical form will be important in the next section, where we use it for the hardness proof.

Aharanov and Regev \cite{AR03} define a variant of $\QMA$, which we call $\SuperQMA$, with a ``super-verifier'', which is a classical randomized circuit that is given access to the input $x$ and outputs a description of a quantum circuit $V$ and two numbers $r,s\in[0,1]$.
An honest prover then needs to send a state $\rho$ such that $\Pr_{V,r,s}(\Pr(V\text{ accepts }\rho)\in[r-s,r+s])=1$, where the outer probability is over the randomness of the circuit.

\begin{definition}[{\SuperQMA \cite{AR03}}]
  A promise problem $A$ is in $\SuperQMA$ if there exists a super-verifier and polynomials $p_1,p_2,p_3$ such that
  \begin{itemize}
    \item $\forall x\in \Ayes\; \exists \rho\colon \Pr_{V,r,s}(\abs{\Tr(\Piacc V\rho V^\dagger) - r}\le s)=1$,
    \item $\forall x\in \Ano \; \forall \rho\colon \Pr_{V,r,s}(\abs{\Tr(\Piacc V\rho V^\dagger) - r}\le s + 1/p_3(\abs{x}))\le 1-1/p_2(\abs{x})$,
  \end{itemize}
  where probabilities are taken over the output of the super-verifier and $\rho$ is a density matrix on $p_1(\abs{x})$ qubits.
\end{definition}

\begin{proposition}[\cite{AR03}]
  $\SuperQMA = \QMA$.
\end{proposition}

We define the pure state analog $\PureSuperQMA$.
Note that it is not obvious at all whether $\PureSuperQMA\subseteq\SuperQMA$.
Also note that \SuperQMA can essentially perform an exponential number of different measurements.
However, our techniques only apply for a polynomial number of measurements.

\begin{definition}[{\PureSuperQMA}]
  \label{def:PureSuperQMA}
  A promise problem $A$ is in $\PureSuperQMA(m,\epsilon,\delta)$ if there exists a uniformly generated \emph{super-verifier}\footnote{Our definition of a super verifier here might seem slightly different from Aharanov and Regev's definition. Note that their definition can be recovered from our by repeating the same constraint as necessary.} $\calV = \{(V_{x,i},r_{x,i},s_{x,i})\}_{i\in [m]}$ on $n_1(n)\in n^{O(1)}$ proof qubits and $n_2(n)\in n^{O(1)}$ ancilla qubits ($n=\abs{x}$) such that
  \begin{itemize}
    \item $\forall x\in \Ayes\; \exists \ket{\psi}\in \calP \colon \Pr_i(\abs{p(V_{x,i},\psi) - r_{x,i}}\le s_{x,i})=1$,
    \item $\forall x\in \Ano\; \forall \ket{\psi}\in \calP\colon \Pr_i(\abs{p(V_{x,i},\psi)- r_{x,i}}\le s_{x,i} + \epsilon) \le 1-\delta$,
  \end{itemize}
  where $\calP$ is the set of unit vectors on $n_1(n)$ qubits, $i\in[m]$ is drawn uniformly at random, and 
  \begin{equation}
    p(V,\psi)\coloneq\Tr(\Piacc V\ketbrab{\psi,0^{n_2}} V^\dagger)
  \end{equation}
   denotes the acceptance probability of $V$ on input $\ket{\psi}$.\footnote{Here and in the following, we use $\psi$ to denote the density operator $\ketbrab\psi$.}
  We call each triple $(V_{x,i},r_{x,i},s_{x,i})$ a \emph{constraint} or a \emph{check}.
  Let
  \begin{equation}
    \PureSuperQMA = \bigcup_{m\in n^{O(1)},\epsilon,\delta\in n^{-O(1)}}\PureSuperQMA(m,\epsilon,\delta).
  \end{equation}
\end{definition}

\begin{lemma}\label{lem:kPureCLDM in PureSuperQMA}
  $\kPureCLDM\in\PureSuperQMA$.
\end{lemma}
\begin{proof}
  The result follows analogously to the containment of $\pureNrep$ in $\QMAt$ \cite{LCV07}.
  Let $\rho_1,\dots,\rho_m$ on qubits $C_j\subseteq[n]$ be a given $\kPureCLDM$ instance.
  Note, the generalization to qu$d$its is straightforward via embedding into qubits.
  Let $\wtrho_i = \Tr_{\overline{C_i}}(\ketbrab\psi)$ be the reduced density matrix on qubits $C_i$ of an $n$-qubit state $\ket\psi$.
  We need to construct a super-verifier to verify $\wtrho_j = \rho_j$ for all $j\in [m]$.

  Let $j\in [m]$.
  We use \emph{quantum state tomography}~\cite{NC10} to verify $\wtrho_j=\rho_j$.
  It suffices to verify that $\wta_{j,w} := \Tr(P_w \wtrho_j) = \Tr(P_z\rho_j) =: a_{j,w}$ for all $w\in \{0,1,2,3\}^k =: \calW$, where $P_w := 2^{-k/2}\bigotimes_{i=1}^k \sigma_{w_i}$ and $\sigma_0=I,\sigma_1=\Xgate,\sigma_2=\Ygate,\sigma_3=\Zgate$.
  This holds because the $P_w$ form an orthonormal basis with respect to the Hilbert-Schmidt inner product~\cite{NC10}. Note that the $a_{j,w}$ can be computed classically because the reduced density matrices are small.

  The super-verifier $\calV$ for \PureSuperQMA then selects a random $j\in[m]$, $w\in \calW$ and outputs circuit $V_{j,w}$ that performs the Pauli measurement $P_w$ on qubits $C_j$ and accepts on outcome $+1$.
  Thus, $p(V_{j,w},\psi) = 1/2 + 2^{k/2-1}\wta_{j,w}$.
  We set the corresponding $s_{j,w} := 1/\exp(n)$\footnote{We do not set $s_{j,w}=0$ because there are potential issues regarding the exact representation of the circuits and probabilities.} and $r_{j,w} := 1/2 + 2^{k/2-1}a_{j,w}$.
  Completeness is obvious as then $\wta_{j,w} = a_{j,w}$.
  Soundness follows because the $\ell^2$-norm of ``errors'' is proportional to the Frobenius distance of $\rho_j$ and $\wtrho_j$:
  \begin{equation}
    \fnorm{\wtrho_j - \rho_j}^2 = \sum_{w\in \calW}\Tr(P_w (\wtrho_j-\rho_j))^2 = \sum_{w\in \calW}(\wta_{j,w}-a_{j,w})^2
  \end{equation}
\end{proof}
We analogously get:
\begin{lemma}\label{lem:pureNrep in PureSuperQMA}
  $\pureNrep\in\PureSuperQMA$.
\end{lemma}

\PureSuperQMAinQMAtwo*
\begin{proof}[Proof sketch]
  First note that $\QMA \subseteq \PureSuperQMA$ follows by setting $r = 1$ and $s = \frac{1}{\exp(n)}$ in the definition of $\PureSuperQMA$. For the other inclusion we use the fact that $\QMAt = \QMA(\poly)$ \cite{HM13}.
  With probability $1/2$ each, the verifier performs one of the following tests:
  (i) Run swap tests between random disjoint pairs of the registers to ensure the input state is close to a state of the form $\ket{\psi}^{\otimes k}$.
  (ii) Pick $i\in[m]$ uniformly at random and run $V_{x,i}$ on all $k$ proofs, recording the outcomes as $y_1,\dots, y_k\in \{0,1\}$, and let $\mu = \frac1k\sum_{i=1}^k y_k$.
  Accept if $\abs{\mu-r_{x,i}}\le s_{x,i} + \epsilon/2$.
  For sufficiently large $k\in n^{O(1)}$, we can use Hoeffding's inequality to prove completeness and soundness.
\end{proof}

Note that we still get containment in $\QMAt$, even with an exponential number of constraints, as long as in the NO-case the acceptance probability deviates from $r_{x,i}$ for a significant fraction of constraints.
However, we do have some reason to believe that $\kPureCLDM$ is only hard for $\PureSuperQMA$ with a polynomial number of constraints since the hardness proof works for any precision parameter and $\kPureCLDM \in \PSPACE$ (see \cref{thm:parallelalgforGP}) even for exponentially small precision.
In contrast, if we have both an exponential number of constraints and exponential precision, then $\PureSuperQMA$ contains $\NEXP$.

\logsizeproofs*
\begin{proof}[Proof sketch]
  We will show that $\PureSuperQMA$ with $\log$-size proofs contains the $\NP$ complete $3$-coloring problem.

  Let $G = (V,E)$ be a graph on $n$ vertices. Following \cite{BT12} we want to force the prover to send a proof state of the form:
  \begin{equation}
    \ket{\psi} = \frac{1}{\sqrt{n}}\sum_{v \in V} \ket{v}_V\ket{C(v)}_C
  \end{equation}
  where $C$ is a valid $3$-coloring of $G$. To ensure that the proof indeed has this structure we define a super-verifier using the following checks:
  \begin{enumerate}[label=(\roman*)]
    \item $$\forall v\in V\colon \Tr((\ketbrab{v}_V\otimes I_C)\ketbrab{\psi})=1/n.$$
    \item $$\forall v\in V\;\forall \ket\phi\in S\colon \Tr((\ketbrab{v}_V\otimes \ketbrab{\phi}_C)\ketbrab{\psi})=1/3n,$$ where $S = \{\ket{0} + \ket{1} + \ket{2}, \ket0+\ket1-\ket{2},\ket0+\iu\ket1+\iu\ket{2},\ket{0}+\ket{1}+\iu\ket{2}\}/\sqrt{3}$.
    \item $$\forall \{u,v\}\in E\; \forall\ket{\phi_{uv}'} \in S_{uv}'\colon \Tr((\ketbrab{\phi_{uv}'}_V\otimes I_C)\ketbrab{\psi})=1/n,$$ where $S_{uv}' = \{(\ket{u}+\ket{v})/\sqrt{2}, (\ket{u} + i \ket{v})/\sqrt{2}\}$.
  \end{enumerate}
  For completeness, note that the desired proof $\ket{\psi}$ passes all tests perfectly.

  For soundness, let $\ket{\psi} = \sum_{v \in V} \alpha_v \ket{v}_V \sum_{c = 0}^2 \beta_{v,c} \ket{c}_C$ be the given proof. Without loss of generality we can assume that $\sum_{c = 0}^2 |\beta_{v,c}|^2 = 1$ and $\sum_{c} \beta_{v,c} \in \RR^+$ for all $v$ by absorbing the modulus and phase into the $\alpha_v$. 

  Note that any $\ket{\psi}$ that passes check (i) with error $\epsilon$ must have 
  \begin{equation}
    \label{eqn:checki}
    \forall v\in V\colon \quad |\alpha_v|^2 \in \left[\frac{1}{n} - \epsilon, \frac{1}{n} + \epsilon\right].
  \end{equation}
  Observe that any $\psi$ that passes (i) and (ii) exactly satisfies the following equations (over the reals) for all $v$:
  \begin{subequations}
    \begin{align}
      a_{v0}^2 + b_{v0}^2 + a_{v1}^2 + b_{v1}^2 + a_{v2}^2 + b_{v2}^2 &= 1 \label{eqa}\\ 
      a_{v0} + a_{v1} + a_{v2} &= 1 \label{eqb}\\
      b_{v_0} + b_{v1} + b_{v2} &= 0 \label{eqc}\\
      (a_{v0} + a_{v1} - a_{v2})^2 + (b_{v0} + b_{v1} - b_{v2})^2 &= 1 \label{eqd}\\
      (a_{v0} - b_{v1} - b_{v2})^2 + (b_{v0} + a_{v1} + a_{v2})^2 &= 1 \label{eqe}\\
      (a_{v0} + a_{v1} - b_{v2})^2 + (b_{v0} + b_{v1} + a_{v2})^2 &= 1. \label{eqf}
    \end{align}
  \end{subequations}
  Here we have written $\beta_{vj} = a_{vj} + b_{vj} i$ for $j = 0,1,2$. \cref{eqa,eqb,eqc} follow from our assumptions combined with the constraint where $\ket{\phi} = (\ket{0} + \ket{1} + \ket{2})/\sqrt{3}$. The other three equations follow from the other constraints. It can be verified (for example with a computer algebra system) that the only solutions to the system are $(\beta_{v0}, \beta_{v1}, \beta_{v2}) \in \{(1,0,0), (0,1,0), (0,0,1)\}$. By the \L{}ojasiewicz inequality \cite{Loj65,BM88} any $\ket{\psi}$ that passes these check with error at most $\epsilon$ has $(\beta_{v0}, \beta_{v1}, \beta_{v2})$ at distance at most $\poly(\epsilon)$ from these solutions for all $v$.
  
  The constraints on $\ket{\psi}$ imposed by (iii) can be written as
  \begin{subequations}\label{eqn:checkiii}
    \begin{align}
      \quad \frac{1}{2}\left(|\alpha_u|^2 + |\alpha_v|^2 + \alpha_u\alpha_v^* \sum_c \beta_{vc}\beta_{uc}^* + \alpha_v\alpha_u^*\sum_c \beta_{uc}\beta_{vc}^*\right) &= \frac{1}{n} \\
      \frac{1}{2}\left(|\alpha_u|^2 + |\alpha_v|^2 + i\alpha_u\alpha_v^* \sum_c \beta_{vc}\beta_{uc}^* - i\alpha_v\alpha_u^*\sum_c \beta_{uc}\beta_{vc}^*\right) &= \frac{1}{n}.
    \end{align}
  \end{subequations}
  Assume these checks are also passed with error at most $\epsilon$ and assume that $G$ has no valid $3$-coloring. Then, by (ii), there must by $u', v'$ such that $|\sum_c \beta_{vc}\beta_{uc}^* - 1| \le \poly(\epsilon)$ and $|\sum_c \beta_{uc}\beta_{vc}^* - 1| \le \poly(\epsilon)$.\footnote{By check (ii), for all $v$ there is exactly one $c$ such that $\beta_{vc}$ is close to $1$ and for the other $c$, $\beta_{vc}$ is close to $0$. Think of the $c$ where $\beta_{vc}$ is close to $1$ as the color assigned to $v$. Since $G$ is assumed to be not $3$-colorable, there has to be some edge $(u,v)$ such that $\beta_{uc}$ and $\beta_{vc}$ are close to for the same $c$.} By (i) and \cref{eqn:checkiii} we get $|\alpha_u\alpha_v^* + \alpha_v \alpha_u^*| \le \poly(\epsilon)$ and $|\alpha_u\alpha_v^* - \alpha_v \alpha_u^*| \le \poly(\epsilon)$. Summing we get $|\alpha_u||\alpha_v| \le poly(\epsilon)$, which contradicts \cref{eqn:checki} for sufficiently small $\epsilon$. We have thus proven that, if $G$ is not $3$-colorable, at least one of the checks must fail with probability at least $\epsilon$, establishing soundness and completing the proof.
\end{proof}

\expPSQMAisNEXP*
\begin{proof}[Proof sketch]
  The proof is analogous to the proof of \cref{logsizeproofs} except using the $\NEXP$-complete succinct $3$-coloring \cite{Pap86} instead of normal $3$-coloring. As the number of vertices is now exponential, we need a proof of polynomial size, an exponential number of checks, and exponential precision. 
\end{proof}

Since $\PureSuperQMA(\exp,1/\exp,1/\exp) \subseteq \QMAtexp$ (i.e. \QMAt with exponentially small promise gap) by the same argument as \Cref{prop:PureSuperQMA in QMA(2)}, we also recover the following known result:

\begin{corollary}[\cite{BT12,Per12}]
  $\QMAtexp = \NEXP$, where $\QMAtexp$ denotes $\QMAt$ with exponentially small promise gap.
\end{corollary}

To prove that $\tPureCLDM$ is \PureSuperQMA-complete, we will use a simplified but equivalent definition:

\begin{definition}[{\PSQMA}]
  A promise problem $A$ is in $\PSQMA(m,\epsilon)$ if there exists a uniformly generated super-verifier $\calV = \{V_{x,i}\}_{i\in [m]}$ on $n_1(n)\in n^{O(1)}$ proof qubits and $n_2(n)\in n^{O(1)}$ ancilla qubits ($n=\abs{x}$) such that
  \begin{itemize}
    \item $\forall x\in \Ayes\; \exists \ket{\psi}\in \calP\; \forall i\in [m]\colon \abs*{\Tr(\Piacc V_{x,i}\ketbrab{\psi,0^{n_2}} V_{x,i}^\dagger) - \frac12}=0$,
    \item $\forall x\in \Ano\; \forall \ket{\psi}\in \calP\; \exists i\in[m]\colon \abs*{\Tr(\Piacc V_{x,i}\ketbrab{\psi,0^{n_2}} V_{x,i}^\dagger) - \frac12}\ge \epsilon$,
  \end{itemize}
  where $\calP$ is the set of unit vectors on $n_1(n)$ qubits.

  Let $\PSQMA = \bigcup_{m\in n^{O(1)},\epsilon\in n^{-O(1)}}\PSQMA(m,\epsilon,\delta)$.
\end{definition}
In other words, $\PSQMA$ is $\PureSuperQMA$ where all the $r$'s are set to $\frac{1}{2}$ and all the $s$'s to 0.

\begin{lemma}\label{lem:PSQMA=PureSuperQMA}
  $\PSQMA = \PureSuperQMA$.
\end{lemma}
\begin{proof}
  $\PSQMA \subseteq \PureSuperQMA$ holds by definition.
  Now let $A\in \PureSuperQMA(m,\epsilon,\delta)$ with super-verifier $\calV = \{(V_{x,i}, r_{x,i}, s_{x,i})\}_{i\in [m]}$ and $n_1(n)$ proof qubits and $n_2(n)$ ancilla qubits.
  We can assume without loss of generality that $\delta=1/m$, $s_{x,i}\le \frac12 -\epsilon$, and $r_{x,i} = 1/2$ for all $x,i$.
  To see the latter, let $(V,r,s) = (V_{x,i}, r_{x,i}, s_{x,i})$ be a constraint of $\calV$.
  We can transform $(V,r,s)$ to $(V',r',s')$ with $r':=1/2$ and $s':=s/2$.
  Let $V'$ be the verifier that with probability $1/2$ each performs one of the following actions: (i) accept with probability $1-r$. (ii) run $V$ on the input state.
  Then
  \begin{equation}
    p(V',\psi) = \frac{1-r}2 + \frac{p(V,\psi)}2 = \frac12 + \frac{p(V,\psi)-r}2.
  \end{equation}
  Thus $\abs{p(V',\psi)-r'} \le s'$ iff $\abs{p(V,\psi)-r} \le s$.

  The next step is to show $A\in\PSQMA(2m+1,\epsilon')$ by constructing a super-verifier $\calV' = \{(V'_{x,i},\frac12,0)\}_{i\in[m]}\cup \{(W_i,\frac12,0)\}_{i=0}^m$, where the $V'_{x,i}$ constraints correspond to the original $V_{x,i}$ constraints, and the $W_i$ enforce that the proof is given in uniform superposition with ``slack variables'', which are used to achieve $s=0$.
  An honest proof shall be of the form
  \begin{equation}\label{eq:psi'}
    \ket{\psi'} = \frac{1}{\sqrt{m+1}}\left(\ket{0}_A\ket{\psi}_B + \sum_{i=1}^m \ket{i}_A\left(\sqrt{p_i}\ket{1}_{B_1}\ket{\eta_{i,1}}_{B'} + \sqrt{1-p_i}\ket{0}_{B_1}\ket{\eta_{i,0}}_{B'}\right)\right),
  \end{equation}
  where $p_i\in [0,1]$ and $A$ denotes the ``index register'' and $B=B_1B'$ the proof register of $n_1$ qubits.

  For $i=0,\dots,m$, define $W_i$ as the verifier that with probability $1/2$ each, performs one of the following actions:
  (i) accept with probability $1-1/(m+1)$.
  (ii) measure with projector $\Pi_i = \ketbrab{i}_A\otimes I_{B}$.
  Then
  \begin{equation}\label{eq:p(Wi)}
    \abs*{p(W_i,\psi')-\frac12} = \frac12\abs*{\bra{\psi'}\Pi_i\ket{\psi'}-\frac1{m+1}}.
  \end{equation}
  Hence, the $W_i$ constraints ensure $\ket{\psi'}$ is of the form \cref{eq:psi'}.
  Note that $W_i$ can be implemented exactly with the ``Clifford + Toffoli'' universal gateset if $m+1=2^{n_3}$ for some $n_3\in\NN$, which we may assume without loss of generality.

  Next, we define the $V_{x,i}'$ constraints:
  \begin{enumerate}[label=(\roman*)]
    \item Measure the $A$ register in standard basis and denote the outcome by $j\in\{0,\dots,m\}$.
    \item If $j=0$, run $V_{x,i}$ on the $B$ register.
    \item If $j=i$, measure qubit $B_1$ and denote the outcome by $b\in\{0,1\}$. Accept with probability $1/2 + \wts_{x,i}(2b-1)$, where $\wts_{x,i} \in [s_{x,i},s_{x,i}+\epsilon/2]$ such that $\wts_{x,i} = s/2^t$ for some $s,t\in\NN$ with $s >0$ and minimal $t$.
    \item Otherwise, accept with probability $1/2$.
  \end{enumerate}
  Again, we can implement $V_{x,i}'$ exactly with the ``Clifford + T'' gateset.
  Then
  \begin{equation}\label{eq:p(V'xi)}
    p(V'_{x,i},\psi') = \frac1{m+1}p(V_{x,i},\psi) + \frac1{m+1}\left(\frac12 + \left(2p_i-1\right)\wts_{x,i}\right) + \frac{m-1}{m+1}\cdot \frac12.
  \end{equation}

  \emph{Completeness:}
  Let $x\in \Ayes$, then there exists a state $\ket{\psi}$ such that $\abs{p(V_{x,i},\psi)-1/2} \le s_{x,i}$ for all $i\in[m]$.
  Let $\ket{\psi'}$ be as in \cref{eq:psi'} with
  \begin{equation}
    p_i = \frac12 + \frac{\frac12-p(V_{x,i},\psi)}{2\wts_{x,i}}.
  \end{equation}
  Thus,
  \begin{equation}
    p(V'_{x,i},\psi')-\frac12 = \frac1{m+1}\left(p(V_{x,i},\psi)-\frac12 + \left(2p_i-1\right)\wts_{x,i}\right) =  0,
  \end{equation}
  and the constraints $W_i$ are trivially satisfied.
  Hence, all constraints of the super-verifier are satisfied.

  \emph{Soundness:}
  Let $x\in\Ano$ and consider some proof $\ket{\psi'}\in \CC^{2^{n_1+n_3}}$.
  We can write 
  \begin{equation}
    \ket{\psi'} = \sqrt{a_0}\ket{0}_A\ket{\psi}_B + \sum_{i=1}^m \sqrt{a_i}\ket{i}_A\left(\sqrt{p_i}\ket{1}_{B_1}\ket{\eta_{i,1}}_{B'} + \sqrt{1-p_i}\ket{0}_{B_1}\ket{\eta_{i,0}}_{B'}\right),
  \end{equation}
  where for all $i$, $p_i\in [0,1]$, $a_i\ge0$, and $\ket{\psi},\ket{\eta_{i,1}},\ket{\eta_{i,2}}$ are unit vectors.
  Suppose $\abs{p(W_i,\psi')-\frac12} \le \epsilon'$ for $i=0,\dots,m$.
  Then by \cref{eq:p(Wi)}, $a_i = \frac{1}{m+1} + \epsilon_i$ for $\abs{\epsilon_i}\le 2\epsilon'$.
  Since $x\in \Ano$, there exists $j\in[m]$ such that $\abs{p(V_{x,j},\psi)-\frac12} \ge s_{x,j}+\epsilon \ge \wts_{x,j} + \epsilon/2$.
  Analogously to \cref{eq:p(V'xi)}, we have
  \begin{equation}
    p(V'_{x,j},\psi') = a_0 \cdot p(V_{x,j},\psi) + a_1 \cdot \left(\frac12 + \left(2p_j-1\right)\wts_{x,j}\right) + (1-a_0-a_1)\cdot \frac12.
  \end{equation}
  Thus,
  \begin{subequations}
  \begin{align}
    \abs*{p(V'_{x,j},\psi') - \frac12} &= \abs*{a_0\left(p(V_{x,j},\psi)-\frac12\right) + a_1\left(2p_j-1\right)\wts_{x,j}}\\
    &\ge \abs*{\frac1{m+1}\left(p(V_{x,j},\psi)-\frac12\right) + \frac1{m+1}\left(2p_j-1\right)\wts_{x,j}} - 2\epsilon'\\
    &\ge \frac{\wts_{x,j}+\frac{\epsilon}2 - \wts_{x,j}}{m+1} - 2\epsilon' = \frac{\epsilon}{2(m+1)} - 2\epsilon'\ge \epsilon',
  \end{align}
  where the final inequality holds by setting $\epsilon' := \epsilon/6(m+1)$.
  Therefore, either one of the $W_i$ constraints is violated or one of the $V'_{x,i}$ constraints for $x\in\Ano$.
  We conclude $A\in \PSQMA$.
  \end{subequations}
\end{proof}

\section{2-PureCLDM is PureSuperQMA-complete}\label{sec:complete}
In this section we prove our first main result:

\tPureCLDMcomplete*

Containment holds by \cref{lem:kPureCLDM in PureSuperQMA}.
It remains to prove that $\kPureCLDM_1$ is $\PureSuperQMA=\PSQMA$-hard (\cref{lem:PSQMA=PureSuperQMA}).
For simplicity, we only prove hardness of $\tPureCLDM_1$, but $k>2$ is analogous since simulatability works for any $k$ (see \cref{claim:sim}).

Our proof closely follows the proof for the \QMA-completeness of the (mixed) \CLDM problem by Broadbent and Grilo \cite{BG22} based on \emph{locally simulatable codes}. The broad idea behind their proof is as follows: starting with a $\QMA$ verifier $V$ we can use Kitaev's circuit-to-Hamiltonian construction to build a Hamiltonian $H$ that has low-energy state if and only if $V$ would accept. This low-energy state, the history state, encodes the behavior of the verifier. Broadbent and Grilo's aim is to construct local density matrices that are consistent with the history state if $V$ accepts, and are inconsistent otherwise. 

The main difficulty is that the behavior of the verifier, and hence the history state, depend on the proof given to the verifier. This is solved by considering a different verifier $V'$ that mirrors the actions of $V$ but acts on \emph{encoded} data. By choosing an $s$-simulatable code for this encoding this has as crucial advantage that the $s$-local density matrices at any point during the computation can be computed efficiently.\footnote{Note that these no longer depend on the proof. Intuitively, all information contained in the proof has been moved to the correlations between larger sets of qubits.} 

One technical detail that will be especially important to us is that the new verifier $V'$ needs to check whether the proof it received is properly encoded. To ensure that the local density matrices can still be computed during this checking procedure, Broadbent and Grilo have the prover send the original proof encoded under a one-time pad and the one-time pad keys. It is this one-time pad constructions that causes the Broadbent Grilo proof to fail for $\PureCLDM$ as a cheating prover could send different pure states for each key, making the full state behave like a mixed state.

Broadly, our proof of \cref{thm:2PureCLDM-complete} and the modifications made to the Broadbent Grilo proof can be described as follows:
\begin{itemize}
  \item Starting off, we consider a problem $A \in \PSQMA$ that is decided by a super-verifier $\calV = \{V_{x,i}\}_{i\in[m]}$.
  \item First, we define a modified super-verifier $\Votp$ that accepts a one-time padded proof. A crucial difference with Broadbent and Grilo is that we use the same key to encrypt every qubit. The expected proof will be of the form $\ket{\psiotp} = \frac{1}{2}\sum_{a,b \in \{0,1\}} (X^aZ^b)^{\otimes n_1} \ket{\psi_{ab}}\ket{abab}$ where the $\ket{\psi_{ab}}$ are proofs accepted by $\calV$. Every individual qubit will still be maximally mixed, but there are only $4$ different keys. This allows the modified super-verifier to check, for each of the possible keys, that $\ket{\psi_{ab}}$ is indeed accepted by $\calV$. 
  \item Next, we want to combine all checks in $\Votp$ in one circuit, that we can then apply a circuit-to-Hamiltonian construction to. We simply chain all checks, followed by their inverses, together as $V_{x,1}V_{x,1}^\dagger \dots V_{x,m}V_{x,m}^\dagger$. Following Broadbent and Grilo we also modify the circuit to make it act on data encoded by a $s$-simulatable code. This will allow us to compute the $2$-local density matrices of the history state.
  \item Unlike Broadbent and Grilo, we use the 2-local circuit-to-Hamiltonian construction from \cite{KKR06} to obtain a Hamiltonian $H$.\footnote{Broadbent and Grilo use Kitaev's original 5-local construction.} We would like to force any state consistent with the output of the reduction to be close to a history state. To achieve this, we have the reduction compute the energy of the reduced density matrices it would output (this is possible since $H$ is $2$-local). If this energy is too high, the reduction simply outputs a trivial ``No'' instance\footnote{We would like to reject here, but as we are describing the reduction we cannot directly do this. Instead, we output a trivial ``No'' instance which has the same effect.}. If the energy is sufficiently small, we can show that any consistent state must be close to a history state as desired. 
  \item The way we combined the checks might seem puzzling: $V_{x,i}V_{x,i}^\dagger$ acts as the identity after all. The trick is that we will read of the acceptance probability of $V_{x,i}$ from the history state of the combined circuit. We use a technical tool we call the \emph{Extraction Lemma}, which will allow extracting $1$-local density matrices at arbitrary points during the execution of the combined circuit from the $2$-local density matrices \emph{of the history state}, provided $V_{x,i}$ ends by decoding its output qubit from the simulatable code into a single qubit.
  \item We have now assembled all the pieces: the reduction computes $H$, the $2$-local density matrices of the history state, and checks the energy of these matrices. If they are indeed low-energy \emph{and} all checks pass, then there will be a consistent state (the history state). If the $2$-local density matrices are low-energy \emph{but} not all checks pass, then we can use the extraction lemma to conclude that there is no consistent state, exactly as desired.
\end{itemize}

For the remainder of this section, let $A\in\PSQMA(m,\epsilon)$ be decided by a super-verifier $\calV = \{V_{x,i}\}_{i\in[m]}$ on $n_1$ proof qubits and $n_2$ ancilla qubits.

\subsection{Super-verifier with locally maximally mixed proof}\label{sec:otp}

As the first step we create a modified super-verifier $\calVotp = \{\Votp_{x,i}\}_{i\in[m']}$ for $A$, such that in the YES case, there exists a proof $\ket{\psiotp}$ whose single qubit reduced density matrices are maximally mixed.
The $\Votp_{x,i}$ will act on $n_1' = n_1+4$ proof qubits and $n_2'\ge n_2$ ancilla qubits.
The expected proof is of the form
\begin{equation}\label{eq:psiotp}
  \ket{\psiotp} = \frac12 \sum_{a,b\in\{0,1\}} (X^aZ^b)^{\otimes n_1}\ket{\psi}\ket{abab},
\end{equation}
where $\ket\psi$ is an accepted proof for the original super-verifier.
\begin{restatable}{claim}{claimPsiotp} 
  $\Tr_{\overline{j}}(\psiotp)=I/2$ for all $j\in[n_1']$.
\end{restatable}
\begin{proof}
  See \cref{sec:omit}.
\end{proof}
\noindent
$\calVotp$ consists of the following $m' = 4m+4$ constraints, which are ``normalized'' as in \cref{lem:PSQMA=PureSuperQMA} to have a target acceptance probability of $1/2$:
\begin{enumerate}
  \item For all $a,b\in\{0,1\}$ and $i\in[m]$: $\Tr((\Piacc\otimes\ketbrab{abab})V_{x,i}(\psiotp\otimes\ketbra00^{\otimes n_2'})V_{x,i}^\dagger)=1/8$. After normalizing we get for $i' = 4i + 2a + b - 3$:
  \begin{equation}\label{eq:Votp1}
    p(\Votp_{x,i'},\psiotp) = \frac{7}{16} + \frac{p(V_{x,i}, (I\otimes\bra{abab})\psiotp(I\otimes\ket{abab}))}2
  \end{equation}
  \item For all $a,b\in\{0,1\}$: $\Tr((I\otimes\ketbrab{abab})\psiotp)=1/4$. Normalizing then yields for $i' = 4m + 1 + 2a + b$:
  \begin{equation}\label{eq:Votp2}
    p(\Votp_{x,i'},\psiotp) = \frac38 + \frac{\Tr((I\otimes\ketbrab{abab})\psiotp)}{2}
  \end{equation}
\end{enumerate}

\begin{claim}
  $\calVotp$ is a $\PSQMA(m',\epsilon')$ super-verifier for $A$ with $\epsilon'=\epsilon/16$.
\end{claim}
\begin{proof}
  If $x\in \Ayes$, let $\ket\psi$ be a proof that perfectly satisfies each constraint of $\calV$.
  Then it is easy to see that $\ket{\psiotp}$ as defined in \cref{eq:psiotp} satisfies the constraints of $\calVotp$ from \cref{eq:Votp1,eq:Votp2}.

  Let $x\in\Ano$ and $\ket{\psiotp}$ be an arbitrary proof.
  We argue that there exists $i\in[m']$ such that $\abs{p(\Votp_{x,i},\psiotp)-1/2} \ge \epsilon'$.
  Without loss of generality, we can write
  \begin{equation}
    \ket{\psiotp} = \frac12 \sum_{a,b\in\{0,1\}} \sqrt{p_{ab}}(X^aZ^b)^{\otimes n_1}\ket{\psi_{ab}}\ket{abab} + \sqrt{p^\perp}\ket{\psi^\perp},
  \end{equation}
  where $\sum_{ab}p_{ab}+p^\perp = 1$ and $(I\otimes\ketbrab{abab})\ket{\psi^\perp}=0$ for all $a,b\in\{0,1\}$.
  Assume $\abs{p(\Votp_{x,i},\psiotp)-1/2} < \epsilon'$ for all $i\in\{4m+1,\dots,4m+4\}$ (i.e. the constraints of \cref{eq:Votp2} are approximately satisfied).
  Hence, $\abs{p_{ab}-1/4}\le 2\epsilon'$ for all $a,b\in\{0,1\}$.
  Fixing some $a,b\in\{0,1\}$, there must exist $j\in[m]$ such that $\abs{p(V_{x,j},\psi_{ab})-1/2}\ge \epsilon$.
  By \cref{eq:Votp1}, we have for
  \begin{equation}
    \begin{aligned}
      \abs*{p(\Votp_{x,j'},\psiotp)-\frac12} &= \abs*{\frac{p(V_{x,j}, (I\otimes\bra{abab})\psiotp(I\otimes\ket{abab}))}2 -\frac1{16} }
      = \abs*{\frac{p_{ab}\cdot p(V_{x,j},\psi_{ab})}{2}-\frac1{16}} \\
      &\ge \left\{ \begin{aligned}
        \frac{(\frac14-2\epsilon')(\frac12+\epsilon)}{2}-\frac1{16} & \quad\text{if } p(V_{x,j},\psi_{ab})>1/2\\
        \frac1{16}-\frac{(\frac14+2\epsilon')(\frac12-\epsilon)}{2} & \quad\text{if } p(V_{x,j},\psi_{ab})<1/2
      \end{aligned} \right\}
      =\frac\epsilon8-\frac{\epsilon'}{2}-\epsilon\epsilon' \ge \epsilon',
    \end{aligned}
  \end{equation}
  assuming $\epsilon\le 1/2$ (recall $\epsilon' = \frac{\epsilon}{16}$).
\end{proof}

\subsection{$2$-local Hamiltonian}\label{sec:2-local}
We will now establish the required properties of the $2$-local circuit-to-Hamiltonian construction. We show that any state consistent with low-energy local density matrices must be close to a history state and prove the Extraction Lemma. For the former proof, we need the following modification of the Projection Lemma. 

\begin{lemma}[State Projection Lemma, cf. \cite{GY19}]\label{lem:state-projection}
  Let $H = H_1 + H_2$ be the sum of two Hamiltonians operating on some Hilbert space $\H = \calS + \calS^\perp$, where $\calS$ is the nullspace of $H_2$ and the other eigenvalues are at least $J$.
  Let $\rho$ be a state in $\H$ such that $\Tr(H\rho) \le 1$.
  Then there exists a state $\sigma$ (pure if $\rho$ is pure) in $\calS$, such that $\trnorm{\rho-\sigma} \le \delta$ and $\Tr(H\sigma) \le 1 + \norm{H_1}\delta$, for $\delta = 2\sqrt{(1+\norm{H_1})/J}$.
\end{lemma}
\begin{proof}
  Let $\Pi_{\calS},\Pi_{\calS^\perp}$ be the projectors onto $\calS,\calS^\perp$.
  We have 
  \begin{subequations}
  \begin{align}
    1&\ge \Tr(H\rho) = \Tr(H_1\rho) + \Tr(H_2(\Pi_{\calS} + \Pi_{\calS^\perp})\rho (\Pi_{\calS} + \Pi_{\calS^\perp})) \ge -\norm{H_1} + \Tr(H_2 \Pi_{\calS^\perp}\rho\Pi_{\calS^\perp}) \\
    &\ge J\cdot\Tr(\Pi_{\calS^\perp}\rho) - \norm{H_1}\\
    \Rightarrow\;& \Tr(\Pi_{\calS^\perp}\rho) \le \frac{1+\norm{H_1}}J =: \epsilon.
  \end{align}
  \end{subequations}
  Thus $\Tr(\Pi_{\calS}\rho)\ge 1-\epsilon$.
  Let $\sigma = \Pi_{\calS}\rho\Pi_{\calS}/\Tr(\Pi_{\calS}\rho)$.
  By the Gentle Measurement Lemma \cite[Lemma 9.4.1]{Wil13}, we have $\trnorm{\rho-\sigma} \le 2\sqrt{\epsilon}=\delta$.
  By the Hölder inequality,
  \begin{equation}
      \Tr(H\sigma) = \Tr(H_1\sigma) = \Tr(H_1\rho) + \Tr(H_1(\sigma-\rho)) \le 1 + \norm{H_1}\cdot \delta.
  \end{equation}
\end{proof}

\begin{lemma}\label{lem:2-local Shist}
  Let $V = U_T\dotsm U_1$ be a quantum circuit on $n$ qubits that only consists of $1$-local gates and CZ gates followed and preceded by two $Z$ gates, $n_1\le n$, and $\epsilon\in n^{-O(1)}$.
  Then there exists a $2$-local Hamiltonian $H$ on $n+T$ qubits and $\norm{H}\in \poly(n)$, such that $\bra{\phi}H\ket{\phi} = 0$ for all $\ket{\phi}\in \Shist$, and if $\Tr(H\rho) \le 1$, then there exists a state $\sigma$ in $\Shist$ (pure if $\rho$ is pure), such that $\trnorm{\rho-\sigma}\le\epsilon$, where for $\ket{\wht} = \ket{1^{t}\,0^{T-t}}$,
  \begin{equation}\label{eq:Shist}
    \Shist := \Span\left\{ \sum_{t=0}^T U_t\dotsm U_1\ket{x,0^{n_2}}\ket{\wht} \;\middle|\; x \in \{0,1\}^{n_1} \right\}.
  \end{equation}
\end{lemma}
\begin{proof}
  The proof is completely analogous to the proof of \cite[Lemma 5.1]{KKR06}, substituting the State Projection Lemma (\cref{lem:state-projection}) for the Projection Lemma \cite[Lemma 3.1]{KKR06}.
\end{proof}

\begin{lemma}[Extraction Lemma]\label{lem:extract}
  Let $\ket{\psihist} = (T+1)^{-1/2}\sum_{t=0}^T\ket{\psi_t}\ket{\wht}\in \Shist$ as defined in \cref{lem:2-local Shist} with $\ket{\psi_t} = U_t\dotsm U_1\ket{\phi,0^{n_2}}$ for some $\ket{\phi}\in\CC^{2^{n_1}}$.
  Let $j\in [T]$, $U_{j} = I$, $S = \{i,n+j\}$, and $\rho = \Tr_{\overline{S}}(\ketbrab\psihist)$.
  There exists a linear function $\fextract$ (independent of the circuit) such that $\fextract(\rho) = \Tr_{\overline{i}}(\ketbrab{\psi_j})$ and $\trnorm{\fextract(A)}\le (T+1)\trnorm{A}$ for all $A$.
\end{lemma}
\begin{proof}
  We have
  \begin{subequations}
    \begin{align}
      \rho &= \frac1{T+1}\sum_{t,t'\in\{0,\dots,T\}} \Tr_{\overline{S}}\left(\ketbra{\psi_t}{\psi_{t'}}\otimes \ketbra{\wht}{\whtp}\right)\\
      &= \frac1{T+1}\sum_{t,t'\in\{0,\dots,T\}} \Tr_{\overline{i}}\left(\ketbra{\psi_t}{\psi_{t'}}\right)\otimes \Tr_{\overline{j}}\left(\ketbra{\wht}{\whtp}\right)\\
      &= \frac1{T+1}\sum_{t=0}^{j-2} \Tr_{\overline{i}}(\ketbrab{\psi_t})\otimes\ketbra00 + \frac2{T+1} \Tr_{\overline{i}}(\ketbrab{\psi_j})\otimes\ketbra++ + \frac1{T+1}\sum_{t=j+1}^{T} \Tr_{\overline{i}}(\ketbrab{\psi_t})\otimes\ketbra11\label{eq:extraction-lemma},
    \end{align}
  \end{subequations}
  where \cref{eq:extraction-lemma} follows from the fact that
  \begin{equation}
    \Tr_{\overline{j}}(\ketbra{\wht}{\whtp}) = \ketbra{\wht_j}{\whtp_j}\cdot \braket{\wht_{\overline{j}}}{\whtp_{\overline{j}}}= \begin{cases}
      \ketbra{\wht_j}{\whtp_j} & \text{if } t=t' \text{ or } \{t,t'\}= \{j-1,j\}\\
      0 & \text{else}
    \end{cases},
  \end{equation}
  since $\{j-1,j\}$ is the only pair of time steps whose unary representation only differs in bit $j$.
  Then for $L_+ = I\otimes\bra+$, and $L_- = I\otimes \bra-$,
  \begin{alignat}{2}
      L_+\rho L_+^\dagger &= \frac1{2(T+1)}\sum_{t=0}^{j-2} \Tr_{\overline{i}}(\ketbrab{\psi_t}) + \frac2{T+1} \Tr_{\overline{i}}(\ketbrab{\psi_j}) &&+ \frac1{2(T+1)}\sum_{t=j+1}^{T} \Tr_{\overline{i}}(\ketbrab{\psi_t}),\\
      L_-\rho L_-^\dagger &= \frac1{2(T+1)}\sum_{t=0}^{j-2} \Tr_{\overline{i}}(\ketbrab{\psi_t})  &&+ \frac1{2(T+1)}\sum_{t=j+1}^{T} \Tr_{\overline{i}}(\ketbrab{\psi_t}).
  \end{alignat}
  Hence, we have 
  \begin{equation}
    \fextract(\rho) := \frac{T+1}2(L_+\rho L_+^\dagger - L_-\rho L_-^\dagger) = \Tr_{\overline{i}}(\ketbrab{\psi_j}).
  \end{equation}
  Finally, \begin{equation}
    \trnorm{\fextract(A)} \le \frac{T+1}2(\trnorm{L_+AL_+^\dagger} + \trnorm{L_-A L_-^\dagger})\le (T+1)\trnorm{A}.
  \end{equation}
\end{proof}

\begin{lemma}[Mixed Extraction Lemma]\label{lem:mixed-extract}
  Let $\rho$ be a mixed state on $n_1$ qubits and
  \begin{equation}
    \rhohist = \frac1{T+1}\sum_{t,t'\in\{0,\dots,T\}} \rho_{t,t'}\otimes \ketbra{\wht}{\whtp}, \quad \rho_{t,t'} := U_t\dotsm U_1(\rho\otimes\ketbrab{0}^{\otimes n_2})U_1^\dagger\dotsm U_{t'}^\dagger
  \end{equation}
  Let $j\in [T]$, $U_{j} = I$, $S = \{i,n+j\}$, and $\sigma = \Tr_{\overline{S}}(\rhohist)$.
  Then $\rhohist$ is a mixture of states in $\Shist$.
  Furthermore, we have $\fextract(\sigma) = \Tr_{\overline{i}}(\rho_{j,j})$.
\end{lemma}
\begin{proof}
  Let $\rho = \sum_{i=1}^{2^{n_1}} p_i\ketbrab{\psi^{(i)}}$, $\ket{\psi^{(i)}_t} = U_t\dots U_t\ket{\psi_i,0^{n_2}}$, and $\ket{\psihist^{(i)}} = (T+1)^{-1/2}\sum_{t=0}^T \ket{\psi_t^{(i)}}\ket{\wht} \in \Shist$.
  By linearity, we have
  \begin{equation}
    \rhohist = \sum_{i=1}^{2^{n_1}}\frac{p_i}{T+1}\sum_{t,t'\in\{0,\dots,T\}} \ketbra{\psi_{i,t}}{\psi_{i,t'}}\otimes \ketbra{\wht}{\whtp} = \sum_{i=1}^{2^{n_1}} p_i\ketbrab{\psihist^{(i)}}.
  \end{equation}
  By \cref{lem:extract}, we have for $\rho^{(i)} = \Tr_{\overline S}(\ketbrab{\psihist^{(i)}})$
  \begin{equation}
    \fextract(\sigma) = \sum_{i=1}^{2^{n_1}}p_i\fextract\left(\rho^{(i)}\right) = \sum_{i=1}^{2^{n_1}}p_i \Tr_{\overline{i}}\left(\ketbrab{\psi_t^{(i)}}\right) = \Tr_{\overline{i}}(\rho_{t,t}).
  \end{equation}
\end{proof}

Now suppose we have a state $\ket{\psi}$ approximately consistent with a set of $2$-local density matrices $\{\rho_{ij}\}$ that has low energy with respect to the $2$-local Hamiltonian $H=\sum_{ij} H_{ij}$ of \cref{lem:2-local Shist}.
Thus $\Tr(H\ketbrab\psi)\approx \sum_{ij}\Tr(H_{ij}\rho_{ij}) \le 1$ and $\ket{\psi}$ must be close to a history state.
By \cref{lem:extract}, we also know the $1$-local density matrices at certain snapshots of the computation.
In general, we do not know how to compute local density matrices of an accepting history state.
Therefore, the next step is to apply simulatable codes as in \cite{BG22}.

\subsection{Simulatable single-circuit super-verifier}\label{sec:sim-super-verifier}

We now combine all constraints of the super-verifier $\calVotp$ into a single circuit $\Vxs = U_T\dotsm U_1$ in the structure required by \cref{lem:2-local Shist}. Using $s$-simulatable codes we ensure that $s$-local reduced density matrices of the history state corresponding to an accepting proof can be classically computed in time $\poly(n,2^s)$.
Additionally, if a pure state is consistent with all $s$-local density matrices (for $s\ge2$), then it satisfies all constraints of $\calVotp$.

\begin{definition}[$s$-simulatable code \cite{BG22}\protect\footnotemark]\label{def:simulatable}
  \footnotetext{Our definition deviates from \cite[Definition 4.1]{BG22} in that it does not require transversal circuits for all gates. In fact the simulatable codes in \cite{BG22} also use non-transversal gates for the $T$-gadgets given in \cite[Section 4.3]{BG22}.}
  Let $\calC$ be an $[[N,1,D]]$-QECC, and $\calG$ a universal gateset, such that for each logical gate $G\in\calG$ on $k_G$ qubits, there exists a physical circuit $U_1^{(G)},\dots,U_\ell^{(G)}$ with $\ell\in\poly(N)$ that implements $G$ with the help of an $m_G$-qubit magic state $\tau_G$.
  We say $\calC$ is $s$-\emph{simulatable} if there exists a deterministic $2^{O(N)}$-time algorithm $\SimC(G,t,S)$ with $G\in\calG$, $t\in\{0,\dots,\ell\}$, $S\subseteq[N(m_G+k_G)]$, $\abs{S}\le s$, and output $\rho(G,t,S)$, such that for any $k_G$-qubit state $\sigma$ 
  \begin{equation}
    \rho(G,t,S) = \Tr_{\overline{S}}\left(\bigl(U_t^{(G)}\dotsm U_1^{(G)}\bigr)\Enc(\sigma\otimes \tau_G)\bigl(U_t^{(G)}\dotsm U_1^{(G)}\bigr)^\dagger\right),
  \end{equation}
  where $\Enc(\rho)$ denotes the encoding of $\rho$ under $\calC$.
\end{definition}

\begin{lemma}[\cite{GSY19,BG22}]\label{lem:steane-simulatable}
  For every $k>\log(s+3)$, the $k$-fold concatenated Steane code is $s$-simulatable.
\end{lemma}

\begin{remark}
  In the rest of this paper, $\calC$ will be the $k$-fold concatenated Steane code from \cref{lem:steane-simulatable} with $\calG = \{\CNOT, \Pgate, \Hgate, \Tgate\}$ for $\Pgate=\sqrt{\Zgate}$ in \cref{def:simulatable}.
  Only the \Tgate-gate requires a magic state $\tau_\Tgate = \Tgate\ketbra++\Tgate^\dagger\otimes \ketbra00$; the other gates can be applied transversally.
  $\rho(G,t,S)$ can therefore be represented exactly in $\QQ[e^{\iu\pi/4}]$.
\end{remark}

Note that the physical circuits in \cref{def:simulatable} are not necessarily in the form required by \cref{lem:2-local Shist}.
The next lemma shows that simulatability is robust under changes of the physical gateset.

\begin{lemma}\label{lem:steane-simulatable-gates}
  Let $\calC$ be an $s$-simulatable code.
  Define a code $\calC'$ identical to $\calC$, but using a different physical circuit implementation of the universal gateset:
  For each gate $G\in\calG$, $\calC'$ uses the physical circuit $W_{1,1}^{(G)},\dots,W_{1,r}^{(G)},W_{\ell,1}^{(G)},\dots,W_{\ell,r}^{(G)}$ for $r\in O(1)$, such that $U_{t}^{(G)} = W_{t,r}^{(G)}\dotsm W_{t,1}^{(G)}$, and $U_t^{(G)}$, $W_{t,j}^{(G)}$ are at most $k$-local for all $t\in[\ell]$, $j\in[r]$.
  $\calC'$ is $(s-k+1)$-simulatable.
\end{lemma}
\begin{proof}
  Let $G\in \calG$, $S\subseteq[N(m_G+k_G)]$ with $\abs{S}\le s-k+1$, $t\in [\ell]$, $j\in[r]$, and
  \begin{equation}
    \rho'(G,t,r,S) = \Tr_{\overline{S}}\left(\bigl(W_{t,j}^{(G)}\dotsm W_{1,1}^{(G)}\bigr)\Enc(\rho\otimes \tau_G)\bigl(W_{t,j}^{(G)}\dotsm W_{1,1}^{(G)}\bigr)^\dagger\right).
  \end{equation}
  Let $C$ ($\abs{C}\le k$) be the set of qubits $U_t^{(G)}$ acts on.
  If $C\cap S = \emptyset$, then
  \begin{equation}
    \begin{aligned}
      \rho'(G,t,r,S) &= \Tr_{\overline{S}}\left(\bigl(W_{t-1,r}^{(G)}\dotsm W_{1,1}^{(G)}\bigr)\Enc(\rho\otimes \tau_G)\bigl(W_{t-1,r}^{(G)}\dotsm W_{1,1}^{(G)}\bigr)^\dagger\right)\\
      &= \Tr_{\overline{S}}\left(\bigl(U_{t-1}^{(G)}\dotsm U_{1}^{(G)}\bigr)\Enc(\rho\otimes \tau_G)\bigl(U_{t-1}^{(G)}\dotsm W_{1}^{(G)}\bigr)^\dagger\right)
      = \SimC(G,t-1,S).
    \end{aligned}
  \end{equation}
  Otherwise, we have $\abs{S'}\le s-k+1$ for $S' = S \cup C$, and $W' = W_{t,j}^{(G)}\dotsm W_{t,1}^{(G)}$
  \begin{equation}
    \begin{aligned}
      \rho'(G,t,r,S) &= \Tr_{\overline{S}}\left(W'\cdot \Tr_{\overline{S'}}\left(\bigl(W_{t-1,r}^{(G)}\dotsm W_{1,1}^{(G)}\bigr)\Enc(\rho\otimes \tau_G)\bigl(W_{t-1,r}^{(G)}\dotsm W_{1,1}^{(G)}\bigr)^\dagger\right)\cdot W'{}^\dagger\right)\\
      &= \Tr_{\overline{S}}\left(W'\cdot \SimC(G,t-1,S')\cdot W'{}^\dagger\right),
    \end{aligned}
  \end{equation}
  which can be computed in time $2^{O(N)}$.
\end{proof}

\begin{figure}[ht]
  \centering
  \begin{subfigure}{\textwidth}
    \centering
    \begin{quantikz}[classical gap=0.7mm,wire types={b,b}]
      \lstick{$B:\ket{\psis}$}& \gate{\ChkEnc} && \gate[2]{\Chk(\Votp_{x,1})} & \ \ldots\ & \gate[2]{\Chk(\Votp_{x,m'})} &\\
      \lstick{$A:\ket{0^{n_2''}}$} && \gate{\ResGen} && \ \ldots\ & &
    \end{quantikz}
    \caption{$\Vxs$ first checks whether the proof is encoded, then generates resource states, and finally checks all constraints of $\calVotp$.}\label{fig:Vxs:a}
  \end{subfigure}\\[2mm]
  \begin{subfigure}{\textwidth}
    \centering
    \begin{quantikz}[classical gap=0.7mm,wire types={q,n,q}]
      \lstick[3]{$\ket{\psis}$}\lstick{$B_1$~~\makebox(0,0)[lb]{\vspace{-8mm}\hspace{2mm}$\vdots$}}& \gate[3]{\ChkEnc}& \midstick[3,brackets=none]{=\makebox(0,0)[lb]{\vspace{-0.7mm}\hspace{-6mm}$\vdots$\hspace{7.5mm}$\vdots$}} & \gate{\Udec} & \push{~I/2~} & \gate{\Uenc} &&&&\\
      \\
      \lstick{$B_{n_1'}$~~} &&&&&& \gate{\Udec} & \push{~I/2~} & \gate{\Uenc} &
    \end{quantikz}
    \caption{$\ChkEnc$ successively encodes and decodes each qubit. The annotation ``$I/2$'' denotes the expected reduced density matrix of that qubit. $\Udec$ is a unitary such that $\Udec\Enc(\ket\psi) = \ket{\psi}\otimes\ket{0^{N-1}}$ and $\Uenc = \Udec^\dagger$.}\label{fig:Vxs:b}
  \end{subfigure}\\[2mm]
  \begin{subfigure}{\textwidth}
    \centering
    \begin{quantikz}[classical gap=0.7mm,wire types={b,q,q,n,q}]
      \lstick{$B:\ket{\psis}$}& \gate[5]{\Chk(\Votp_{x,i})} &\midstick[5,brackets=none]{=\makebox(0,0)[lb]{\vspace{-11mm}\hspace{-6mm}$\vdots$\hspace{7.5mm}$\vdots$}} & \gate[5]{\Enc(\Votp_{x,i})} &&&& \gate[5]{\Enc(\Votp_{x,i}{}^\dagger)}& \\
      \lstick{$A_1:\ket{0}$} &&&& \gate{\Udec} & \push{~I/2~} & \gate{\Uenc}&&\\
      \lstick{$A_2:\ket{0}$\makebox(0,0)[lb]{\vspace{-7mm}\hspace{2mm}$\vdots$}} &&&&&&&&\\
      \\
      \lstick{$A_{n_2''}:\ket{0}$} &&&&&&&&
    \end{quantikz}
    \caption{To verify $\Votp_{x,i}$, first apply the verifier in the code space, then decode the output qubit, which is maximally mixed for acceptance probability $1/2$, and finally undo these steps again.}\label{fig:Vxs:c}
  \end{subfigure}
  \caption{Super-verifier $\Vxs$. Wires represent \emph{logical} qubits under $\calC$. Note that while the $\Chk$ procedures act as identity, their outcome can be read from the history state via the Extraction Lemma.}\label{fig:Vxs}
\end{figure}

We are now ready to unify all constraints of $\calVotp$ into the circuit $\Vxs$. $\Vxs$ expects the proof to be encoded with $\calC$, such that $\calC$ is $(3s+2)$-simulatable.
$\Vxs$ has a proof register $B$ on $n_1'$ logical qubits, and an ancilla register on $n_2''> n_2'$, where the additional ancilla qubits are used as resources for the $\Tgate$-gates.
$\Vxs$ is defined in \cref{fig:Vxs} and acts as follows:

\begin{enumerate}
  \item Receive proof $\ket{\psis} = \Enc(\ket{\psiotp})$ with $\ket{\psiotp} = \frac12 \sum_{a,b\in\{0,1\}} (X^aZ^b)^{\otimes n_1}\ket{\psi}\ket{abab}$ as in \cref{eq:psiotp}.
  \item Run \ChkEnc (see \cref{fig:Vxs:b}):
  \begin{enumerate}[label*=\arabic*.]
    \item For $i=1,\dots,n_1'$, decode and encode proof qubit $B_i$ under $\calC$. Add an identity gate after each $\Udec$ for \cref{lem:extract}.
  \end{enumerate}
  \item Run \ResGen:
  \begin{enumerate}[label*=\arabic*.]
    \item For each auxiliary qubit of $\Vxotp$, encode $\ket{0}$ under $\calC$.
    \item For each \Tgate-gate of the circuits $\Votp_{x,1},\Votpd_{x,1}{},\dots, \Votp_{x,m'},\Votpd_{x,m'}$ (see 4.), encode $\Tgate\ket{+}$ under $\calC$.
  \end{enumerate}
  \item For $i=1,\dots,m'$ run $\Chk(\Votp_{x,i})$ (see \cref{fig:Vxs:c}):
  \begin{enumerate}[label*=\arabic*.]
    \item Run the logical circuit $\Enc(\Votp_{x,i})$ under $\calC$.
    \item Decode and encode the first ancilla $A_1$ under $\calC$. Add an identity gate after each $\Udec$.
    \item Undo the logical circuit by running $\Enc(\Votpd_{x,i})$ under $\calC$.
  \end{enumerate}
\end{enumerate}

Here, $\Udec$ is an $N$-qubit unitary such that $\Udec\ket{\calC_b} = \ket{b}\ket{0^{N-1}}$ for $b\in \{0,1\}$, where $\ket{\calC_b}$ denotes the encoding of bit $b$ under $\calC$.
Hence, $\Udec\Enc(\ket\psi) = \ket\psi\otimes\ket{0^{N-1}}$ and $\Udec\ket{\psi}\ket{0^{N-1}} = \Enc(\ket\psi)$ for $\Uenc = \Udec^\dagger$.
\begin{restatable}{claim}{claimUenc}
  $\Uenc$ can be implemented exactly with $O(N)$ gates in $\calG$.
\end{restatable}
\begin{proof}
  See \cref{sec:omit}.
\end{proof}

Let $\ket\psishist$ be the history state for $\Vxs$, such that the input $\ket{\psiotp}$ satisfies every constraint in $\calVotp$ perfectly:
\begin{equation}\label{eq:psishist}
\ket\psishist = \frac1{\sqrt{T+1}}\sum_{t=0}^T U_t\dotsm U_1\ket{\psiotp}\ket{0^{n_2''}}\ket{\wht}
\end{equation}

\begin{claim}\label{claim:sim}
  There exists a deterministic classical algorithm $\SimV(x,S)$ that outputs a classical description of $\Tr_{\overline{S}}(\ketbrab{\psishist})$ for $S\subset[N(n_1'+n_2'')+T]$ with $\abs{S}\le s$ in time $\poly(2^N,T)$.
\end{claim}
\begin{proof}
  The proof is completely analogous to \cite[Lemma 3.5]{BG22} since each gate of $\Vxs$ either belongs to applying a logical gate, or it is part of encoding or decoding a maximally mixed qubit.
  Both of these cases are handled in \cite[Lemmas 4.8 and 4.9]{BG22}.
  \cref{lem:steane-simulatable-gates} gives simulatability with our modified physical gates.

  We can assume without loss of generality that the constraints of $\calVotp$ copy their output onto a fresh ancilla, such that 
  $\Votp_{x,i}\ket{\psiotp}\ket{0^{n_2'}} = \sqrt{p}\ket{11}\ket{\eta_1} + \sqrt{1-p}\ket{00}\ket{\eta_0}$ with $p = p(\Votp_{x,i},\psiotp)=1/2$.
  Thus, the output qubit is in state $p\ketbra11 + (1-p)\ketbra00 = I/2$.

  Note that we can compute reduced density matrices \emph{exactly} because \cref{lem:PSQMA=PureSuperQMA} allows us to assume without loss of generality that the proof is accepted with probability exactly $1/2$ by each circuit of the super-verifier.
\end{proof}

\subsection{Proof of hardness}
We are now ready to complete the hardness proof.

\begin{lemma}\label{lem:pure-hardness}
  $\tPureCLDM_1$ is $\PureSuperQMA$-hard.
\end{lemma}
\begin{proof}
  Let $A\in\PSQMA=\PureSuperQMA$ be as at the beginning of \cref{sec:complete}, and let $V_x := V_x^{(2)} = U_T\dotsm U_1$ be as defined below \cref{lem:steane-simulatable-gates}.
  Given an instance $x\in\{0,1\}^n$, we compute a $\tPureCLDM$ instance as follows:
  \begin{enumerate}
    \item Compute the Hamiltonian $H = \sum_{i,j\in [n']}H_{i,j}$ with $n':=N(n_1'+n_2'')+T$ by applying \cref{lem:2-local Shist} to $V_x$ for $\epsilon_1 \in n^{-O(1)}$ (to be determined in the soundness proof).
    \item For all $i,j\in[n']$, compute $\rho_{ij}=\SimV(x,\{i,j\})$.
    \item Compute $E = \sum_{ij}\Tr(H_{ij}\rho_{ij})$.
    \begin{enumerate}[label*=\arabic*.]
      \item If $E = 0$, output $\{\rho_{ij}\}_{i,j\in[n']}$ and $\beta$ (to be determined in the soundness proof).
      \item Otherwise, output a NO-instance for $\tPureCLDM$.\footnote{In principle, we would like to reject at this point since we know by \cref{lem:2-local Shist} that a valid history state has energy $E=0$. However, the reduction formalism requires us to map $x$ to an instance for $\tPureCLDM$. For example, a trivial NO-instance on $3$ qubits would be $\rho_{12}=\ketbrab{11}$ and $\rho_{23} = \ketbrab{00}$.}
    \end{enumerate}
  \end{enumerate}
  This reduction clearly runs in $\poly(n)$ time.

  If $x\in\Ayes$, then the history state $\ket{\psihist^{(2)}}$ defined in \cref{eq:psishist} is consistent with density matrices $\rho_{ij}$ by \cref{claim:sim}.

  Now let $x\in\Ano$, and assume there exists some $n'$-qubit state $\ket{\phi}$ approximately consistent with the $\rho_{ij}$, i.e., $\trnorm{\phi_{ij}-\rho_{ij}} < \beta$ with $\phi_{ij}:=\Tr_{\overline{ij}}(\ketbrab{\phi})$ for all $i,j\in[n']$.
  Hence, we must be in case 3.1 with $E=0$ and hence
  \begin{equation}
    \Tr(H\ketbrab\phi) = \sum_{ij}\left(\Tr(H_{ij}\phi_{ij}) - \Tr(H_{ij}\rho_{ij}) \right) = \sum_{ij}\Tr(H_{ij}(\phi_{ij}-\rho_{ij})) < \sum_{ij} \norm{H_{ij}}\beta \le \norm{H}\beta = 1
  \end{equation}
  for $\beta\le 1/\norm{H}$.
  By \cref{lem:2-local Shist}, 
  there exists a state $\ket{\psi}\in \Shist$, such that $\trnorm{\ketbrab\psi - \ketbrab\phi} \le \epsilon_1$.
  We can write $\ket{\psi} = (T+1)^{-1/2}\sum_{t=0}^T U_t\dotsm U_1\ket{\psi_0}\ket{0^{n_2''}}\ket{\wht}$.
  Next, we show that $\ket{\psi_0}$ is close to a valid codeword under $\calC$.

  Let $\psi_{ij} = \Tr_{\overline{ij}}(\ketbrab\psi)$.
  Then $\trnorm{\psi_{ij}-\rho_{ij}} \le \trnorm{\phi_{ij}-\rho_{ij}} + \trnorm{\psi_{ij} - \phi_{ij}} \le \beta + \epsilon_1 =: \epsilon_2$, where $\trnorm{\psi_{ij} - \phi_{ij}}\le\epsilon_1$ follows from the operational interpretation of the trace norm (also proven in \cite[Eq. (23)]{Ras12}).
  We can write $\ket{\psi_0} = \sqrt{1-\epsilon_3}\ket{\eta} + \sqrt{\epsilon_3}\ket{\eta^\perp}$, such that $\ket{\eta}$ is a valid codeword under $\calC$ and $\ket{\eta^\perp}$ is orthogonal to the codespace.
  Let $\ket{\wheta},\ket{\wheta^\perp}$ be obtained by applying $\Udec^{\otimes n_1'}$ to $\ket{\eta},\ket{\eta^\perp}$ and permuting the physical qubits such that the first physical qubit of each logical qubit is at the top.
  Then we have \begin{equation}
    \ket{\wheta} \in \Span\{\ket{x,0^{n_1'(N-1)}}\mid x\in \{0,1\}^{n_1'}\}
  \end{equation}
  and
  \begin{equation}
    \ket{\wheta^\perp} \in \Span\{\ket{x,y}\mid x\in \{0,1\}^{n_1'}, y\ne 0\in \{0,1\}^{n_1'(N-1)}\}.
  \end{equation}
  Let $\gamma_i = \Tr_{\overline{i}}(\Udec^{\otimes n_1'}\ketbrab{\psi_0}\Udec^{\dagger\otimes n_1'})$ for some qubit $i$.
  By the union bound, there exists an $i$ such that $\Tr(\ketbra11\gamma_i)\ge \epsilon_3/(n_1'(N-1))$.
  Let $t$ be the time step directly after applying $\Udec$ to qubit $i$ in \ChkEnc (see \cref{fig:Vxs:b}).
  By \cref{lem:extract}, there exists $j$ such that $\fextract(\psi_{ij}) = \gamma_i$ and $\fextract(\rho_{ij}) = \ketbra00$.
  Thus, 
  \begin{equation}
    \begin{aligned}
      \frac{\epsilon_3}{n_1'(N-1)}&\le \Tr(\ketbra11\gamma_i) = \Tr(\ketbra11(\fextract(\psi_{ij})-\fextract(\rho_{ij}))) \le \trnorm{\fextract(\psi_{ij})-\fextract(\rho_{ij})} \\
      &\le (T+1)\trnorm{\psi_{ij} - \rho_{ij}} \le (T+1)\epsilon_2.
    \end{aligned}
  \end{equation}
  
  Now let $\ket{\psi'} = (T+1)^{-1/2}\sum_{t=0}^T U_t\dotsm U_1\ket{\eta}\ket{0^{n_2''}}\ket{\wht}$.
  Then $\abs{\braket{\psi}{\psi'}}^2 = (\sum_{t=0}^T \braket{\psi_0}{\eta}/(T+1))^2= 1-\epsilon_3$ and thus $\trnorm{\psi-\psi'} \le 2\sqrt{\epsilon_3} =: \epsilon_4$.
  We argue that $\ket{\eta}$ approximately satisfies all constraints of the super-verifier $\calVotp$.
  Consider constraint $\Votp_{x,l}$.
  By \cref{lem:extract}, there exist $i,j$ (corresponding to qubit $A_1$ at the time step after applying $\Udec$ during $\Chk(\Votp_{x,l})$ (see \cref{fig:Vxs:c})), such that $\Tr(\ketbra11\fextract(\psi'_{ij})) = p(\Votp_{x,i},\eta)$.
  Thus,
  \begin{equation}
    \begin{aligned}
      \abs*{p(\Votp_{x,l},\eta)-\frac12} &= \abs*{\Tr\bigl(\ketbra11\fextract(\psi'_{ij})\bigr)-\Tr\bigl(\ketbra11\fextract(\rho_{ij})\bigr)} \le \trnorm{\fextract(\psi'_{ij}) - \fextract(\rho_{ij})}\\
      &\le (T+1)\trnorm{\psi_{ij}'-\rho_{ij}} \le (T+1)(\epsilon_4 + \epsilon_2) =: \epsilon_5.
    \end{aligned}
  \end{equation}
  For sufficiently small $\epsilon_1,\beta \in n^{-O(1)}$, we get $\epsilon_5 < \epsilon'$ (recall $\epsilon'$ is the soundness threshold of $\calVotp$).
  This however contradicts the assumption $x\in\Ano$ and $\ket{\phi}$ could not have been consistent with $\{\rho_{ij}\}$.
\end{proof}

\subsection{Mixed States}
Using the machinery developed for working with the $2$-local circuit-to-Hamiltonian construction we can also prove that mixed $\kCLDM$ is already $\QMA$-hard for $k \ge 2$.

\begin{theorem}
  $\kCLDM_1$ is $\QMA$-complete for all $k\ge2$.
\end{theorem}
\begin{proof}[Proof sketch]
  We prove the equivalent statement: $\kCLDM_1$ is $\SuperQMA$-complete.
  Containment is shown in \cite{Liu06,BG22}.
  The hardness proof is completely analogous to \cref{lem:pure-hardness} because \cref{lem:2-local Shist} and \cref{lem:mixed-extract} also hold for mixed states.
\end{proof}

\section{The pure-$N$-representability problem}\label{sec:pureNrep}
In this section, we use the fact that $\tPureCLDM_1$ is $\PureSuperQMA$-complete to prove completeness also for Bosonic and Fermionic $\pureNrep$. We begin with the Fermionic version.
\subsection{Fermions}
We will prove the following theorem:
\fermioncomplete*

Containment directly follows from the containment of mixed $\Nrep$ in $\SuperQMA$ as proven in \cite{LCV07}.
\begin{lemma}[\cite{LCV07}]
  $\pureNrep\in\PureSuperQMA$.
\end{lemma}

It remains to prove hardness, which we do by reduction from $\tPureCLDM_1$.
Let $\{\sigma_{ij}\}_{i,j\in[N]}$ be a $\tPureCLDM_1$ instance on $N$ qubits with soundness threshold $\beta$.
Note, by \cref{thm:2PureCLDM-complete}, we can assume we are given $\sigma_{ij}$ for \emph{all} $i,j\in[N]$.

We represent an $N$-qubit state with an $N$-fermion state on $d=2N$ modes (same as in \cite{LCV07}).
\begin{equation}
  \ket{z} \mapsto \ket{\whz}:=\prod_{i=1}^N(a^\dagger_{2i-1})^{1-z_i}(a^\dagger_{2i})^{z_i} \ket{\Omega} = \ket{1-z_1,z_1,\dots,1-z_N,z_N}
\end{equation}
Let $\calS := \Span\{\ket{\whz}\mid z\in \{0,1\}^N\}$ be the subspace of \emph{legal} $N$-fermion states.
\begin{claim}\label{claim:fermion-extract}
  Let $\ket{\psi} \in \CC^{2^N}$ and $\ket{\whpsi}\in \calS$ its fermionic representation, $u>v\in[N]$. Then
  \begin{equation}
    \Tr_{\overline{uv}}(\ketbrab\psi) = N(N-1)\sum_{b_i,b_j,b_k,b_l\in \{0,1\}}\rho^{[2]}_{ijkl}\ketbra{b_i,b_j}{b_k,b_l},
  \end{equation}
  where $i=2u-1+b_i,j=2v-1+b_j,k=2u-1+b_k,l=2v-1+b_l$ and $\rho^{[2]} = \Tr_{3,\dots,N}(\ketbrab{\whpsi})$.
\end{claim}
\begin{proof}
  We write $\ket{\psi} = \sum_z c_z\ket{z}$ and $\ket{\whpsi} = \sum_z c_z\ket{\whz}$.
  By \cref{eq:rhoijkl}, we have
  \begin{subequations}
  \begin{align}
    N(N-1)\rho^{[2]}_{ijkl} &= \Tr\left((a_k^\dagger a_l^\dagger a_ja_i)\ketbrab\whpsi\right) = \bra{\whpsi}(a_k^\dagger a_l^\dagger a_ja_i)\ket{\whpsi}\label{eq:expand ijkl}\\
    &=\sum_{y,z\in\{0,1\}^N}c_y^*c_z\bra{\why}(a_k^\dagger a_l^\dagger a_ja_i)\ket{\whz}\\
    &= \sum_{y,z\in\{0,1\}^N}c_y^*c_z(-1)^{(u-1)+(v-1)+(v-1)+(u-1)}\bra{y}(\ketbra{b_k,b_l}{b_i,b_j}_{uv}\otimes I_{\overline{uv}})\ket{z}\\
    &= \Tr\bigl((\ketbra{b_k,b_l}{b_i,b_j}_{uv}\otimes I_{\overline{uv}})\ketbrab\psi\bigr)=\Tr\bigl(\ketbra{b_k,b_l}{b_i,b_j}_{uv}\Tr_{\overline{uv}}(\ketbrab{\psi})\bigr)
  \end{align}
  \end{subequations}
\end{proof}
\begin{claim}\label{claim:fermion-cross}
  Let $\ket{\psi} \in \CC^{2^N}$, $\ket{\whpsi}\in \calS$ its fermionic representation, and $\rho^{[2]} = \Tr_{3,\dots,N}(\ketbrab{\whpsi})$.
  Then $\rho_{ijkl}^{[2]}=0$ unless $i,j,k,l$ are as in \cref{claim:fermion-extract} (up to swapping $i,j$ or $k,l$).
\end{claim}
\begin{proof}
  If $\{i,j\} =\{2u-1,2u\}$ for some $u\in [N]$, then $a_ja_i\ket{\whpsi}=0$, since in $\ket{\whpsi}$ exactly one of the modes $2u-1,2u$ can be occupied.
  Similarly, if $\{k,l\} =\{2u-1,2u\}$ for some $u\in [N]$, then $\bra{\whpsi}(a_k^\dagger a_l^\dagger a_ja_i)\ket{\whpsi}=0$ (see \cref{eq:expand ijkl}).
  
  The last case is $i=2u_i -1+ b_i,j=2u_j-1 + b_j,k=2u_k-1 + b_k,l=2u_l-1 + b_l$ with $u_i\ne u_j$ and $u_k\ne u_l$.
  If $i,j,k,l$ are not as in \cref{claim:fermion-extract}, then $\{u_i,u_j\} \ne \{u_k,u_l\}$, and therefore without loss of generality $u_i\notin \{u_k,u_l\}$.
  Thus $(a_k^\dagger a_l^\dagger a_ja_i)\ket{\whpsi}$ is either $0$ or has no fermion in modes $\{2u_i-1,2u_i\}$ and therefore $\bra{\whpsi}(a_k^\dagger a_l^\dagger a_ja_i)\ket{\whpsi}=0$.
\end{proof}
\begin{claim}\label{claim:fermion-compute}
  Given $\Tr_{\overline{uv}}(\ketbrab\psi)$ for all $u,v\in[N]$, $\rho^{[2]}=\Tr_{3,\dots,N}(\ketbrab{\whpsi})$ can be computed in polynomial time.
\end{claim}
\begin{proof}
  Follows from \cref{claim:fermion-extract,claim:fermion-cross}.
\end{proof}
\begin{claim}\label{claim:fermion-logical}
  Let $\rho$ be an $N$-fermion state, such that $\abs{\rho^{[2]}_{2u-1,2u,2u-1,2u}}\le\epsilon$ for all $u\in [N]$, where $\rho^{[2]} = \Tr_{3,\dots,N}(\rho)$.
  Then there exists a state $\whrho$ in $\calS$ ($\whrho$ is pure if $\rho$ is pure), such that $\maxnorm{\rho^{[2]}-\whrho^{[2]}}\le \sqrt{\epsilon}$.
\end{claim}
\begin{proof}
  Let $\Pi_{\calS},\Pi_{\calS^\perp}$ be the projector onto $\calS,\calS^\perp$.
  Let $\delta = \Tr(\Pi_{\calS^\perp}\rho)$, and $\rho' = \Pi_{\calS^\perp}\rho\Pi_{\calS^\perp}/\delta$, and $\whrho = \Pi_{\calS}\rho\Pi_{\calS}/(1-\delta)$.
  Note that when measuring $\rho'$ in the Fock basis, then there will always be two adjacent occupied modes $2u-1,2u$ for some $u\in[N]$.
  Thus by the Gentle Measurement Lemma,
  \begin{subequations}
  \begin{align}
    &N(N-1)\abs*{\rho^{[2]}_{2u-1,2u,2u-1,2u}} = \abs*{\Tr\left((a_{2u-1}^\dagger a_{2u}^\dagger a_{2u}a_{2u-1})\rho\right)} = \delta\abs*{\Tr\left((\hat{n}_{2u-1}\hat{n}_{2u})\rho'\right)} \ge \frac{\delta}{N}\\
    \Rightarrow\;& \trnorm*{\rho-\whrho} \le 2\sqrt\delta \le 2N^{3/2}\sqrt{\epsilon}.\\
    \Rightarrow\;& \forall i,j,k,l\colon\abs*{\rho^{[2]}_{ijkl}-\whrho^{[2]}_{ijkl}} = \frac{1}{N(N-1)}\abs*{\Tr\left((a_k^\dagger a_l^\dagger a_j a_i)(\rho-\whrho)\right)}\le \sqrt{\epsilon}.
  \end{align}
  \end{subequations}
\end{proof}

\begin{lemma}\label{lem:PureRDM-hard}
  (Fermionic) $\pureNrep_1$ is $\PureSuperQMA$-hard.
\end{lemma}
\begin{proof}
  Construct $\rho^{[2]}$ from $\{\sigma_{ij}\}_{i,j\in[N]}$ using \cref{claim:fermion-compute}.
  If there exists a state $\ket{\psi}$ consistent with $\sigma_{ij}$, then $\ket{\whpsi}$ is consistent with $\rho^{[2]}$.

  Now suppose there exists an $N$-fermion state $\ket{\phi}$, such that $\trnorm{\wtrho^{[2]}-\rho^{[2]}}\le\beta'$ for $\wtrho^{[2]}=\Tr_{3,\dots,N}(\ketbrab{\phi})$ and $\beta'$ to be determined later.
  Since $\maxnorm{\wtrho^{[2]}-\rho^{[2]}} \le \fnorm{\cdot} \le \trnorm{\cdot} \le \beta'$,
  By \cref{claim:fermion-logical}, there exists a state $\whrho = \ketbrab{\whpsi}$ with $\ket{\whpsi}\in\calS$, such that $\maxnorm{\rho^{[2]}-\whrho^{[2]}}\le \maxnorm{\wtrho^{[2]}-\rho^{[2]}} + \maxnorm{\wtrho^{[2]}-\whrho^{[2]}} \le \beta' + \sqrt{\beta'} =: \epsilon$, using $\maxnorm{\cdot} \le \trnorm{\cdot}$.

  For $i,j\in [N]$, let $\whsigma_{ij} = \Tr_{\overline{ij}}(\ketbrab\psi)$.
  Thus by \cref{claim:fermion-extract}, $\trnorm{\sigma_{ij}-\whsigma_{ij}} \le \frac{d(d-1)}2 \maxnorm{\sigma_{ij}-\whsigma_{ij}} \le \frac{N(N-1)d(d-1)}2\epsilon < \beta$ for sufficiently small $\beta' \in N^{-O(1)}$.
\end{proof}

\begin{theorem}
  $\RDM_1$ is $\QMA$-complete.
\end{theorem}
\begin{proof}
  We prove the equivalent statement: $\kCLDM_1$ is $\SuperQMA$-complete.
  Containment is proven in \cite{LCV07}.
  Hardness is completely analogous to \cref{lem:PureRDM-hard}.
\end{proof}

\subsection{Bosons}

\cref{claim:fermion-compute,claim:fermion-cross,claim:fermion-extract,claim:fermion-logical} also work for bosons.
Note in the proof of \cref{claim:fermion-logical}, we also need to use $N(N-1)\abs{\rho_{iiii}^{[2]}} = \Tr((a_i^\dagger a_i^\dagger a_ia_i)\rho) \approx 0$ to argue that $\rho$ is close to a state with at most one boson in each mode.
Thus, we obtain analogous completeness results:

\bosoncomplete*

\begin{theorem}
  $\BosonRDM_1$ is $\QMA$-complete.
\end{theorem}

Containment and $\QMA$-hardness under Turing reductions was shown in \cite{WMN10}.

\section{$\kPureCLDM$ is in $\PSPACE$}\label{sec:PSPACE}

In this section we set out to prove the containment of  $\PureSuperQMA$ in $\PSPACE$.
However, the techniques we develop for this result turn out to be much more general and lead to a poly-logarithmic time parallel algorithm (i.e., $\NC$) for solving quadratic systems with few constraints.

\PureSuperQMAinPSPACE*
At a high level, the proof consists of three steps. First, we reduce $\PureSuperQMA$ to deciding if $p(Q(X))$ has any zeros. Here $p\colon\RR^{\poly(n)} \to \RR$ is a degree $\poly(n)$ polynomial in $\poly(n)$ variables and $Q\colon \RR^{\exp(n)} \to \RR^{\poly(n)}$ is a quadratic polynomial.
Next, we use results by Grigoriev and Pasechnik \cite{GP05} that reduce finding zeros of $p(Q(X))$ to finding limits of zeros of smaller systems. There will be $\exp(n)$ smaller systems, each consisting of $\poly(n)$ equations of degree at most $\exp(n)$. Crucially, these equation will have at most $\poly(n)$ variables. We modify this reduction to work in $\PSPACE$. We then compute these limits efficiently in parallel using an algorithm for the first-order theory of the reals by Renegar \cite{Ren92}, making the total computation a $\PSPACE$ computation and proving the theorem.

\subsection{$\PureSuperQMA$ as a polynomial system}

\label{subsection:PSQMAtopolynomials}
We first describe how to reduce $\PureSuperQMA$ to a \emph{GP system}.

\begin{definition}[Grigoriev-Pasechnik system (GP system) \cite{GP05}]
  A GP system is a polynomial of the form $p(Q(X))$ where $p\colon \RR^k \to \RR$ is a degree $d$ polynomial and $Q = (Q_1, \dots, Q_k) \colon \RR^N \to \RR^k$ is a quadratic map. Both are assumed to have coefficients in $\ZZ$.\footnote{The techniques by Grigoriev and Pasechnik are valid in the more general case where $\RR$ is replaced by an arbitrary real closed field $\KK$ and $\ZZ$ with a computable subring of that field.} We say a GP system is satisfiable if there exists $x \in \RR^N$ such that $p(Q(x)) = 0$.
\end{definition}

Consider a super-verifier $\calV = \{(V_i, r_i, s_i)\}_{[k]}$ where $k= \poly(n)$. In this section we will let $V_i$ denote the POVM measurement operator implemented by the circuit so that the acceptance probability $p(V, \psi) = \bra{\psi}V_i\ket{\psi}$. In the YES-case, there exists some quantum state $\ket{\psi} \in \CC^{2^n}$ such that for all $i$, 
\begin{equation}
  \label{eq:psi V psi}
  |\bra{\psi}V_i\ket{\psi} - r_i| \le s_i
\end{equation}
 whereas in the NO-case there will be at least one $i$ for which $|\bra{\psi}V_i\ket{\psi} - r_i| > s_i$. Note that depending on the $\delta$ parameter in \cref{def:PureSuperQMA}, there could be more such $i$ but this is not required for our methods. What \emph{is} required for our method to work is that $k$, the total number of different $(V_i, r_i, s_i)$ the super-verifier could output, is polynomial.

To distinguish the two cases, we write the $n$-qubit state $\ket{\psi}$ as an exponentially long vector in $\CC^N$. Here and in the rest of this section $N = 2^n$. For each entry of the vector we introduce two variables: one for the real part and one for the complex part. That is, we write
\begin{equation}
  \ket{\psi} = \begin{pmatrix}
    a_1 + b_1 i \\ \vdots \\ a_N + b_N i.
  \end{pmatrix}
\end{equation}

Note that for any assignment of the variables $(a_1, b_1, \dots, a_N, b_N) \in \RR^{2N}$ we get an (unnormalized) vector in $(\CC^2)^{\otimes n}$. We now construct a GP system that is satisfiable if and only if an accepting proof state exists.

First, to ensure that the $a_j, b_j$ represent a normalized quantum state we define 
\begin{equation}
  \label{eq:Qnormalization}
  Q_0(a_1, b_1, \dots, a_N, b_N) = \|\ket{\psi}\|^2 - 1 = \left(\sum_{j = 1}^N a_j^2 + b_j^2\right) - 1
\end{equation}
Next, note that $\bra{\psi}V_i\ket{\psi}$ is already a quadratic equations. Although $\ket{\psi}$ can have a complex part $\bra{\psi}V_i\ket{\psi}$ will be real since $V_i$ is positive semidefinite. We define 
\begin{equation}
  Q_i(a_1, b_1, \dots, a_N, b_N) = \bra{\psi}V_i\ket{\psi} - r_i, \quad \text{for } i \in [k].
\end{equation}
In order to handle the inequalities in \cref{eq:psi V psi} we will add slack variables $c_1, \dots c_k$ and define
\begin{equation}
  Q_{k+i} = c_i, \quad \text{for } 1 \le i \le k.
\end{equation}
We now put all constraints together and define 
\begin{equation}
  p(y_0, \dots, y_{2k}) = y_0^2 + \sum_{j = 1}^k (y_j^2 + y_{k + j}^2 - s_j^2)^2.
\end{equation}
We then have
\begin{equation}
  p(Q(a_1, b_1 \dots, a_N, b_N, c_1, \dots, c_k)) = \left(\|\ket{\psi}\|^2 - 1\right)^2 + \sum_{j = 1}^k \left((\bra{\psi}V_j\ket{\psi} - r_i)^2 + c_{j}^2 - s_j^2\right)^2.
\end{equation}
Note that $p(Q(X))$ will be zero only when all component parts are zero. The first term enforces that the norm of $\ket{\psi}$ is $1$. Meanwhile, the $j$-th term in the sum makes sure that $(\bra{\psi}V_j\ket{\psi} - r_j)^2$ can be made equal to $s_j^2$ by adding some nonnegative value $c_j^2$. In other words, it ensures that $|\bra{\psi}V_j\ket{\psi} - r_j| \le s_j$.
We conclude that $p(Q(a_1, b_1 \dots, a_N, b_N, c_1, \dots, c_k))$ has a zero if and only if there exists some quantum state $\ket{\psi}$ such that for all $i$, $|\bra{\psi}V_i\ket{\psi} - r_i| \le s_i$, which is exactly as we set out to construct.

The coefficients of $p(Q(X))$ will be rational if the entries of the $V_i$ are rational (i.e., in $\QQ(i)$). 
If the $V_i$ contain irrational algebraic coefficients, the Primitive Element Theorem ensures that all coefficients are in the field extension $\QQ(a)$ for some algebraic $a$.
As Grigoriev and Pasechnik's methods work more generally when the coefficients of the polynomials are in a computable subring of $\RR$, we can use $\QQ(\Im(a),\Re(a))\subseteq\RR$.
Computation in $\QQ(\Im(a),\Re(a))$ is still efficient since $\QQ(\Im(a),\Re(a))=\QQ(b)$ for some real algebraic number $b$, and we only need to keep track of coefficients of $b^k$ for $k=0,\dots,\deg(b)-1$.

\subsection{Reducing the number of variables}\label{subsec:reduce}

In this subsection, we consider a general GP system $p(Q(X)) = \zeta$ and describe Grigoriev and Pasechnik's methods for reducing the number of variables. Our treatment of their methods is quite thorough, reproducing large parts of their construction. We hope this makes their methods easily accessible to other quantum theorists. 

Without loss of generality, we can write $Q$ as
\begin{equation}
  Q_j:X\mapsto \frac12 X^\sfT H_j X + b_j^\sfT X + c_j, \quad j\in[k], c_j\in \ZZ, b_j\in \ZZ^N, H_j = H_j^\sfT \in \ZZ^{N\times N}.
\end{equation}
Define 
\begin{equation}
  p_i = \frac{\partial p(Y)}{\partial Y_i}, \quad i \in [k].
\end{equation}

For their main construction, Grigoriev and Pasechnik rely on some assumptions. These assumptions hold generically, but could fail in certain degenerate cases. We will later discuss how these assumptions are removed.
We assume the following:
\begin{enumerate}[label=(\arabic*)]
  \item The set $Z = Z(p(Q(X)))$ of zeros of $p(Q(X))$ is bounded, i.e., there is some $r$ such that $\|X\| >r \implies p(Q(X)) \ne 0$. In our case, this assumption is trivially satisfied because of the normalization constraint \cref{eq:Qnormalization}.
  \item $\zeta$ is a regular value of $p(Q(x))$ and $p(Y)$. That is, there exists no $X \in \RR^n$ (respectively $Y \in \RR^k$) with $p(Q(X)) = \zeta$ (respectively $p(Y) = \zeta$) for which $\nabla p(Q(X)) = 0$ (respectively $\nabla p(Y) = 0$). 
  \item The matrices $\hat{H}_i$ obtained from $H_i$ by deleting the first row are in \emph{$r$-general posistion}. I.e., \begin{equation}
    \rk\left(\sum_{i=1}^k t_i\hat{H}_i\right)\ge r,\qquad \forall t\in \{(p_1(Q(X)),\dots,p_k(Q(X)))\mid X\in Z\}.
  \end{equation}
  Note that Assumption 2 implies that $t$ is never $0$ in the above.
\end{enumerate}
We are now ready to summarize Grigoriev and Pasechnik's construction as the following theorem.
\begin{theorem}[{\cite[Theorem 4.5]{GP05}}]\label{thm:4.5}
  Let $p, Q$ be a GP system satisfying the three assumptions above. Then, one can construct sets $V_c(U,W)$, such that $\bigcup V_c(U,W)$ intersects every connected component of $Z$. 
  
  The sets $V_c(U,W)$, also called ``pieces'', are indexed by row sets $U \subseteq \{2, \dots, N\}$ and column sets $W\subset [N]$ such that $r\le \abs{U} = \abs{W} \le N-1$.
  For such a $W$, let $\phi_W$ denote the polynomial
  \begin{equation}\label{eq:phi_W}
    \phi_W : \RR^N \to \RR^{k+N-\abs W}, \quad X\mapsto \pmat{Q(X) \\ X_{\Wbar}}.
  \end{equation}
  For each $V_c(U,W)$, the set $\phi_W(V_c(U,W)) \subseteq \RR^k \times \RR^{N-|W|}$ is defined as the set of points satisfying the following equations:
  \begin{subequations}
    \label{eq:pieces}
    \begin{align}
      p(Y) &= \zeta,\\
      \Omega \coloneq \det\Phi(Y)_{UW} &\ne 0, \\
      \Omega^2 Y &= \Omega^2 Q(\phi_{UW}^{-1}(Y,T)),\\
      \Omega b(Y)_{\Ubar} &= \Omega\Phi(Y)_{\Ubar W}\Phi(Y)^{-1}_{UW}\cdot b(Y)_U,\\
      \det\Phi(Y)_{U'W'} &= 0\quad \forall U',W': \abs{U'} = \abs{W'} = \abs{U} +1, U\subset U' \subset\{2,\dots,N\}, W\subset W'\subset[N],
    \end{align}
    \end{subequations}
  Here $\Phi(Y) = \sum_{j=1}^k  p_j(Y)\hat{H}_j$ and $\phi_{UW}^{-1}(Y,T)$ is the inverse of $\phi_W$ on $\phi_W(V_c(U,W))$. $\phi_{UW}^{-1}(Y,T)$ is given explicitly by
  \begin{equation}\label{eq:phi_UW^-1}
    \phi_{UW}^{-1}:\RR^{k}\times\RR^{N-|W|}\to\RR^N, \quad \pmat{Y\\T} \mapsto \pmat{\Phi(Y)^{-1}_{UW}(b(Y)_U - \Phi(Y)_{U\Wbar}T)\\T}
  \end{equation}
  If $r = N - k + 1$ (which will be ensured later with \cref{lem:generalposition}), then there are $N^{O(k)}$ pieces, each defined by $O(k^2)$ polynomials of degree $O(N)$ in $O(k)$ variables.

  Lastly, if the coefficients of $p$ and $Q$ are integers of bit length at most $L$, then $\phi^{-1}_{UW}$ and the equations defining the $V_c(U,W)$ can be computed from the coefficients of $Q$ and $p$ in time $\poly(\log N, \log d, k, \log L)$ using $(dN)^{O(k)}L^{O(1)}$ parallel processors.
\end{theorem}

For completeness we reproduce Grigoriev and Pasechnik's proof here.
\begin{proof}
Let $Q\colon \RR^N \to \RR^k$ be a quadratic map, $p\colon \RR^k \to \RR$ be any polynomial of degree $d$ and $\zeta \in \RR$ be some constant. We will be interested in finding points in $Z = Z(p(Q(X)) - \zeta)$, the set of $X \in \RR^N$ for which $p(Q(X)) = \zeta$. 

Grigoriev and Pasechnik's goal is to efficiently find a point in every connected component of $Z$. To this end, they use a critical point method, that is, they restrict themselves to the set of critical points of a well-chosen projection. This projection will simply be the projection onto the first coordinate $\pi:Z\to\RR,X\mapsto X_1$. By definition, a critical point of $\pi$ is a point $X$ where the rank of the differential $D\pi(X)$ from $T_X(Z)$ to $\RR$ is 0\footnote{By definition a critical point is a point where the rank of the differential is less than the dimension of the codomain, which in this case is $1$.}. Since $\pi$ is the projection onto the first coordinate, a point $X$ will be critical iff the tangent space $T_X(Z)$ at $X$ has a constant $X_1$ coordinate. By definition $T_X(Z)$ is orthogonal to the gradient $\nabla p(Q(X))$, which is nonzero for all $X \in Z$ by assumption. It follows that $\nabla p(Q(X)) = (a, 0,\dots,0)$ for some $a\in\RR$.
We conclude that the set $V_c$, of critical points is defined by
\begin{subequations}
\begin{align}
  p(Q(X)) - \zeta &= 0\label{eq:zeta}\\
  \frac{\partial p(Q(X))}{\partial X_j} &= 0\quad \forall j \in  \{2,\dots,N\}\label{eq:partial0}
\end{align}
\end{subequations}
By \cite[Proposition 7.4]{BPR06}, $V_c$ intersects every semi-algebraically connected component of $Z$ if $Z$ is bounded, which is the case by assumption 1.

The task of finding a point in every connected component of $Z$ has hence been reduced to finding a point in every connected component of $V_c$, giving us more structure to work with.

Using the chain rule, we can write \cref{eq:partial0} as 
\begin{equation}
  \frac{\partial p(Q(X))}{\partial X_i} = \sum_{j=1}^k p_j(Q(X)) e_i^\sfT(H_jX + b_j) = 0, \quad 2\le i\le N,
\end{equation}
where $e_i$ denotes the $i$-th standard basis vector of $\RR^N$. 
We can now rewrite \cref{eq:partial0} as the matrix equation
\begin{align}
  \Phi(Q(X))X &= b(Q(X))
\end{align}
where
\begin{align}
  \Phi(Y) &= \sum_{j=1}^k p_j(Y) \hat H_j,\quad b(Y) = -\sum_{j=1}^k p_j(Y) \hat b_j.\label{eq:Phi(Y)}
\end{align}
The ``hat'' denotes that the first row has been removed.

We substitute $Y=Q(X)$ to obtain a linear system $\Phi(Y)X = b(Y)$ in $X$. We would like to solve this system by inverting $\Phi(Y)$. However, $\Phi(Y)$ is not invertible in general, but we do have $\rk(\Phi(Y))\ge r$ by the $r$-general position assumption.
Therefore, $\Phi(Y)$ has at least one invertible $r\times r$ submatrix, for which there are at most $N^{O(N-r)}$ candidates.
We split solving the system into at most $N^{O(N-r)}$ cases, one for each of the maximal invertible submatrices of $\Phi(Y)$ (maximality will be of use later).

For each $U \subseteq \{2, \dots, N\}, W \subseteq [N]$ with $\abs{U} = \abs{W} \ge r$ let $V_c(U,W) \subseteq V_c$ be the set of $X \in V_c$ for which $\Phi(Q(X))_{UW}$ is a maximal invertible submatrix of $\Phi(Q(X))$. Note that these ``pieces'' can in general intersect or coincide.
The following conditions ensure $\Phi(Q(X))_{UW}$ is a maximal invertible submatrix:
\begin{subequations}
\begin{align}
  \det\Phi(Q(X))_{UW} &\ne 0\label{eq:invertible}\\
  \det \Phi(Q(X))_{U'W'} &= 0\qquad \forall U',W': \abs{U'} = \abs{W'} = \abs{U} +1, U\subset U' \subset\{2,\dots,N\}, W\subset W'\subset[N]\label{eq:maximal}
\end{align}
\end{subequations}
Since for every $X \in V_c$ there is some invertible submatrix of size at least $r\times r$, we get a decomposition of $V_c$ into the pieces $V_c(U,W)$:
\begin{equation}
  V_c = \bigcup_{\subalign{U&\subseteq\{2,\dots,N\}\\ W&\subset [N]\\\abs{U}&=\abs{W}\ge r}} V_c(U,W)
\end{equation}
with $V_c(U,W)$ defined by \cref{eq:zeta,eq:partial0,eq:maximal,eq:invertible}.

The systems defining the $V_c(U,W)$ still contain all $X$ variables.
In the next step we remove the dependency on $X_W$. To do so, invert $\Psi = \Phi(Q(X))_{UW}$ using Cramer's rule.
Assume, without loss of generality, that $U = \{2,\dots,r+1\}$ and $W = [r]$.
Then for $i,j\in[r]$
\begin{equation}\label{eq:inv-cramer}
  \Psi^{-1}_{ij} = \frac{(-1)^{i+j}\det \Psi(j,i)}{\det \Psi},
\end{equation}
where $\Psi(j,i)$ is obtained by removing row $j$ and column $i$ from $\Psi$.
Next, write $\Phi(Q(X))X = b(Q(X))$ in block form (dropping $(Q(X))$):
\begin{equation}
  \pmat{\Psi&\Phi_{U\Wbar}\\\Phi_{\Ubar W}&\Phi_{\Ubar\Wbar}}\bmat{X_W\\X_{\Wbar}} = \bmat{b_U\\b_{\Ubar}}
\end{equation}
Apply the invertible matrix $\psmallmat{\Psi^{-1}&0\\-\Phi_{\Ubar W}&I}$ to left of both sides of the above equation to obtain
\begin{equation}
   \pmat{I&\Psi^{-1}\Phi_{U\Wbar}\\0&0}\bmat{X_W\\X_{\Wbar}} = \bmat{\Psi^{-1}b_U\\b_{\Ubar} - \Phi_{\Ubar W}\Psi^{-1} b_U}.
\end{equation}
Note that the lower blocks of the above matrix are $0$, since otherwise its rank would be $>\abs{W}$, contradicting the choice of $\Psi$ as a maximal invertible submatrix (see \cref{lem:max-invertible-submat}).
Hence, we get another definition of $V_c(U,W)$ as \cref{eq:zeta,eq:invertible} and 
\begin{subequations}
\begin{align}
  X_W &= \Phi(Q(X))_{UW}^{-1}\cdot\bigl[b(Q(X))_U - \Phi(Q(X))_{U\Wbar}X_{\Wbar}\bigr],\label{eq:X_W}\\
  b(Q(X))_{\Ubar} &= \Phi(Q(X))_{\Ubar W}\Phi(Q(X))^{-1}_{UW}\cdot b(Q(X))_U.
\end{align}
\end{subequations}
Then, aside from the left hand side of \cref{eq:X_W}, $X_W$ only occurs as an argument to $Q$.
We are now in a position to define the map
\begin{equation}
  \phi_{UW}^{-1}:\RR^{k}\times\RR^{N-|W|}\to\RR^N, \quad \pmat{Y\\T} \mapsto \pmat{\Phi(Y)^{-1}_{UW}(b(Y)_U - \Phi(Y)_{U\Wbar}T)\\T}
\end{equation}
We need to show that $\phi_{UW}^{-1}$ as above is indeed the inverse of $\phi_W$ (\cref{eq:phi_W}) on $V_c(U,W)$. It suffices to show that $\phi_{UW}^{-1}(\phi_W(x))=x$ for all $x\in V_c(U,W)$.
Since both $\phi_{UW}^{-1}$ and $\phi_{W}$ act as identity on $x_{\Wbar}$, it remains to check $\Phi(Q(x))^{-1}_{UW}(b(Q(x))_U - \Phi(Q(x))_{U\Wbar}\,x_{\Wbar})=x_W$, which holds by \cref{eq:X_W}.

We can define $\phi_W(V_c(U,W))$ explicitly in terms of the variables $Y,T$ as used in \cref{eq:phi_UW^-1}.
Let $\Omega = \det \Phi(Y)_{UW}$.
\begin{subequations}\label{eq:phiW(Vc)}
\begin{align}
  p(Y) &= \zeta,\label{eq:p(Y)=zeta}\\
  \Omega^2 Y &= \Omega^2 Q(\phi_{UW}^{-1}(Y,T)),\label{eq:Y=Q(phi)}\\
  \Omega b(Y)_{\Ubar} &= \Omega\Phi(Y)_{\Ubar W}\Phi(Y)^{-1}_{UW}\cdot b(Y)_U,\label{eq:b(Y)}\\
  \det\Phi(Y)_{U'W'} &= 0\quad \forall U',W': \abs{U'} = \abs{W'} = \abs{U} +1, U\subset U' \subset\{2,\dots,N\}, W\subset W'\subset[N],\label{eq:det=0}\\
  \det\Phi(Y)_{UW} &\ne 0,\label{eq:Omega_ne_0}
\end{align}
\end{subequations}
where the multiplication by $\Omega$ cancels the denominators coming from Cramer's rule in \cref{eq:inv-cramer} to make both sides polynomials, and \cref{eq:Omega_ne_0} ensures $\Omega\ne 0$.

We now argue that the semialgebraic set defined by \cref{eq:phiW(Vc)} is indeed equal to $\phi_W(V_c(U,W))$.
First, let $(y,t)$ satisfy \cref{eq:phiW(Vc)} and $x = \phi_{UW}^{-1}(y,t)$, which is well defined by \cref{eq:Omega_ne_0}.
We argue $x\in V_c(U,W)$, recalling its definition by \cref{eq:zeta,eq:partial0,eq:maximal,eq:invertible}.
\cref{eq:zeta} follows from \cref{eq:p(Y)=zeta,eq:Y=Q(phi)}, \cref{eq:partial0} from \cref{eq:b(Y)}, and \cref{eq:maximal,eq:invertible} from \cref{eq:det=0,eq:Omega_ne_0}.
Then $\phi_{W}(x) = \psmallmat{Q(\phi_{UW}^{-1}(y,t))\\t} = \psmallmat{y\\t}\in \phi_W(V_c(U,W))$ by \cref{eq:Y=Q(phi)}.
Similarly, \cref{eq:phiW(Vc)} is satisfied for any $(y,t) = \phi_W(x)$ for $x\in V_c(U,W)$.

Note, the entries of $\Phi(Y)$ are linear combinations of the polynomials $p_j(Y)$ (see \cref{eq:Phi(Y)}).
Hence, the determinants of its $m\times m$ submatrices are of degree at most $m$ in the $p_j(Y)$.
The other degree counts in \cref{thm:4.5} follow analogously.

To argue the parallel complexity bounds, the only nontrivial part is the computation of determinants of submatrices of $\Phi(Y)$.
The parallel complexity for this task follows from \cref{lem:determinant} below.
\end{proof}

\begin{lemma}[Implicit in \cite{Ren92}]\label{lem:determinant}
  The determinant of a matrix $A \in \ZZ[X_1,\dots,X_m]^{n\times n}$ with entries of degree $d$ and coefficients of bit length $L$ can be computed in parallel time $\poly(\log n, \log d, m, \log L)$ with  $(dn)^{O(m)}L^{O(1)}$ processors.
\end{lemma}
\begin{proof}
  The polynomial $f:\ZZ^m\to \ZZ,x\mapsto n!\cdot \det A(x)$ has degree $d'=dn$.
  We can compute the coefficients of $Ef$ for $E=\prod_{0\le j<k\le d'}(k-j)^n$ via polynomial interpolation with \cite[Lemma 2.1.3]{Ren92} from the values $f(x)$ for all $x\in \{0,\dots,d'\}^m$.
  We can compute $f(x)$ in time $O(\log^2 n)$ using $n^{O(1)}$ parallel processors using \cite[Proposition 2.1.1]{Ren92}.
  Note, the coefficients of $A(x)$ have bit length at most $L' = dL\log d'$.
  Hence, the coefficients occurring during the computation have bit length at most $L'' = L'n^{O(1)}$.
  Then the computation of the coefficients of $Ef$ takes $(m\log d')^{O(1)}$ time with $(d')^{O(m)}$ processors.
  Integers during that computation have bit length at most $L''+m(d')^{O(1)}$.
  Finally, division by $En!$ can be done efficiently with \cite{SR88} in time $O(\log L'')$ with $(L'')^{O(1)}$ processors.
\end{proof}

\begin{remark}\label{remark:constants}
  \cref{thm:4.5} still holds with the same runtime and processor bounds when the GP system has coefficients in $\ZZ[C]$ of degree $O(1)$ (instead of $\ZZ$), for real constants $C=(C_1,\dots,C_m)\in\RR^m$ and $m=O(k)$.
  The algorithm then also outputs the equations for the pieces with coefficients in $\ZZ[C]$.
  Note that the input and output representations treat these constants symbolically (i.e. like free variables), so that the output is independent of the actual value of $C$.
  In the proof of \cref{thm:4.5}, and specifically the determinant calculations via \cref{lem:determinant}, the additional constants can be handled just like the variables.
\end{remark}

\subsubsection{Removing the assumptions}

We now turn to the removal of the assumptions made in \cref{thm:4.5}. To do so, Grigoriev and Pasechnik use limit arguments. The key idea is that if we perturb the initial system of polynomials slightly, there will be at most finitely many values of the perturbation for which the assumptions fail. Hence, if the perturbation is sufficiently small, all assumptions will hold. They then argue that the solutions to our initial system are equal to the solutions of the perturbed system if we let the perturbation go to zero. The assumptions will hold in the limit. We will later discuss how we can compute these limits.

The following lemma deals with the general position assumption.
\begin{lemma}[{\cite[Lemma~5.2]{GP05}}]
  \label{lem:generalposition}
  Let $p, Q$ be a GP system and write $Q$ as
  \begin{equation}
    Q_j:X\to \frac12 X^\sfT H_j X + b_j^\sfT X + c_j, \quad j\in[k], c_j\in \ZZ, b_j\in \ZZ^N, H_j = H_j^\sfT \in \ZZ^{N\times N}.
  \end{equation}  
  Define $J(j)$ to be the $N\times N$ matrix with diagonal $(1^{j-1}, 2^{j-1}, \dots, N^{j-1})$ and perturb $Q$ by defining 
  \begin{equation}
    \tilde{Q}_j(X,\epsilon) = Q_j(X) + \frac{\epsilon}{2} X^\sfT J(j) X.
  \end{equation}
  Then, there is some constant $\epsilon'$ such that for $0 < \epsilon < \epsilon'$, the matrices $\tilde{H}_j(\epsilon)$ obtained by deleting the first row of $H_j + \epsilon J(j)$ will be in $N-k+1$ general position. In fact, for all $0 \ne y\in \RR^k$ the matrix
  \begin{equation}
    A(y,\epsilon) = \sum_{j = 1} Y_j(H_j + \epsilon J(j))
  \end{equation}
  will have $\rk(A) \ge N - k + 1$.
\end{lemma}

Next, Grigoriev and Pasechnik deal with the assumption 2 by showing that any $\zeta \ne 0$ that is sufficiently close to zero is a regular value of $p(Q(X))$ and $p(Y)$.
\begin{lemma}[{\cite[Lemma~5.3]{GP05}}]
  \label{lem:regularvalue}
  Let $p\colon \RR^k \to \RR$ and $Q\colon \RR^n \to \RR^k$ be polynomials. Then $p(Y)$ and $p(Q(X))$ each have at most finitely many critical values.
\end{lemma}

We have now established that $p(\tilQ(X, \epsilon)) - \zeta$ satisfies all assumptions required for \cref{thm:4.5} when $\epsilon, \zeta$ are sufficiently small. In order to use this to find solutions to our original equation we need to relate $Z = Z(p(Q(X)))$ to $\tilZ(\zeta, \epsilon) = Z(p(\tilQ(X, \epsilon)) - \zeta)$. Clearly, if $\zeta$ and $\epsilon$ are very small and $X^*$ is a zero of $p(\tilQ(X, \epsilon)) - \zeta$ then $p(Q(X^*))$ will be close to zero. This could suffice for finding approximate solutions, but to solve the problem exactly we need something more. Grigoriev and Pasechnik provide this by proving that $Z$ coincides exactly with the limit of $\tilZ(\zeta, \epsilon)$ as $\zeta, \epsilon \downarrow 0$. They consider the solutions of $p(\tilQ(X, \epsilon)) - \zeta$ as Puiseux series in $\zeta$ and $\epsilon$ and prove that the limits of these are exactly the zeros of $p(Q(X))$. For us, however, it will be more convenient to consider $\tilZ(\zeta, \epsilon)$ as sets depending on $\zeta$ and $\epsilon$ and take the limit of these sets as $\zeta, \epsilon \downarrow 0$.
In the rest of the paper we use the following notion of limits of sets. Strictly speaking, our definition is that of the Kuratowski limit inferior, which of course coincides with the Kuratowski limit if this exists.
\begin{definition}
  For $\epsilon \in (0,1)$ let $S_\epsilon \subseteq \RR^n$. Then we define 
  \begin{equation}
    \lim_{\epsilon \downarrow 0} S_\epsilon = \{X \in \RR^n \colon \forall r \in \RR_{>0} \exists \epsilon_0 \forall \epsilon < \epsilon_0 \exists Y \in S_\epsilon \text{ s.t. } \norm*{X - Y} < r\}.
  \end{equation} 
\end{definition}

\begin{theorem}[{\cite[Theorem 5.4]{GP05}}]
  \label{thm:limitsofzeros}
  For fixed $\epsilon, \zeta$, let $\tilZ(\epsilon,\zeta)$ be the zeros of $p(\tilde{Q}(X,\epsilon)) - \zeta$ and let $Z = Z(p(Q(X)))$ be the zeros of $p(Q(X))$. Then we have
  \begin{equation}
    Z = \lim_{\zeta \downarrow 0} \lim_{\epsilon \downarrow 0} \tilZ(\epsilon, \zeta).
  \end{equation} 
\end{theorem}
The proof for our notion of limit is almost the same as Grigoriev and Pasechnik's proof, but is included for completeness.
\begin{proof}
We will prove inclusions in both directions separately. First, let us focus on showing the $\supseteq$ direction. Let $\tilde{X}$ be a limit point of $\tilZ(\epsilon, \zeta)$ and assume, towards a contradiction, that $p(Q(\tilde{X})) \ne 0$. Note that $p(Q(\tilde{X}))$ is just $p(\tilde{Q}(\tilde{X}, \epsilon)) - \zeta$ where $\epsilon = \zeta = 0$. Since the map $(\epsilon, \zeta) \mapsto p(\tilde{Q}(\tilde{X}, \epsilon)) - \zeta$ is clearly continuous and nonzero for $(\epsilon,\zeta) = (0,0)$ (by assumption), it must be nonzero when $\epsilon, \zeta$ are sufficiently close to $0$. It follows that $\tilde{X}$ is not in $\lim_{\zeta \downarrow 0} \lim_{\epsilon \downarrow 0} \tilZ(\epsilon, \zeta)$ achieving a contradiction as desired.

We will now prove $Z \subseteq \lim_{\zeta \downarrow 0} \lim_{\epsilon \downarrow 0} \tilZ(\epsilon, \zeta)$ by showing that for all $X^* \in Z$, there exist points $\tilde{X}(\epsilon,\zeta) \in \tilZ(\epsilon, \zeta)$ such that 
\begin{equation}
  \lim_{\zeta \downarrow 0} \lim_{\epsilon \downarrow 0} \tilde{X}(\epsilon, \zeta) = X^*.
\end{equation}
We start by noting that for any $X^*\in Z$, the polynomial $\epsilon \mapsto p(\tilde{Q}(X^*, \epsilon))$ is $0$ for $\epsilon = 0$ (as $X^* \in Z$) and hence has no constant part. Since we let $\epsilon$ go to $0$ before $\zeta$, this implies that $p(\tilde{Q}(X^*, \epsilon)) - \zeta < 0$ as $\epsilon$ and $\zeta$ go to $0$. Next, we assume without loss of generality that $p(Y) \ge 0$. This can be ensured by replacing $p$ with $p^2$, which leaves the zeros unchanged. Since $p$ is not the all-zero polynomial, every open neighbourhood around $X^*$ must contain some point $Y$ with $p(Q(Y)) > 0$, that is, $X^*$ lies in the closure of the semialgebraic set $\{X \colon p(Q(X)) > 0\}$. We now invoke the Curve Selection Lemma (\cite[Thm~3.19]{BPR06}) to obtain a continuous semialgebraic map $\gamma \colon [0,1) \to \RR^n$ such that $\gamma(0) = X^*$ and $\gamma((0,1)) \subseteq \{X \colon p(Q(X)) > 0\}$.

Consider the semialgebraic\footnote{  Polynomials are semialgebraic, $\gamma$ is semialgebraic, and compositions of semialgebraic maps are semialgebraic.} and continuous map $f\colon(\tau, \epsilon,\zeta) \mapsto p(\tilde{Q}(\gamma(\tau), \epsilon)) - \zeta$. Note that $f(0, \epsilon, \zeta) < 0$ as $\epsilon, \zeta$ go to $0$ and that $f(\tau, 0,0) > 0$ if $\tau \in (0,1)$. Define the map
\begin{equation}
  \beta \colon (t, \epsilon, \zeta) \mapsto t\begin{pmatrix}1 \\ 0 \\ 0 \end{pmatrix} + (\zeta - t)\begin{pmatrix}0 \\ \epsilon/\zeta \\ 1 \end{pmatrix}.
\end{equation}
Note that $\beta$ is continuous on $(0,1)^3$ and that $\lim_{\zeta \downarrow 0}\lim_{\epsilon \downarrow 0} \beta(t, \epsilon, \zeta)$ exists since $\epsilon$ goes to $0$ before $\zeta$.
We have $\beta(0, \epsilon, \zeta) = (0, \epsilon, \zeta)$ and $\beta(\zeta, \epsilon, \zeta) = (\zeta, 0, 0)$. Then, for fixed $\epsilon, \zeta$, the map $t \mapsto f(\beta(t, \epsilon, \zeta))$ is $<0$ for $t = 0$ and $> 0$ for $t = \zeta$. Hence, by the Intermediate Value Theorem, there exists some value $0 < T_{\epsilon, \zeta} < \zeta$ such that $f(\beta(T_{\epsilon,\zeta}, \epsilon, \zeta)) = 0$. Define $\tilde{X}(\epsilon, \zeta) = \gamma(T_{\epsilon, \zeta})$. Note that when $\epsilon, \zeta$ are sufficiently small, $p(\tilde{Q}(\tilde{X}(\epsilon,\zeta), \epsilon)) - \zeta = 0$ by construction. Furthermore, since $ 0 < T_{\epsilon, \zeta} < \zeta$ it goes to $0$ as $\zeta$ goes to $0$. Therefore, we have 
\begin{equation}
  \lim_{\zeta\downarrow 0} \lim_{\epsilon\downarrow 0} \tilde{X}(\epsilon, \zeta) = \lim_{\zeta\downarrow 0} \lim_{\epsilon\downarrow 0} \gamma(T_{\epsilon, \zeta}) = X^*,
\end{equation}
which completes the proof.
\end{proof}

To test whether $Z$ is empty, it suffices to check whether there is a sequence of nonempty $\tilde Z$.

\begin{lemma}\label{lem:eps-sequence-hack}
  $Z \ne \emptyset$ iff there exists a sequence $\zeta_n,\epsilon_n \to 0$ such that $\tilde Z(\epsilon_n,\zeta_n)\ne\emptyset$.
\end{lemma}
\begin{proof}
  If $Z\ne \emptyset$, then by \cref{thm:limitsofzeros}, for every sufficiently small $\zeta>0$, there exists an $\epsilon_\zeta$, such that $\tilde Z(\epsilon,\zeta)\ne 0$ for all $\epsilon\in(0,\zeta_{\epsilon}]$.
  Thus, we can easily construct the sequence $\zeta_n = O(1/n),\epsilon_n = \epsilon_{\zeta_n}$ with $\tilde Z(\epsilon_n,\zeta_n)\ne\emptyset$.

  Now given a sequence $\zeta_n,\epsilon_n$ with $\tilde Z(\epsilon_n,\zeta_n)\ne \emptyset$, we can select $X_{n}\in\tilde Z(\epsilon_n,\zeta_n)$, and since $\tilde Z(\epsilon,\zeta)$ is bounded (for sufficiently small $\epsilon, \zeta$), we can apply the Bolzano-Weierstrass theorem to pick a subsequence $X_n\to X^*$.
  We have $p(\tilde Q(X_n,\epsilon_n)) = \zeta_n$.
  Due to continuity it follows $\lim_n p(\tilde Q(X_n,\epsilon_n)) = p(Q(X^*)) = 0$ and $X^*\in Z$.
\end{proof}

\subsection{Solving the smaller systems using algorithms for the first-order theory of the reals}

We now have a way to determine if $p(Q(X))$ has a zero: first, we consider the perturbed system $p(\tilQ(X,\epsilon)) - \zeta$. By \cref{lem:eps-sequence-hack} it suffices to check if there is a sequence of nonempty $\tilde{Z}$. For sufficiently small $\zeta, \epsilon$ we can use \cref{thm:4.5} with \cref{remark:constants} (setting $C \coloneq (\zeta,\epsilon)$) to find sets $V_c(U,W, \zeta, \epsilon)$ such that $\bigcup_{U,W} V_c(U,W, \zeta, \epsilon)$ intersects all connected components of $\tilZ(\zeta, \epsilon)$. In particular, a sequence of nonempty $\tilde{Z}$ exists iff there exists a sequence of nonempty $\phi_W(V_c(U,W, \zeta, \epsilon), \epsilon)$. Here 
\begin{equation}
  \phi_W\colon (X, \epsilon) \mapsto \begin{pmatrix}\tilde{Q}(X, \epsilon) \\ X_{\overline{W}} \end{pmatrix}.
\end{equation}
The sets $\phi_W(V_c(U,W, \zeta, \epsilon), \epsilon)$ will be defined as the solutions to
\begin{subequations}\label{eq:phiW(Vc)eps}
  \begin{align}
    p(Y) &= \zeta,\label{eq:p(Y)=delta}\\
    \Omega \coloneq \det\Phi(Y, \epsilon)_{UW} &\ne 0,\label{eq:Omega_ne_0_eps} \\
    \Omega^2 Y &= \Omega^2 \tilQ(\phi_{UW}^{-1}(Y,T, \epsilon), \epsilon),\label{eq:Y=Q(phi_eps)}\\
    \Omega b(Y)_{\Ubar} &= \Omega\Phi(Y, \epsilon)_{\Ubar W}\Phi(Y, \epsilon)^{-1}_{UW}\cdot b(Y)_U,\label{eq:b(Y)_eps}\\
    \det\Phi(Y, \epsilon)_{U'W'} &= 0\quad \forall U',W': \abs{U'} = \abs{W'} = \abs{U} +1, U\subset U' \subset\{2,\dots,n\}, W\subset W'\subset[n],\label{eq:det=0_eps}
  \end{align}
\end{subequations}
where we have written $\tilQ_j(X, \epsilon) = \frac{1}{2}X^\sfT H_j(\epsilon) X + b_j^\sfT X + c_j$ (see \cref{lem:generalposition}) and $\Phi(Y, \epsilon) = \sum_{j = 1}^k p_j(Y)\hat{H}_j(\epsilon)$ for $\hat{H}_j(\epsilon)$ the matrix $H_j(\epsilon)$ with the first row deleted.

We will write the existence of such a sequence of $\phi_W(V_c(U,W, \zeta, \epsilon), \epsilon)$ as a sentence in the first-order theory of the reals.
\begin{definition}[First-order theory of the reals]
  The first-order theory of the reals is concerned with sentences of the form
  \begin{equation}
    \label{eq:FTR}
    Q_1x_1 \in \RR^{n_1} Q_2x_2 \in \RR^{n_2} \dots Q_q x_q \in \RR^{n_q} P(x_1, \dots, x_q),
  \end{equation}
  where
  \begin{itemize}
    \item the $Q_i$ are alternating quantifiers
    \item $P$ is a quantifier-free Boolean formula that can involve atomic formulas of the form \begin{equation}
      f(x_1, \dots, x_q) \Delta 0.  
    \end{equation}
          Here $f \colon \RR^{n_1} \times \dots \times \RR^{n_q} \to \RR$ is a polynomial with integer coefficients and $\Delta$ is one of the relation symbols $\le, <, =, >, \ge , \ne$.
  \end{itemize}
  We write the set of all true sentences of this form as $\FTR$.
\end{definition}
\begin{theorem}[{\cite[Theorem~1.1]{Ren92}}]
  \label{thm:FTR}
  There is an algorithm that, given a sentence $\varphi$ of the form of \cref{eq:FTR} involving only polynomials with integer coefficients\footnote{To deal with algebraic coefficients we can add additional variables and polynomials enforcing that these variables take the value of the required algebraic numbers.}, decides whether $\varphi \in \FTR$. The algorithm uses 
  \begin{equation}
    L^2(md)^{2^{O(q)} \prod_{i=1}^{q} n_i}
  \end{equation}
  parallel processors and requires 
  \begin{equation}
    \log(L) \left(2^q\left(\prod_{i=1}^{q}n_i\right)\log(md)\right)^{O(1)} + \Time(P)
  \end{equation}
  time. Here $L$ denotes the maximal bit size of the coefficients of the polynomials in $\varphi$, $m$ the number of polynomial (in)equalities in $\varphi$, $d$ the maximal degree of these polynomials, and $\Time(P)$ the worst case time required for computing $P$ when the atomic formulas are substituted for Boolean values.
\end{theorem}

Combining this result with \cref{lem:eps-sequence-hack} we get the following result.
\GPparallel*

\begin{proof}
  By \cref{lem:eps-sequence-hack}, it suffices to check whether there exists a sequence of $\epsilon,\zeta\to 0$, such that $\tilde Z(\epsilon,\zeta)\ne\emptyset$.
  By \cref{lem:generalposition}, assumption (3) is satisfied for sufficiently small $\epsilon$ (independent of $\zeta$).
  By \cref{lem:regularvalue}, for every $\epsilon$ there exist only a finite number of critical values (or ``bad'' $\zeta$).
  Thus, if there exist a suitable sequence $\epsilon,\zeta$, we may also assume that the assumptions of \cref{subsec:reduce} hold.
  By \cref{thm:4.5}, there must be a subsequence $(\epsilon_n, \zeta_n) \to 0$ and $U,W$ such that $\tilde Z(\epsilon_n,\zeta_n)\cap V_c(U,W,\epsilon_n,\zeta_n)\ne\emptyset$ for all $n$.
  Therefore $Z\ne\emptyset$ is equivalent to the following FOTR sentence being true for some $U,W$, which we can test in parallel:
  \begin{equation}\label{eq:piece-sat}
    \Xi_{U,W} \coloneq \forall \gamma>0\;\exists \epsilon,\zeta\in(0,\gamma)\;\exists Y,T\colon (Y,T) \in \phi_W\bigl(V_c(U,W,\epsilon,\zeta)\bigr),
  \end{equation}
  where $(Y,T) \in \phi_W(V_c(U,W,\zeta,\epsilon))$ is shorthand for $(Y,T)$ satisfies \cref{eq:phiW(Vc)eps}, which then implies $\phi^{-1}_{UW}(Y,T,\epsilon)\in\tilde Z(\epsilon,\zeta)$.
  We can test this sentence using \cref{thm:FTR}.
  The complexity bounds follow from those of \cref{thm:FTR,thm:4.5}.
\end{proof}

The proof of \cref{thm:PureSuperQMAinPSPACE} now follows directly from applying \cref{thm:parallelalgforGP} to the discussion in \cref{subsection:PSQMAtopolynomials}.

\section{Approximation and Optimization}\label{sec:approximation}

In the previous section, we showed how to solve the (decision) $\PureCLDM$ problem.
Next, we give a parallel polynomial time algorithm (i.e. (function) $\NC(\poly)$) to compute an approximate description of the consistent state.
It is worth emphasizing that our algorithm computes an approximation to an exactly consistent state, and not just a state that is approximately consistent.

If the density matrices are slightly perturbed, there may not be a \emph{perfectly} consistent state.
We can compute the state that minimizes the ``error'' (e.g. the minimum $\alpha$ such that the given instance is a YES-instance per \cref{def:kPureCLDM}) by solving OptGP as defined in \cref{def:OptGP}

\OptGP*

\subsection{Computing the minimum}\label{subsec:minimum}

The first step in solving the approximation problem of \cref{thm:approx-opt} is to compute $t$.
For that, we will use Renegar's algorithm for finding approximate solutions to first-order theory of the reals formulas.

\begin{theorem}[{\cite[Theorem~1.2]{Ren92-4}}]
  \label{thm:Renapprox}
  Let $\varphi(y)$ be a formula as in \cref{eq:FTR} involving free variables $y \in \RR^l$, let $0 < \epsilon < r$ be powers of $2$ and define $\sol(\varphi, r) = \{y \in \RR^l \colon \varphi(y) \wedge \|y\| \le r\}$. Then there exists an algorithm that constructs a set $\{y_i\}_i$, such that for every connected component of $\sol(\varphi, r)$, there is at least one $y_i$ within distance $\epsilon$ of the component.
  
  The algorithm can be implemented in parallel using 
  \begin{equation}
    (L + |\log \epsilon| + |\log r|)^{O(1)}(md)^{2^{O(q)}l\prod_{j = 1}^{q} n_j}
  \end{equation}
  processors. It then requires time
  \begin{equation}
    \left(2^q l \left(\prod_{j = 1}^{q} n_j\right) \log\left(mdL + |\log \epsilon| + |\log r|\right)\right)^{O(1)} + \Time(P).
  \end{equation}
  Here $m$ is the number of polynomials in $\varphi$, $d$ their maximum degree, $L$ the maximum bit length of the coefficients of these polynomials, $q$ the number of quantifiers and $\Time(P)$ the worst case parallel time to compute $P$ when the atomic formulas are substituted for Boolean values.
\end{theorem}

It is relatively straightforward to compute $\theta$ with the above theorem.
To that end, define $p_\theta(Y) \coloneq (p(Y))^2 + (r(Y) - \theta)^2$.
We will compute $\theta$ by finding the minimum $\theta$ for which $Z_\theta \coloneq Z(p_\theta(Q(X))) \ne \emptyset$.
Since $Z$ is compact, the minimum $\theta=\min_{X\in Z}r(Q(X))$ must exist.
Hence, there must exist a piece $U,W$, such that the FOTR sentence $\Xi_{U,W}(\theta)$ is true, where $\Xi_{U,W}(\theta)$ is defined as in \cref{eq:piece-sat}, but with the additional parameter $\theta$ (see \cref{remark:constants}).
Note that since $\theta$ is minimal, $\Xi_{U,W}(\theta')\wedge\theta'<\theta$ cannot hold.
Formally, $\theta$ is minimal iff the sentence following sentence is true:
\begin{equation}
  \Xi_{U,W} \wedge \forall \theta'\colon \Xi_{U,W}(\theta')\rightarrow (\theta'\ge\theta)
\end{equation}
Since $t$ is a free variable in the above sentence, we can approximate it using \cref{thm:Renapprox} for all pieces in parallel, and then take the minimum.

\subsection{Computing the solution vector}

Next, we are going to show how to compute $X^*$ such that $p(Q(X^*))=0$ and $r(Q(x^*))=\theta$ and thereby prove \cref{thm:approx-opt}.
Naively, we could use \cref{thm:Renapprox} to compute $X^*$.
However, that will not lead to the desired runtime since the dimension of $X$ is exponentially large.
We modify $\Xi_{U,W}$ from \cref{eq:piece-sat} to find ``good enough'' $(\tilY, \tilT)\in \tilZ(\epsilon,\zeta)$ with $\tilX \coloneq\phi^{-1}_{U,W}(\tilY,\tilT,\epsilon)$, such that there exists some $X\in Z$ with $\norm{X - \tilX}\le\delta$ for some desired accuracy parameter $\delta$, which can be chosen exponentially small in $N$.
The existence of $\tilY,\tilT$ follows by a similar to the proof of \cref{lem:eps-sequence-hack}.

We will use a univariate representation to compute $\tilY,\tilT$ exactly and then \cref{thm:Renapprox} to obtain approximations to the entries of $\tilX$ in parallel.%
\footnote{One could also compute $\tilX$ numerically, but then numerical errors need to be considered.}
The next theorem shows how to compute univariate representations of the solutions to an FOTR sentence. It is a straightforward combination of \cite[Proposition~2.3.1 and 3.8.1]{Ren92} and \cite[Proposition~6.2.2]{Ren92-2}.

\begin{theorem}\label{thm:Renunivariate}
  Let $\varphi(y)$ be a formula of the form \cref{eq:FTR} with free variables $y \in \RR^l$. Then, there exists a set $\cal{P}(\varphi)$ of $(md)^{2^{O(q)}l\prod_j n_j}$ pairs of polynomials $(p, F)$ where $p \colon \RR \to \RR$ and $F \colon \RR \to \RR^{l+1}$ with the following property. For every connected component $C$ of $\{y\in \RR^l \colon \varphi(y)\}$, there is a $(p, F) \in \cal{P}(\varphi)$ such that for some root $t^*$ of $p$, $\Aff(F(t^*))$ is well defined and in $C$. Here $\Aff(F(t^*))$ denotes the affine image $\frac{1}{F_{l+1}(t^*)}(F_1(t^*), \dots, F_l(t^*))$

  The sets $\cal{P}(\varphi)$ can be constructed in parallel using
  \begin{equation}
    L^2(md)^{2^{O(q)}l \prod_{j = 1}^{q} n_j}
  \end{equation}
  parallel processors and time
  \begin{equation}
    (\log L)\left(2^q l \left(\prod_{j = 1}^{q} n_j\right) \log (md)\right)^{O(1)},
  \end{equation}
  where $m$ is the number of polynomial (in)equalities appearing in $\varphi$, and $d$ is their maximum degree.
\end{theorem}
\begin{proof}
  Let $\varphi$ be of the form \cref{eq:FTR}. By \cite[Proposition~6.2.2]{Ren92-2}, we can construct a polynomial $g_\varphi = (g_{\varphi,1} \dots, g_{\varphi,k}) \colon \RR^l \to \RR^k$ such that if $\varphi(y)$ is true, and $y$ and $\hat{y}$ are in the same component of the connected sign partition\footnote{The connected sign partition of a polynomial $f\colon \RR^n \to \RR^m$ is a partition of $\RR^n$ whose elements are the maximally connected sets $C$ with the property that if $x \in C$ and $y \in C$ then $\sign(f_i(x)) = \sign(f_i(y))$ for all $i$.} $\CSP(g)$, then $\hat{y}$ also satisfies $\varphi$. Furthermore, the degree $d_g$ of $g_{\varphi}$ and the number of polynomials $k$ will both be at most $(md)^{2^{O(q)}\prod_j n_j}$.

  It follows that, if for every component $C$ of $\CSP(g_{\varphi})$, we can compute a point $y \in C$, then we have found a point in every connected component of $\{y\in \RR^l \colon \varphi(y)\}$. To find such a $y$, we can use \cite[Proposition~3.8.1]{Ren92} to find a set $\calR(g_\varphi)$ of $(kd_g)^{O(l)}$ polynomials $R\colon \RR^{l+1} \to \RR$ of degree at most $d_R = (d_g k)^{O(l)}$, such that for every component of $\CSP(g_{\varphi})$ there is an $R \in \calR(g_\varphi)$ that is nontrivial and factors linearly as 
  \begin{equation}
    R(U) = \prod_j \xi^{(j)}\cdot U,
  \end{equation}
  where for some $i$, $\Aff(\xi^{(i)})$ is well-defined and $\Aff(\xi^{(i)}) \in C$.

  By \cite[Proposition~2.3.1]{Ren92}, there exist sets $\calP(R)$ of $ld_R^3$ pairs of polynomials $(p, F)$, where $p\colon \RR \to \RR$ and $F \colon \RR \to \RR^l$ have degree at most $d_R$ and, for some pair $(p, F) \in \calP(R)$, $p$ is not identically zero and has a root $t^*$ such that $\Aff(F(t^*)) = \Aff(\xi^{(i)})$.

  Hence taking $\calP(\varphi) = \bigcup_{R \in \calR(g_\varphi)} \calP(R)$ satisfies the conditions from the theorem. Every $\calP(R)$ contains $ld_R^3 = (d_g k)^{O(l)}$ pairs and there are at most $(d_g k)^{O(l)}$ different $R \in \calR(g,\varphi)$. In total there will be at most $(d_g k)^{O(l)} = (md)^{2^{O(q)}l\prod_j n_j}$ pairs in $\calP(\varphi)$.

  To get the claimed runtime, note that $g_\varphi$ can be constructed in time $(\log L)[2^q l \prod_j n_j \log(md)]^{O(1)}$ using $L^2(md)^{2^{O(q)}l \prod_j n_j}$ parallel processors. Its coefficients will be integers of bit length at most $(L  + l)(md)^{2^{O(q)}\prod_j n_j}$. By part (ii) of \cite[Proposition~3.8.1]{Ren92}, $\calP(\varphi)$ can be computed from the coefficients of $g_\varphi$ in time $[l \log(kd_g)]^{O(1)} = [2^{q}lmd \prod_j n_j]^{O(1)}$ on $(kd_g)^{O(l)} = (md)^{2^{O(q)}l\prod_j n_j}$ processors. The coefficients of the elements of $\calP(\varphi)$ will be integers of bit length at most $(L  + l)(md)^{2^{O(q)}l\prod_j n_j}$. The runtime follows. 
\end{proof}

We are now ready to prove \cref{thm:approx-opt}.
\begin{proof}[Proof of \cref{thm:approx-opt}]
  First, use \cref{subsec:minimum} and \cref{thm:Renunivariate} to create a univariate representation of $\theta$.
  In parallel, we can try all polynomials in $\calP$ to find a rational function $F$ and polynomial $p$, such that $\theta = F(t^*)$ for some root $t^*$ of $p$.
  We can specify $t^*$ using Thom's lemma (see e.g. \cite[Proposition 2.28]{BPR06}), which states that a root of a univariate polynomial $p$ is uniquely specified by the signs of the derivatives.
  We can try all possible sign vectors efficiently in parallel using \cite[Proposition 4.1]{Ren92} to find the one that identifies $t^*$.
  Next, we shall proceed analogously\footnote{Naively, the number of sign vectors to try might seem to large as the degree of \cref{eq:phiW(Vc)eps} can be exponential. However, we can use \cite[Theorem~4.1]{Ren92}, taking the $g_i$s to be the derivatives to find a suitable set of candidate sign vectors.} to compute representations of $\tilY,\tilT,\tilepsilon,\tilzeta$, such that
  \begin{equation}
    \begin{aligned}
    &\bigl((\tilY,\tilT)\in \phi_W(V_c(U,W,\tilepsilon,\tilzeta),\tilepsilon)\bigr)\\
    &\wedge \forall \gamma>0\; \exists \epsilon,\zeta\in(0,\gamma)\;\exists Y,T\colon \bigl((Y,T)\in \phi_W(V_c(U,W,\epsilon,\zeta),\epsilon)\bigr)\\
    &\hphantom{\wedge\forall \gamma>0\; \exists \epsilon,\zeta\in(0,\gamma)\;\exists Y,T\colon}\wedge \norm*{\phi^{-1}_{UW}(\tilY,\tilT,\tilepsilon) - \phi^{-1}_{UW}(Y,T,\epsilon)}\le\delta,
    \end{aligned}
  \end{equation}
  which guarantees by Bolzano-Weierstrass that there exists an $X\in Z$ with $\abs{X-\tilX} \le \delta$.
  There is a slight problem here: the degree of \cref{eq:phiW(Vc)eps} can be $\poly(N)$ which naively means there are $2^{\poly(N)}$ sign vectors to consider in parallel which is too many. However, we can use \cite[Theorem~4.1]{Ren92} (taking the $g_i$ to be the derivatives) to compute a set of only $N^{O(k)}$ candidate sign vectors which we can all try in parallel.  
  Once we have the univariate representations, we can compute all entries of $\tilX$ in parallel using \cref{thm:Renapprox}.
  The runtime follows from \cite[Theorem~4.1]{Ren92} and \cref{thm:4.5,thm:Renapprox,thm:Renunivariate}.
\end{proof}

\section{Applications of the parallel algorithm for GP systems}\label{sec:applications}
In this section we give some applications of \cref{thm:parallelalgforGP}. We begin by reproducing and slightly improving results by Shi and Wu regarding separable Hamiltonians. Thereafter, we briefly sketch how \cref{thm:parallelalgforGP} can be used to solve two variants of $\PureCLDM$ in $\PSPACE$.
Afterwards, we show how to put a variant of $\BellQMA$ into $\PSPACE$, which also contains $\PureSuperQMA$.
We conclude with our most general application, and that is an $\NC$ algorithm for solving (non-convex) quadratically constraint quadratic programs.

\subsection{The separable Hamiltonian problem}

\begin{definition}[\cite{SW15}]
  Let $H$ be a Hermitian over $\calA_1\dotsm \calA_k$ (each of dimension $d$).
  We say $H$ is \emph{$M$-decomposable} if there exist (not necessarily Hermitian) $H_{i}^j\in \linear(\calA_j)$ for all $i\in[M]$ and $j\in [k]$, such that
  \begin{equation}
    H = \sum_{i=1}^M \bigotimes_{j=1}^k H_i^j.
  \end{equation}
\end{definition}

\begin{definition}[Separable Hamiltonian Problem (\SH)]\label{def:sh}
  Given a description of an \emph{$M$-decomposable} Hamiltonian $H$ and a threshold $\alpha$, decide if there exists a state $\rho\in\SepD(\calA_1\otimes\dotsm\otimes\calA_k)$ such that $\Tr(H\rho) \le \alpha$, where $\SepD$ denotes the set of separable density operators.
\end{definition}

\begin{remark}
  By convexity, it suffices to consider pure product states (i.e. $\ket{\psi_1}\otimes\dotsm\otimes\ket{\psi_k}\in \calA_1\otimes\dots\otimes\calA_k$) in \cref{def:sh}.
\end{remark}

\begin{theorem}\label{thm:optsep}
  There are algorithms to decide and compute an approximate solution to $\SH$ running in parallel time $\poly(\log d, Mk, \log L)$ on $L^3 (Mkd)^{O(Mk)}$ parallel processors, where $L$ is the bit size of the instance.
  Thus $\SH\in\PSPACE$ for $Mk\in\polylog(d)$.
\end{theorem}
\begin{proof}
  We represent the problem as the following GP system with $X\in\RR^{2dk + 1}$ encoding the real and complex parts of quantum states $\ket{\psi_1},\dots,\ket{\psi_k}$, as well as a ``slack variable'' $\delta$.
  \begin{subequations}
  \begin{align}
    p(Q(X)) &= \left(\sum_{i=1}^M \prod_{j=1}^k Q_{i,j}(X) +\delta^2- \alpha\right)^2 + \sum_{j=1}^k(Q_j(k)-1)^2 = 0\\
    \forall i\in [M],j\in[k]\colon Q_{i,j}(X) &= \bra{\psi_j}H_i^j\ket{\psi_j}\\
    \forall j\in[k]\colon Q_{i}(X) &= \braketc{\psi_j}
  \end{align}
  \end{subequations}
  Let $\ket{\psi_1},\dots,\ket{\psi_k},\delta$ be a satisfying solution.
  Let $\rho_j = \ketbrab{\psi_j}$ for $j=1,\dots,k$, and $\rho = \rho_1\otimes\dotsm\otimes\rho_k$.
  Then $Q_j(k) = \Tr(\rho_j) = 1$ and $\sum_{i=1}^M \prod_{j=1}^k Q_{i,j}(X) = \Tr(H\rho)\le \alpha$.
\end{proof}

\begin{remark}
  In comparison, \cite[Theorem 4]{SW15} gives a runtime of 
  \begin{equation}
    \left((k-1)^2W^2M^2\delta^{-2}\right)^{(k-1)M}\cdot \poly(d,M,k,W,\delta^{-1})
  \end{equation}
  for computing the maximum $\max_{\rho\in \SepD}\Tr(H\rho)$, where $W = \prod_{j=1}^k \max_{i\in[M]} \norm{H_i^j}$ and $\delta$ is the additive error.
  Our runtime only depends logarithmically $\delta,W$ (the dependence is implicit via description size $L$).
\end{remark}

\begin{remark}
  One could also compute the maximum in \cref{thm:optsep} exactly using \cite[Theorem 1.5]{GP05}. However, the proof of that theorem is not yet published.
\end{remark}

\begin{definition}[Separable Local Hamiltonian Problem (\SLH)]\label{def:slh}
  Given the description of an $l$-local $n$-qubit Hamiltonian $H = \sum_{i=1}^r H_i$ such that each $H_i$ acts non-trivially on at most $l$ qubits, and a threshold $\alpha$, decide if there exists a separable state in $\rho\in\Herm(\calA_1\otimes\dotsm\otimes\calA_k)$ such that $\Tr(H\rho) \le \alpha$.
\end{definition}

\begin{remark}
  The local Hamiltonian problem and its variants are usually defined with a \emph{promise gap}, i.e., we would only have to distinguish between $\Tr(H\rho)\le \alpha$ or $\Tr(H\rho)\ge \beta$.
  For $\beta-\alpha\ge n^{-O(1)}$, $\SLH\in\QMA$ \cite{CS12}.
  Note that their proof does not immediately give $\SLH\in\PSPACE$ via $\PSPACE=\QMA_{\exp}$ \cite{FL16} (i.e., $\QMA$ with inverse exponentially small promise gap) since it relies on $\CLDM\in \QMA$, which does not trivially generalize to super-polynomially small promise gap.
\end{remark}

\begin{corollary}
  We can decide and compute an approximate solution to $\SLH$ in parallel time $\poly((4nk)^l, \log L)$ on $L^3 (Mkd)^{O(k(4nk)^l)}$ processors, where $L$ is the bit size of the instance.
  Thus, $\SLH\in\PSPACE$ for constant $l$.
\end{corollary}
\begin{proof}
  Follows from \cref{thm:optsep} and the fact that an $l$-local Hamiltonian on $n$-qubits is $(4nk)^l$-decomposable \cite{SW15}.
\end{proof}

\begin{definition}[$\QMAt$ \cite{SW15}]
  A promise problem $A$ is in $\QMAt_{c,s,r}$ if there exists a poly-time uniform family of verifiers $\{V_n\}_{n\in\NN}$ acting on two proofs registers of $n_1\in\poly(n)$ qubits and an ancilla register of $n_2\in\poly(n)$ qubits, such that at most $r(n)$ gates act on multiple registers (i.e., on both proofs, or on one proof and the ancilla).
  \begin{itemize}
    \item $\forall x\in\Ayes\;\exists\ket{\psi_1}\in\CC^{2^{n_1}}\exists\ket{\psi_2}\in\CC^{2^{n_1}}\colon p(V_{\abs{x}}, \ket{\psi_1}\otimes\ket{\psi_2}) \ge c$.
    \item $\forall x\in\Ano\;\forall\ket{\psi_1}\in\CC^{2^{n_1}}\forall\ket{\psi_2}\in\CC^{2^{n_1}}\colon p(V_{\abs{x}}, \ket{\psi_1}\otimes\ket{\psi_2}) \le s$,
  \end{itemize}
  for acceptance probability $p(V_n,\ket{\psi}) = \Tr(\Piacc V_n\ketbrab{\psi,0^{n_2}}V_n)$.
\end{definition}

\begin{corollary}
  $\QMAt_{c,s,O(\log n)}\subseteq\PSPACE$ for any $c,s$.
\end{corollary}
\begin{proof}
  Follows from \cref{thm:optsep} and the fact that the POVM is $4^{r(n)}$-decomposable.
\end{proof}

For $c-s\ge n^{-O(1)}$, $\QMAt_{c,s,O(\log n)}$ was shown in \cite{SW15}.

\subsection{$\UniquePureCLDM$}
One interesting variant of $\PureCLDM$ is its unique version. Here one interested if a consistent pure state not only exists, but is unique in the sense that all orthogonal states are far from consistent. We now briefly sketch how the GP framework can be used to put this problem into $\PSPACE$.

We take two sets of variables $a_1, b_1 \dots, a_N, b_N$ and $\hat{a}_1, \hat{b}_1, \dots, \hat{a}_N, \hat{b}_N$ where the $a_i, b_i$ represent the real and complex parts of the coefficient some pure state $\ket{\psi}$ and the $\hat{a}_i, \hat{b}_i$ do the same for some other pure state $\ket{\phi}$.
We can now express $\braket{\psi}{\phi} $ using two quadratic equations
\begin{subequations}
\begin{align}
  \label{innerprodiszero}
  \Re(\braket{\psi}{\phi}) &= \sum_{i=1}^N (a_i\hat{a}_i - b_i\hat{b}_i),\\
  \Im(\braket{\psi}{\phi}) &= \sum_{i=1}^N(\hat{a}_i b_i + a_i\hat{b}_i).
\end{align}
\end{subequations}
We have already seen that we can express ``$\ket{\psi}$ is a consistent state'' as a super-verifier and hence as a GP system.
We can then express $\UniquePureCLDM$ as the optimization variant of GP systems (\cref{thm:approx-opt,subsec:minimum}).
Then we can find a pair of consistent states with minimal overlap using 
\begin{subequations}
  \begin{align}
    p(Q(X)) &= \text{``$\ket{\psi}$ and $\ket{\phi}$ satisfy all checks of the super-verifier.''}\\
    &=  \sum_{i=1}^m\left(\bigl(\pacc(V_i,\psi)-1/2\bigr)^2 + \bigl(\pacc(V_i,\phi)-1/2\bigr)^2\right),\nonumber\\
    r(Q(X)) &= \abs{\braket{\psi}{\phi}}^2,
  \end{align}
\end{subequations}
or we can find a pair of orthogonal states, such that $\ket{\psi}$ is exactly consistent, and $\ket{\phi}$ is as consistent as possible using
\begin{subequations}
  \begin{align}
    p(Q(X)) &= \sum_{i=1}^m\bigl(\pacc(V_i,\psi)-1/2\bigr)^2 + \abs{\braket{\psi}{\phi}}^2,\\
    r(Q(X)) &= \sum_{i=1}^m\bigl(\pacc(V_i,\phi)-1/2\bigr)^2.
  \end{align}
\end{subequations}

\subsection{$\SpectralPureCLDM$}
In another $\PureCLDM$ variant, which we call $\SpectralPureCLDM$, the input does not fully specify the reduced density matrices, but only their spectra. The task is to determine if there exists some pure state, such that the spectra of the reduced density matrices agree with these given spectra.

We can also use the GP machinery to solve this problem. To do so we add, for every reduced density matrix $\rho$ of which we have the spectrum given, additional variables representing its eigenbasis $\ket{\phi_{\rho,i}}$. Note that the number of variables needed to represent the eigenbases is polynomial in the input size. We then enforce
\begin{equation}
  \Tr_{\overline{C}}(\ketbrab{\psi}) = \sum_i \lambda_i \ketbrab{\phi_{\rho, i}}
\end{equation}
in the standard way. Here the $\lambda$ are the given spectrum of $\rho$. 
We then add constraints ensuring that the $\ket{\phi_{\rho,i}}$ are orthonormal similar to \cref{eq:Qnormalization,innerprodiszero}. 

\subsection{$\BellPureSymQMA$}\label{sec:BellPureSymQMA}

In order to simulate a the super-verifier of $\PureSuperQMA$, one would need $\QMA(\poly)$-verifier that receives many copies of some state $\ket{\psi}$ and runs each check $V_i$ a $\poly$-number of times to obtain a sufficiently good estimate of $\pacc(V_i, \psi)$ for all $i=1,\dots,m$.
However, we do not require the ``full power'' of $\QMA(2)$, since Bell measurements suffice, i.e., each copy is measured independently, and then these measurement results are sent to a (classical or quantum) referee that decides whether to accept.
The notion of $\QMA$ with multiple states and Bell measurements was first proposed in \cite{ABDFS08} as $\BellQMA(m)$, where $m$ denotes the number of copies.
Later, Brandão and Harrow \cite{BH17} proved that $\BellQMA(\poly) = \QMA$.
To simulate $\PureSuperQMA$, we additionally require all proofs to be pure and identical.
We call our class $\BellPureSymQMA$ as a special case of $\BellSymQMA(\poly)=\QMA$ \cite{BH17}.
Here, we add the additional requirement that all measurements have only a polynomial number of outcomes (i.e. $O(\log n)$ bits).

\begin{definition}[$\BellPureSymQMA(\poly)$]
  A promise problem $A$ is in $\BellPureSymQMA(\poly)$, if there exist polynomials $m,n_1,n_2$, constant $c$ and a poly-time uniformly generated family of quantum circuits $\{U_{x,i}\}_{x\in\{0,1\}^*,i\in m(\abs{x})}$ with a proof register of $n_1(n)$ qubits ($n=\abs{x}$), an ancilla register of $n_1(n)$ qubits, and an output register of $l \le c\cdot \log n$ qubits, as well as a family of quantum circuits $\{R_x\}$ for the referee with an input register of $c\cdot m\log{m}$ qubits and an ancilla of $n_2(n)$ qubits, such that:
  \begin{itemize}
    \item $\forall x\in \Ayes\;\forall \ket{\psi}\in\calP\colon p(V_{x},\ket{\psi}^{\otimes m})\ge 2/3$,
    \item $\forall x\in \Ano\;\forall \ket{\psi}\in\calP\colon p(V_{x},\ket\psi^{\otimes m})\le 1/3$,
  \end{itemize}
  where $\calP$ denotes the set of unit vectors on $n_1(n)$ qubits, $V_x$ is the combined circuit that takes as input a proof of $m\cdot n_1$ qubits and $(m+1)\cdot n_2$ ancillas, runs $U_{x,1}\otimes \dotsm \otimes U_{x,m}$ in parallel, measures the output register of each in the standard basis, and then runs $R_x$ on the string of outputs, and outputs the single output qubit of $R_x$.
\end{definition}

While this definition appears very technical, it is just a $\QMA$ protocol, where Merlin promises to send $\ket{\psi}^m$, and Arthur promises to measure each copy independently using a $\poly(n)$-outcome measurement, and then decides whether to accept or reject solely based on the classical measurement outcomes.

\begin{proposition}
  $\BellPureSymQMA(\poly)\subseteq\PSPACE$.
\end{proposition}
\begin{proof}[Proof sketch]
  We reduce $\BellPureSymQMA(\poly)$ to a GP system, which we can solve using \cref{thm:parallelalgforGP}.
  Let $t$ be an upper bound on the number of outcomes for each measurement.
  Let $x$ be the input.
  We construct a GP system to check whether $x\in \Ayes$.
  The variables $X$ shall again encode the real and imaginary parts of the state $\ket{\psi}$ as described in \cref{subsection:PSQMAtopolynomials}.
  The quadratic terms are then as follows:
  \begin{equation}
    Q_{i,j}(X) = \Pr[V_{i,x} \text{ measures outcome $j$ on input $\ket{\psi}$}]
  \end{equation}
  Using the polynomial term, we can compute the acceptance probability of $R_x$:
  \begin{equation}
    p(Q(X)) = \sum_{y\in [t]^m} \prod_{i=1}^m Q_{i,y_{i}}(X)\cdot \Pr[R_x \text{ accepts input } y]
  \end{equation}
  This polynomial is easy to prepare in $\NC(\poly)$.
\end{proof}

\subsection{Quadratic programming with few constraints}

\begin{corollary}
  Let $p,r,Q$ be as in \cref{thm:approx-opt}.
  If $L,\abs{\log\delta}\in N^{O(1)}$ and $k\in O(1)$,
  then we can compute a $\delta$-approximation of $X^*$ on $N^{O(1)}$ parallel processors in time $(\log N)^{O(1)}$, i.e., we can solve OptGP in (function) $\NC$.
\end{corollary}

Hence, we can solve \emph{quadratically constrained quadratic programs} (QCQP) with $m=O(1)$ quadratic constraints in $\NC$:
\begin{subequations}
  \begin{alignat}{2}
    &\min_{x\in\RR^n}\quad && x^\sfT A_0 x + a_0^{\sfT}x\\
    &\text{s.t.}&& x^\sfT A_i x + a_i^{\sfT}x \le b_i\quad\forall i=1,\dots,m
  \end{alignat}
\end{subequations}
Note that we have no requirements on the $A_i$, such as positivity, symmetry, or condition number.
We can compute a minimizer $x\in\RR^n$ up to $1/\exp(n)$ precision using a $\polylog(n)$-depth circuit.
We are not aware of much literature regarding QCQP with a constant number of constraints.
\cite{WK20} do mention that they can do feasibility testing of few quadratic systems under certain assumptions, partially recovering \cite{Bar93}.
A polynomial time algorithm for QCPQ's with few constraints is also implied \cite{GP05}.

The special case of \emph{sphere constraint quadratic optimization} (SQO) was already shown to be in $\NC$ by Ye in 1992 \cite{Ye92}.
\begin{subequations}
  \begin{alignat}{2}
    &\min_{x}\quad && x^\sfT Q x + c^{\sfT}x\\
    &\text{s.t.}&& x\in S = \{x\in\RR^n:\norm{x}^2 \le 1\}
  \end{alignat}
\end{subequations}
This problem appears in the \emph{trust region method} for nonlinear programming (see e.g. \cite{Conn2000}).
We can also optimize a quadratic function with cubic regularization parameter, for which a numerical method with guaranteed convergence is presented in \cite{CD20}:
\begin{equation*}
  \min_{x}\quad \frac12x^\sfT Q x + c^{\sfT}x + \frac{\rho}3\norm{x}^3,
\end{equation*}
where recall $\norm{x}^2 = x^\sfT x$ is quadratic, and we can obtain $y^2=\norm{x}$ by adding a summand $(y^4-\norm{x}^4)$ to $p(Q(X))$.

\section{Acknowledgements}
The authors would like to thank Dima Grigoryev for suggesting that the GP algorithm could be parallelized, pointing to helpful references, and giving us advice on how this could be proven.
We would also like to thank Georgios Karaiskos and Sevag Gharibian for helpful discussions and remarks.

\appendix

\section{Postponed Proofs}\label{sec:omit}

\begin{lemma}\label{lem:max-invertible-submat}
  Let $M = \psmallmat{A&B\\C&D} \in \RR^{n\times m}$ with invertible $A\in\RR^{r\times r}$ and $\rk(M)>r$.
  Then there exists an invertible submatrix of $M_{UW}$ with rows $U = [r] \cup \{i\}$ and columns $W = [r] \cup \{j\}$ for $i,j > r$.
\end{lemma}
\begin{proof}
  There exists a column $j$ of $\psmallmat{B\\D}$ that is not in the span of columns $\psmallmat{A\\C}$.
  Then $\psmallmat{A&C_j\\B&D_j}$ has rank $r+1$.
  Thus, there exists a row $i$ of $\psmallmat{B&D_j}$ that is not in the span of rows $\psmallmat{A&C_j}$.
  Hence, $\psmallmat{A&C_j\\B_i&D_{ij}}$ is invertible.
\end{proof}

\begin{lemma}\label{lem:partial-trace-chaining}
  Let $\rho\in\H_A\otimes\H_B\otimes\H_C$ be a quantum state on registers $A,B,C$.
  Then $\Tr_{BC}(\rho) = \Tr_B(\Tr_C(\rho))$.
\end{lemma}
\begin{proof}
  We can write $\rho = \sum_{abca'b'c'}x_{abca'b'c'}\ketbra{abc}{a'b'c'}$.
  Then 
  \begin{equation}
    \Tr_B(\Tr_C(\rho)) = \Tr_B\left(\sum_{abca'b'}x_{abca'b'c}\ketbra{ab}{a'b'}\right) = \sum_{aa'bc} x_{abca'bc} \ketbra{a}{a'} = \Tr_{BC}(\rho).
  \end{equation}
\end{proof}

\begin{lemma}\label{lem:partial-trace-application}
  Let $\rho\in\H_A\otimes\H_B$ be a quantum state on registers $A,B$ and $U_A,U_B$ unitaries on $\H_A,\H_B$.
  Then $\Tr_{B}((U_A\otimes U_B)\rho(U_A\otimes U_B)^\dagger) = U_A\rho_A U_A^\dagger$, where $\rho_A = \Tr_B(\rho)$.
\end{lemma}
\begin{proof}
  We can write $\rho = \sum_{aba'b'}x_{aba'b'}\ketbra{ab}{a'b'}$.
  \begin{subequations}
    \begin{align}
      \Tr_{B}((U_A\otimes U_B)\rho(U_A\otimes U_B)^\dagger) &= \sum_{aba'b'}x_{aba'b'}\Tr_B((U_A\otimes U_B)\ketbra{ab}{a'b'}(U_A\otimes U_B)^\dagger)\\
      &=\sum_{aba'b'}x_{aba'b'}\Tr_B(U_A\ketbra{a}{a'}U_A^\dagger\otimes U_B\ketbra{b}{b'}U_B^\dagger)\\
      &=\sum_{aba'b'}x_{aba'b'}U_A\ketbra{a}{a'}U_A^\dagger\cdot \underbrace{\Tr(U_B\ketbra{b}{b'}U_B^\dagger)}_{\braket{b}{b'}}\\
      &=\sum_{aba'}x_{aba'b}U_A\ketbra{a}{a'}U_A^\dagger = U_A\rho_A U_A^\dagger.
    \end{align}
  \end{subequations}
\end{proof}

\claimPsiotp*
\begin{proof}
  Recall $\ket{\psiotp} = \frac12 \sum_{a,b\in\{0,1\}} (X^aZ^b)^{\otimes n_1}\ket{\psi}\ket{abab}$.
  For $j\in[n_1]$, we have by \cref{lem:partial-trace-application}
  \begin{equation}
    \Tr_{\overline{j}}(\psiotp)=\frac14\sum_{a,b\in\{0,1\}}X^aZ^b\Tr_{\overline{j}}(\psi)Z^bX^a = I/2,
  \end{equation}
  where the latter equality follows from the fact that effectively a random Pauli gate (note $XZ=\iu Y$) is applied to $\Tr_{\overline{j}}$ see (\cite[Eq. 8.101]{NC10}).

  For $j>n_1$ the claim follows from \cref{lem:partial-trace-chaining} and
  \begin{equation}
    \begin{aligned}
      \Tr_{\overline{\{n_1+1,n_1+2\}}}(\psiotp)&=\Tr_{\overline{\{n_1+3,n_1+4\}}}(\psiotp)\\
      &=\frac14\sum_{a,b,a',b'\in\{0,1\}}\Tr\left((X^aZ^b)^{\otimes n_1}\ket{\psi}\ket{ab}\bra{\psi}(Z^{b'}X^{a'})^{\otimes n_1}\bra{a'b'}\right)\cdot \ketbra{ab}{a'b'}\\
      &=\frac14\sum_{a,b\in\{0,1\}}\ketbra{ab}{ab} = \frac I4.
    \end{aligned}
  \end{equation}
\end{proof}

\claimUenc*
\begin{proof}
  Recall for the Steane code
  \begin{align}
    \Enc(\ket{0}) &= \frac1{\sqrt8}\bigl(\ket{0000000} + \ket{1010101} + \ket{0110011} + \ket{1100110}\nonumber\\
    &\qquad\,+ \ket{0001111} + \ket{1011010} + \ket{0111100} + \ket{1101001}\bigr),\\
    \Enc(\ket{1}) &= \Xgate^{\otimes 7}\Enc(\ket1).
  \end{align}
  Observe that $\Enc(\ket{0})$ can be constructed by applying (multi-) controlled $\Xgate$-gates to the first $4$ qubits of $\ket{0}^{\otimes4}\ket{+}^{\otimes 3}$.
  Also note that the label on the last $5$ bits is sufficient to identify the encoded bit.
  Let $\Uenc'$ be the unitary $\Uenc$ for the simple Steane-code, which can be implemented as follows:
  \begin{enumerate}
    \item Assume input $\ket{b000000}$ for $b\in\{0,1\}$.
    \item Apply $\Hgate$ to the last three qubits to obtain 
    \begin{equation*}
      \ket{b000000} + \ket{b000101} + \ket{b000011} + \ket{b000110} + \ket{b000111} + \ket{b000010} + \ket{b000100} + \ket{b000001}
    \end{equation*}
    \item Apply triply controlled $\Xgate$-gates (target qubits $2,3,4$, controls $5,6,7$) to obtain
    \begin{equation*}
      \ket{b000000} + \ket{b010101} + \ket{b110011} + \ket{b100110} + \ket{b001111} + \ket{b011010} + \ket{b111100} + \ket{b101001}
    \end{equation*}
    \item Apply $\CNOT$ from the first qubit to the others.
    \item Apply controlled $\Xgate$-gates with the last $5$ qubits as controls to obtain $\Enc(\ket{b})$.
  \end{enumerate}
  $\Uenc'$ only uses $\Hgate,\Xgate$, as well as multi-controlled $\Xgate$-gates, which can be implemented exactly with $O(1)$ Clifford and $\Tgate$ gates (without additional ancillas) \cite{NC10}.
  For the $k$-fold concatenated Steane code, we need to apply $\Uenc'$ in total $1+7+7^2+\dots+7^{k-1} = O(N)$ times.
\end{proof}

\clearpage

\printbibliography 

\end{document}